%% file: main.tex
\documentclass[envcountsame]{llncs}

\usepackage{enumerate}
\usepackage{wrapfig}
\usepackage{amsopn}
\usepackage{amsfonts}
\usepackage{complexity}
\usepackage{mathtools}
\usepackage{tabularx}
\usepackage{environ}
\usepackage{hyperref}
\usepackage{xcolor,colortbl}
\usepackage{algorithm}
\usepackage{algorithmic}
\usepackage{caption}
\usepackage{tikz}
\usetikzlibrary{automata,fit}

\input{content/macros}

\begin{document}

\title{Fine-Grained Complexity of Safety Verification}

\author{Peter Chini \and Roland Meyer \and Prakash Saivasan}

\institute{TU Braunschweig, Germany \\ \email{ \{p.chini, roland.meyer, p.saivasan\}@tu-bs.de}}

\authorrunning{P. Chini, R. Meyer, and P. Saivasan}

\maketitle

\input{content/abstract}

\input{content/introduction}

\input{content/related}

\input{content/prelim}

\input{content/LCR}

\input{content/BSR}

\input{content/conclusion}

\bibliographystyle{plain}
\bibliography{content/cite}

\newpage
\appendix
\input{content/appendix}

\end{document}

%% file: content/macros.tex
\DeclareMathOperator{\Op}{\mathit{Op}}
\DeclareMathOperator{\op}{\mathit{op}}
\DeclareMathOperator{\Write}{\mathit{W}}
\DeclareMathOperator{\Read}{\mathit{R}}
\newcommand{\Ops}[1]{\Op(#1)}
\newcommand{\Writes}[1]{\Write(#1)}
\newcommand{\Reads}[1]{\Read(#1)}
\newcommand{\initmem}{a^0}
\newcommand{\asms}{\mathcal{A}}

\newcommand{\barS} {\bar{\Domain}}
\newcommand{\fw}{\mathit{fw}}
\newcommand{\Domain}{\mathit{D}}
\newcommand{\Domainbot}{\mathit{D}_{\bot}}
\newcommand{\pc}{\mathit{pc}}
\newcommand{\pcof}[1]{\pc(#1)}

\newcommand{\Pos}[1]{\texttt{Pos}(#1)}

\newcommand{\moves}[1]{\xrightarrow{#1}}
\newcommand{\zigmoves}[1]{\stackrel{#1}{\rightsquigarrow}}
\newcommand{\expr}{\mathcal{E}}

\newcommand{\seqmap}{\mu}
\newcommand{\funnyp}{\mathcal{S}}
\newcommand{\fingraph}{\mathcal{G}}
\newcommand{\initnode}{v^0}
\newcommand{\countC}{\#_C}

\newcommand{\smscc}{\text{\tt scc}}
\newcommand{\sccg}{G}
\newcommand{\sccexpr}{\mathcal{C} }
\newcommand{\validw}{\mathcal{V}}
\newcommand{\enc} {\text{\tt Enc}}

\newcommand{\Var}{\mathit{Var}}

\newcommand{\sizeP}{\texttt{P}}
\newcommand{\sizeL}{\texttt{L}}
\newcommand{\sizeM}{\texttt{D}}
\newcommand{\sizeC}{\texttt{C}}
\newcommand{\scccnt}{\texttt{s}}
\newcommand{\sld}{\texttt{d}}
\newcommand{\nrt}{\texttt{t}}
\newcommand{\nrs}{\texttt{s}}
\newcommand{\nrvars}{\texttt{v}}

\newcommand{\setcon}[1]{ \lbrace #1 \rbrace }
\newcommand{\Setcon}[2]{ \lbrace #1 \mid #2 \rbrace }
\newcommand{\Naturals}{\mathbb{N}}
\newcommand{\polyrel}{\mathcal{R}}
\newcommand{\power}{\mathcal{P}}
\newcommand{\validforms}{\mathcal{F}}
\newcommand{\crosscomp}{\mathcal{C}}

\newcommand{\abs}[1]{|#1|}
\DeclareMathOperator{\row}{\texttt{row}}
\DeclareMathOperator{\col}{\texttt{col}}
\DeclareMathOperator{\pos}{\texttt{pos}}
\DeclareMathOperator{\bin}{\texttt{bin}}
\DeclareMathOperator{\Loop}{\texttt{Loop}}
\DeclareMathOperator{\scc}{\texttt{SCC}}

\newcommand{\prjauto}[2]{#1 \! \downarrow_{#2}}
\newcommand{\Wit}{\mathit{Wit}}
\newcommand{\Valid}{\mathit{Valid}}
\newcommand{\WitSCC}{\mathit{Wit_{SCC}}}
\newcommand{\ValidSCC}{\mathit{Valid_{SCC}}}

\newcommand{\bigO} {\mathcal{O}}
\newcommand{\bigOS}{\bigO^*}
\newcommand{\LCR}{\ComplexityFont{LCR}}
\newcommand{\BSR}{\ComplexityFont{BSR}}
\newcommand{\ETH}{\ComplexityFont{ETH}}
\newcommand{\SETH}{\ComplexityFont{SETH}}
\newcommand{\kkClique}{\ComplexityFont{k \times k ~ Clique}}
\newcommand{\kSAT}[1]{\ComplexityFont{#1}\text{-}\ComplexityFont{SAT}}
\newcommand {\kClique}{\ComplexityFont{k \text{-} Clique}}
\newcommand {\SetCov}{\ComplexityFont{Set~Cover}}

\makeatletter
\newcommand {\problemtitle}[1]{\gdef\@problemtitle{#1}}
\newcommand {\problemshort}[1]{\gdef\@problemshort{#1}}
\newcommand {\probleminput}[1]{\gdef\@probleminput{#1}}
\newcommand {\problemparameter}[1]{\gdef\@problemparameter{#1}}
\newcommand {\problemquestion}[1]{\gdef\@problemquestion{#1}}

\NewEnviron{myproblem}{
	\problemtitle{}
	\problemshort{}
	\probleminput{}
	\problemquestion{}
	
	\BODY
	\par\addvspace{.5\baselineskip}
	\noindent
	\framebox[\textwidth]{
		\begin{tabularx}{\textwidth}{@{\hspace{\parindent}} l X c}
			\multicolumn{2}{@{\hspace{\parindent}}l}{ \normalsize \emph{\@problemtitle} \@problemshort} \\
			\normalsize \textbf{Input:} & \normalsize \@probleminput \\
			\normalsize \textbf{Question:} & \normalsize \@problemquestion
		\end{tabularx}
	}
	\par\addvspace{.5\baselineskip}
}

\pgfdeclarelayer{background}
\pgfdeclarelayer{backbackground}
\pgfsetlayers{backbackground,background,main}

%% file: content/abstract.tex
\begin{abstract}
	We study the fine-grained complexity of Leader Contributor Reachability ($\LCR$) and Bounded-Stage Reachability ($\BSR$), two variants of the safety verification problem for shared memory concurrent programs.
	For both problems, the memory is a single variable over a finite data domain.
	Our contributions are new verification algorithms and lower bounds.
	The latter are based on the Exponential Time Hypothesis ($\ETH$), the problem $\SetCov$, and cross-compositions.
	
	$\LCR$ is the question whether a designated leader thread can reach an unsafe state when interacting with a certain number of equal contributor threads.
	We suggest two parameterizations:
	(1) By the size of the data domain $\sizeM$ and the size of the leader $\sizeL$, and (2) by the size of the contributors $\sizeC$.
	We present algorithms for both cases. 
	The key techniques are compact witnesses and dynamic programming.
	The algorithms run in $\bigOS((\sizeL \cdot (\sizeM+1))^{\sizeL \cdot \sizeM} \cdot \sizeM^{\sizeM})$ and $\bigOS(2^{\sizeC})$ time, showing that both parameterizations are fixed-parameter tractable.
	We complement the upper bounds by (matching) lower bounds based on $\ETH$ and $\SetCov$.
	Moreover, we prove the absence of polynomial kernels.
	
	For $\BSR$, we consider programs involving $\nrt$ different threads.
	We restrict the analysis to computations where the write permission changes $\nrs$ times between the threads.
	$\BSR$ asks whether a given configuration is reachable via such an $\nrs$-stage computation.
	When parameterized by $\sizeP$, the maximum size of a thread, and $\nrt$, the interesting observation is that the problem has a large number of difficult instances.
	Formally, we show that there is no polynomial kernel, no compression algorithm that reduces the size of the data domain $\sizeM$ or the number of stages $\nrs$ to a polynomial dependence on $\sizeP$ and $\nrt$.
	This indicates that symbolic methods may be harder to find for this problem.
\end{abstract}

%% file: content/introduction.tex
\section{Introduction}
\label{Section:Intro}
We study the fine-grained complexity of two safety verification problems~\cite{Atig14,Esparza13,Hague11} for shared memory concurrent programs. 
The motivation to reconsider these problems are recent developments in fine-grained complexity theory~\cite{Cygan2016,Lokshtanov2011,Calabro2009,Impagliazzo2001}. 
They suggest that classifications such as \NP\ or even \FPT\ are too coarse to explain the success of verification methods.
Instead, it should be possible to identify the precise influence that parameters of the input have on the verification time. 
Our contribution confirms this idea.
We give new verification algorithms for the two problems that, for the first time, can be proven optimal in the sense of fine-grained complexity theory.
To state the results, we need some background.
As we proceed, we explain the development of fine-grained complexity theory.

There is a well-known gap between the success that verification tools see in practice and the judgments about computational hardness that worst-case complexity is able to give.
The applicability of verification tools steadily increases by tuning them towards industrial instances. 
The complexity estimation is stuck with considering the input size or at best assuming certain parameters to be constant.
However, the latter approach is not very enlightening if the runtime is $n^k$, where $n$ is the input size and $k$ the parameter.

The observation of a gap between practical algorithms and complexity theory is not unique to verification but made in every field that has to solve hard computational problems.
Complexity theory has taken up the challenge to close the gap.
So-called \emph{fixed-parameter tractability} (\FPT)~\cite{Cygan2015,Downey2013} proposes to identify parameters $k$ so that the runtime is $f(k)\mathit{poly}(n)$, where $f$ is a computable function and $\mathit{poly(n)}$ denotes any polynomial dependent on $n$.
These parameters are powerful in the sense that they dominate the~complexity.

For an \FPT-result to be useful, function $f$ should only be mildly exponential, and of course $k$ should be small in the instances of interest.
Intuitively, they are what one needs to optimize. 
\emph{Fine-grained complexity} is the study of upper and lower bounds on the function. 
Indeed, the fine-grained complexity of a problem is written as $O^*(f(k))$, emphasizing $f$ and $k$ and suppressing the polynomial part.
For upper bounds, the approach is still to come up with an algorithm.

For lower bounds, fine-grained complexity has taken a new and very pragmatic perspective.
For the problem of $n$-variable \kSAT{3} the best known algorithm runs in $\bigO(2^n)$ time, and this bound has not been improved since 1970.
The idea is to take improvements on this problem as unlikely, known as the exponential-time hypothesis (\ETH)~\cite{Impagliazzo2001}. 
Formally, it asserts that there is no $2^{o(n)}$-time algorithm for \kSAT{3}.
\ETH~serves as a lower bound that is reduced to other problems~\cite{Lokshtanov2011}. 
An even stronger assumption about \SAT, called strong exponential-time hypothesis (\SETH)~\cite{Impagliazzo2001,Calabro2009}, and a similar one about \SetCov~\cite{Cygan2016} allow for lower bounds like the absence of $\bigOS((2-\delta)^n)$-time algorithms.

In this work, we contribute fine-grained complexity results for verification problems on concurrent programs.
The first problem (\LCR) is reachability for a leader thread that is interacting with an unbounded number of contributors~\cite{Hague11,Esparza13}.
We show that, assuming a parameterization by the size of the leader $\sizeL$ and the size of the data domain $\sizeM$, the problem can be solved in $\bigOS((\sizeL \cdot (\sizeM+1))^{\sizeL\cdot \sizeM} \cdot \sizeM^{\sizeM})$.
At the heart of the algorithm is a compression of computations into witnesses. 
To check reachability, our algorithm then iterates over candidates for witnesses and checks each of them for being a proper witness.  
Interestingly, we can formulate a variant of the algorithm that seems to be suited for large state spaces.

Using $\ETH$, we show that the algorithm is (almost) optimal. 
Moreover, the problem is shown to have a large number of hard instances.
Technically, there is no polynomial kernel~\cite{Bodlaender2009,Bodlaender2014}.
Experience with kernel lower bounds is still limited. 
This notion of hardness seems to indicate that symbolic methods are hard to apply here.
The lower bounds that we present share similarities with the reductions presented in \cite{Furbach2015,CantinLS03,GibbonsK97}.

If we consider the size $\sizeC$ of the contributors as a parameter, we obtain an $\bigOS(2^\sizeC)$ upper bound.
Our algorithm is based on dynamic programming.
We use the technique to solve a reachability problem on a graph that is shown to be a compressed representation for $\LCR$.
The compression is based on a saturation argument which is inspired by thread-modular reasoning \cite{FlanaganFQ02,FlanaganQ03,GotsmanBCS07,HolikMVW17}.
With the hardness assumption on $\SetCov$ we show that the algorithm is indeed optimal.
Moreover, we prove the absence of a polynomial kernel.

Parameterizations of $\LCR$ involving just a single parameter $\sizeM$ or $\sizeL$ are intractable.
We show that these problems are $\W[1]$-hard.
This proves the existence of an $\FPT$-algorithm for those parameterizations unlikely.

The second problem we study generalizes bounded context switching.
Bounded stage reachability (\BSR) asks whether a state is reachable if there is a bound $\nrs$ on the number of times the write permission is allowed to change between the threads~\cite{Atig14}.
Again, we show the new form of kernel lower bound.
The result is tricky and highlights the power of the computation model.

The results are summarized by the table below.
Main findings are highlighted in gray.
We present two new algorithms for $\LCR$.
Moreover, we suggest kernel lower bounds as hardness indicators for verification problems.
The corresponding lower bound for $\BSR$ is particularly difficult to achieve.
\begin{center}
	\begin{tabular}{ |m{1.9cm} | m{2.7cm}| m{2.7cm}| m{1.1cm} |}
		\hline
		\scriptsize{Problem} & \scriptsize{Upper Bound} & \scriptsize{Lower Bound} & \scriptsize{Kernel} \\ 
		\hline 
		\hline
		\scriptsize{$\LCR(\sizeM, \sizeL)$} &
		\cellcolor{black!15} \scriptsize{$\bigOS((\sizeL \cdot (\sizeM+1))^{\sizeL\cdot \sizeM} \cdot \sizeM^{\sizeM})$} & 
		\scriptsize{$2^{o(\sqrt{\sizeL \cdot \sizeM} \cdot \log (\sizeL \cdot \sizeM )}$} &
		\scriptsize{No poly.} \\
		\hline
		\scriptsize{$\LCR(\sizeC)$} &
		\cellcolor{black!15}\scriptsize{$\bigOS(2^{\sizeC})$} & 
		\scriptsize{$(2-\delta)^{\sizeC}$} &
		\scriptsize{No poly.} \\
		\hline
		\scriptsize{$\LCR(\sizeM) , \LCR(\sizeL)$} &
		\multicolumn{3}{c|}{\scriptsize{Intractable}} \\
		\hline
		\scriptsize{$\BSR(\sizeP, \nrt)$} & 
		\scriptsize{$\bigOS(\sizeP^{2\nrt})$} &
		\scriptsize{$2^{o(\nrt \cdot \log (\sizeP))}$} &
		\cellcolor{black!15}\scriptsize{No poly.} \\
		\hline
		\scriptsize{$\BSR(\nrs,\sizeM)$} &
		\multicolumn{3}{c|}{\scriptsize{Intractable}} \\
		\hline
	\end{tabular}
\end{center}

The paper at hand is the full version of \cite{Saivasan2018}.
It presents some new results.
This includes an improved algorithm for $\LCR$ running in $\bigOS(2^\sizeC)$ time instead of $\bigOS(4^\sizeC)$ and a new $(2-\delta)^\sizeC$ lower bound based on $\SetCov$.
Together, upper and lower bound show that the optimal algorithm for the problem has been found.
Moreover, we give proofs for the intractability of certain parameterizations of $\LCR$ and $\BSR$.
This justifies our choice of parameters.
Technical details can be found in the appendix of the paper.

%% file: content/related.tex
\subsubsection*{Related work}
\label{Section:RelatedWork}

Concurrent programs communicating through a shared memory and having a fixed number of threads have been extensively studied \cite{Durand-Gasselin15,FortinMW17,HagueL12,AtigBQ09}. 
The leader contributor reachability problem as considered in this paper was introduced as parametrized reachability in \cite{Hague11}. 
In \cite{Esparza13}, it was shown to be $\NP$-complete when only finite state programs are involved and $\PSPACE$-complete for recursive programs. 
In \cite{Kahlon08}, the parameterized pairwise reachability problem was considered and shown to be decidable. 
Parameterized reachability under a variant of round robin scheduling was proven decidable in~\cite{TorreMP10}.

The bounded stage restriction on the computations of concurrent programs as considered here was introduced in \cite{Atig14}. 
The corresponding reachability problem was shown to be $\NP$-complete when only finite state programs are involved.
The problem remains in $\NEXP$-time and $\PSPACE$-hard for a combination of counters and a single pushdown. 
The bounded stage restriction generalizes the concept of bounded context switching from \cite{QadeerR05}, which was shown to be $\NP$-complete in that paper. 
In \cite{Chini2017}, $\FPT$-algorithms for bounded context switching were obtained under various parameterization. 
In \cite{AtigBT08}, networks of pushdowns communicating through a shared memory were analyzed under topological restrictions.

There have been few efforts to obtain fixed-parameter tractable algorithms for automata and verification-related problems. 
\FPT-algorithms for automata problems have been studied in \cite{Fernau2016,Fernau2015,Wareham2000}. 
In \cite{DemriLS02}, model checking problems for synchronized executions on parallel components were considered and proven intractable. 
In \cite{EneaF16}, the notion of conflict serializability was introduced for the TSO memory model and an \FPT-algorithm for checking serializability was provided. 
The complexity of predicting atomicity violation on concurrent systems was considered in~\cite{FarzanM09}. 
The finding is that \FPT-solutions are unlikely to exist.
In \cite{Esparza2014}, the problem of checking correctness of a program along a pattern is investigated.
The authors conduct an analysis in several parameters.
The results range from $\NP$-hardness even for fixed parameters to $\FPT$-algorithms.

%% file: content/prelim.tex
\section{Preliminaries}
\label{Section:Prelim}

We introduce our model for programs, which is fairly standard~\cite{Atig14,Hague11,Esparza13}, and give the basics on fixed-parameter tractability.

\input{content/model}
\input{content/fptbasic}

%% file: content/model.tex
\subsubsection*{Programs}
\label{Section:Prelim_Programs}

A program consists of finitely many threads that access a shared memory.
The memory is modeled to hold a single value at a time. 
Formally, a \emph{(shared memory) program} is a tuple $\asms = (\Domain, \initmem, (P_i)_{i \in [1..t]})$. 
Here, $\Domain$ is the data domain of the memory and $\initmem\in\Domain$ is the initial value.  
Threads are modeled as control-flow graphs that write values to or read values from the memory. 
These operations are captured by $\Ops{\Domain} = \Setcon{!a, ?a}{ a \in \Domain}$.
We use the notation $\Writes{\Domain} = \Setcon{!a}{a \in \Domain}$ for the write operations and $\Reads{\Domain} = \Setcon{?a}{a \in \Domain}$ for the read operations. 
A thread $P_{id}$ is a non-deterministic finite automaton  $(\Ops{\Domain}, Q, q^0, \delta)$ over the alphabet of operations. 
The set of states is $Q$ with $q^0\in Q$ the initial state.
The final states will depend on the verification task. 
The transition relation is $\delta \subseteq Q \times (\Ops{\Domain} \cup \{ \varepsilon \}) \times Q$. 
We extend it to words and also write $q \xrightarrow{w} q'$ for $q' \in \delta(q,w)$. 
Whenever we need to distinguish between different threads, we add indices and write $Q_{id}$ or $\delta_{id}$.

The semantics of a program is given in terms of labeled transitions between configurations. 
A \emph{configuration} is a pair $(\pc, a) \in (Q_1 \times \dots \times Q_t) \times \Domain$.
The program counter $\pc$ is a vector that shows the current state $\pcof{i}\in Q_i$ of each thread $P_i$.
Moreover, the configuration gives the current value in memory.  
We call \mbox{$c^0 = (\pc^0, \initmem)$} with $\pc^0(i)=q^0_i$ for all $i\in [1..t]$ the \emph{initial configuration}.
Let $C$ denote the set of all configurations. 
The program's transition relation among configurations $\rightarrow\ \subseteq C\times (\Ops{\Domain}\cup\setcon{\varepsilon})\times C$ is obtained by lifting the transition relations of the threads. 
To define it, let $\pc_1 = \pc[i=q_i]$, meaning thread $P_i$ is in state $q_i$ and otherwise the program counter coincides with $\pc$. 
Let $\pc_2=\pc[i=q_i']$.
If thread $P_i$ tries to read with the transition $q_i \xrightarrow{?a} q_i'$, then 
$(\pc_1, a)\xrightarrow{?a}(\pc_2, a)$. 
Note that the memory is required to hold the desired value. 
If the thread has the transition $q_i\xrightarrow{!b} q_i'$, then $(\pc_1, a)\xrightarrow{!b}(\pc_2, b)$. 
Finally, $q_i\xrightarrow{\varepsilon}q_i'$ yields $(\pc_1, a)\xrightarrow{\varepsilon}(\pc_2, a)$.
The program's transition relation is generalized to words, $c \xrightarrow{w} c'$. 
We call such a sequence of consecutive labeled transitions a \emph{computation}.
To indicate that there is a word justifying a computation from $c$ to $c'$, we write $c \rightarrow^* c'$.
We may use an index $\xrightarrow{w}_i$ to indicate that the computation was induced by~$P_i$.
Where appropriate, we use the program as an index, $\xrightarrow{w}_\asms$.

%% file: content/fptbasic.tex
\subsubsection*{Fixed-Parameter Tractability}
\label{Section:Prelim_FPTBasic}

We wish to study the fine-grained complexity of safety verification problems for the above programs.
This means our goal is to identify parameters of these problems that satisfy two properties.
First, in practical instances they are small.
Second, assuming that these parameters are small, show that efficient verification algorithms can be obtained. 
\emph{Parametrized complexity} is a branch of complexity theory that makes precise the idea of being efficient relative to a parameter.

Fix a finite alphabet $\Sigma$.
A \emph{parameterized problem} $L$ is a subset of $\Sigma^*\times \Naturals$. 
The problem is called \emph{fixed-parameter tractable} if there is a deterministic algorithm that, given \mbox{$(x,k)\in\Sigma^*\times\Naturals$}, decides $(x,k)\in L$ in time $f(k)\cdot \abs{x}^{O(1)}$. 
We use $\FPT$ for the class of all such problems and say \emph{a problem is $\FPT$} to mean it is in that class. 
Note that $f$ is a computable function only depending on the parameter~$k$.
It is common to denote the runtime by $\bigOS(f(k))$ and suppress the polynomial part.
We will be interested in the precise dependence on the parameter, in upper and lower bounds on the function $f$.
This study is often referred to \mbox{as \emph{fine-grained complexity}.}

Lower bounds on $f$ are usually obtained from assumptions about $\SAT$.
The most famous is the \emph{Exponential Time Hypothesis} ($\ETH$).
It assumes that there is no algorithm solving $n$-variable $\kSAT{3}$ in $2^{o(n)}$ time.
Then, the reasoning is as follows:
If $f$ drops below a certain bound, $\ETH$ would fail.
Other standard assumptions for lower bounds are the \emph{Strong Exponential Time Hypothesis} ($\SETH$) and the hardness assumption of $\SetCov$.
We postpone the definition of the latter and focus on $\SETH$.
This assumption is more restrictive than $\ETH$.
It asserts that $n$-variable $\SAT$ cannot be solved in $\bigOS((2-\delta)^n)$ time for any $\delta > 0$.

While many parameterizations of $\NP$-hard problems were proven to be fixed-parameter tractable, there are problems that are unlikely to be $\FPT$.
Such problems are hard for the complexity class $\W[1]$.
For a theory of relative hardness, the appropriate notion of reduction is called \emph{parameterized reduction}.
Given parameterized problems $L, L'\subseteq\Sigma^*\times \Naturals$, we say that $L$ is \emph{reducible} to $L'$ via a \emph{parameterized reduction} if there is an algorithm that transforms an input $(x,k)$ to an input $(x',k')$ in time $g(k) \cdot \abs{x}^{O(1)}$ such that $(x,k) \in L$ if and only if $(x',k') \in L'$.
Here, $g$ is a computable function and $k'$ is computed by a function only dependent on $k$.

%% file: content/LCR.tex
\section{Leader Contributor Reachability}
\label{Section:LCR}

We consider the \emph{leader contributor reachability problem} for shared memory programs. 
The problem was introduced in \cite{Hague11} and shown to be $\NP$-complete in \cite{Esparza13} for the finite state case.\footnote{The problem is called parameterized reachability in these works. 
We renamed it to avoid confusion with parameterized complexity.}  
We contribute two new verification algorithms that target two parameterizations of the problem.
In both cases, our algorithms establish fixed-parameter tractability. 
Moreover, with matching lower bounds we prove them to be optimal even in the fine-grained sense.

An instance of the leader contributor reachability problem is given by a shared memory program of the form $\asms = (\Domain,  \initmem, (P_L, (P_i)_{i \in [1..t]}))$. 
The program has a designated \emph{leader} thread $P_L$ and several \emph{contributor} threads $P_1, \dots, P_t$.
In addition, we are given a set of unsafe states for the leader.
The task is to check whether the leader can reach an unsafe state when interacting with a number of instances of the contributors. It is worth noting that the problem can be reduced to having a single contributor. 
Let the corresponding thread $P_C$ be the union of $P_1, \dots, P_t$ 
(constructed using an initial $\varepsilon$-transition).
We base our complexity analysis on this simplified formulation of the problem.

For the definition, let $\asms = (\Domain, \initmem, (P_L,P_C))$ be a program with two threads. Let \mbox{$F_L \subseteq Q_L$} be a set of unsafe states of the leader.
For any $t \in \Naturals$, define the program \mbox{$\asms^t = (\Domain, \initmem,  (P_L, (P_C)_{i\in[1..t]} ))$} to have exactly $t$ copies of $P_C$.
Further, let $C^f$ be the set of configurations where the leader is in an unsafe state (from~$F_L$).
The problem of interest is as follows:
\begin{myproblem}
	\problemtitle{Leader Contributor Reachability}
	\problemshort{($\LCR$)}
	\probleminput{ A program $\asms = (\Domain, \initmem, (P_L,P_C))$ and a set of states $F_L \subseteq Q_L$. }
	\problemquestion{ Is there a $t \in \Naturals$ such that $c^0 \rightarrow^*_{\asms^t} c$ for some $c \in C^f$?}
\end{myproblem}

We consider two parameterizations of $\LCR$.
First, we parameterize by $\sizeM$, the size of the data domain $\Domain$, and $\sizeL$, the number of states of the leader $P_L$.
We denote the parameterization by $\LCR(\sizeM, \sizeL)$.
The second parameterization that we consider is $\LCR(\sizeC)$, a parameterization by the number of states of the contributor $P_C$.
For both, $\LCR(\sizeM,\sizeL)$ and $\LCR(\sizeC)$, we present fine-grained analyses that include $\FPT$-algorithms as well as lower bounds for runtimes and kernels.

While for $\LCR(\sizeM, \sizeL)$ we obtain an $\FPT$-algorithm, it is not likely that $\LCR(\sizeM)$ and $\LCR(\sizeL)$ admit the same. 
We prove that these problems are $\W[1]$-hard.

\input{content/LCR_Leader_Mem}
\input{content/LCR_Contributor}
\input{content/LCR_Intractability}

%% file: content/LCR_Leader_Mem.tex
\subsection{Parameterization by Memory and Leader}
\label{Section:LCR_Leader_Mem}

We give an algorithm that solves $\LCR$ in time $\bigOS( (\sizeL \cdot (\sizeM + 1))^{\sizeL \cdot \sizeM} \cdot \sizeM^{\sizeM} )$, which means $\LCR(\sizeM, \sizeL)$ is $\FPT$.
We then show how to modify the algorithm to solve instances of $\LCR$ as they are likely to occur in practice. 
Interestingly, the modified version of the algorithm lends itself to an efficient implementation based on off-the-shelf sequential model checkers.
We conclude with lower bounds for $\LCR(\sizeM, \sizeL)$.

\input{content/LCR_Leader_Mem_Upper}
\input{content/LCR_Leader_Mem_Practical}
\input{content/LCR_Leader_Mem_Lower}
\input{content/LCR_Leader_Mem_Kernel}

%% file: content/LCR_Leader_Mem_Upper.tex
\subsubsection*{Upper Bound}
\label{Section:LCR_Leader_Mem_Upper}

We give an algorithm for the parameterization $\LCR(\sizeM, \sizeL)$. 
The key idea is to compactly represent computations that may be present in an instance of the given program. 
To this end, we introduce a domain of so-called witness candidates.
The main technical result, Lemma~\ref{Lemma:LCR_Leader_Mem_Upper_Correctness}, links computations and witness candidates.
It shows that reachability of an unsafe state holds in an instance of the program if and only if there is a witness candidate that is valid (in a precise sense). 
With this, our algorithm iterates over all witness candidates and checks each of them for being valid.
To state the overall result, let $\Wit(\sizeL, \sizeM \,) = \left( \sizeL \cdot (\sizeM + 1) \right)^{\sizeL \cdot \sizeM} \cdot \sizeM^{\sizeM} \cdot \sizeL$ be the number of witness candidates and let $\Valid(\sizeL, \sizeM, \sizeC \,) = \sizeL^3 \cdot \sizeM^{\, 2} \cdot \sizeC^2$ be the time it takes to check validity of a candidate. Note that it is polynomial.

\begin{theorem}\label{Theorem:LCR_Leader_Mem_Upper}
	$\LCR$ can be solved in time $\bigO( \Wit(\sizeL, \sizeM \,) \cdot \Valid(\sizeL, \sizeM, \sizeC \,) )$.
\end{theorem}

Let $\asms = (\Domain, \initmem, (P_L, P_C))$ be the program of interest and $F_L$ be the set of unsafe states in the leader.
Assume we are given a computation $\rho$ showing that $P_L$ can reach a state in $F_L$ when interacting with a number of contributors.
We explain the main ideas to find an efficient representation for $\rho$ that still allows for the reconstruction of a similar computation.
To simplify the presentation, we assume the leader never writes $!a$ and immediately reads $?a$ (same value). 
If this is the case, the read can be replaced by $\varepsilon$. 

In a first step, we delete most of the moves in $\rho$ that were carried out by the contributors. 
We only keep \emph{first writes}. 
For each value $a$, this is the write transitions $\fw(a) = c \xrightarrow{!a} c'$ where $a$ is written by a contributor for the first time.
The reason we can omit subsequent writes of $a$ is the following:
If $\fw(a)$ is carried out by contributor $P_1$, we can assume that there is an arbitrary number of other contributors that all mimicked the behavior of $P_1$.
This means whenever $P_1$ did a transition, they copycatted it right away.
Hence, there are arbitrarily many contributors pending to write $a$.
Phrased differently, the symbol $a$ is available for the leader whenever $P_L$ needs to read it.
The idea goes back to the \emph{Copycat Lemma} stated in \cite{Esparza13}.
The reads of the contributors are omitted as well.
We will make sure they can be served by the first writes and the moves done by $P_L$.

After the deletion, we are left with a shorter expression $\rho'$. 
We turn it into a word $w$ over the alphabet $Q_L \cup \Domainbot \cup \barS$ with $\Domainbot=\Domain\cup\setcon{\bot}$ and $\barS = \Setcon{\bar{a}}{ a \in \Domain}$.
Each transition $c \xrightarrow{!a / ?a / \varepsilon}_L c'$ in $\rho'$ that is due to the leader moving from $q$ to $q'$ is mapped (i) to $q.a.q'$ if it is a write and (ii) to $q.\bot.q'$ otherwise.
A first write $\fw(a) = c \xrightarrow{a} c'$ of a contributor is mapped to $\bar{a}$.
We may assume that the resulting word $w$ is of the form $w = w_1 . w_2$ with $w_1 \in ( ( Q_L . \Domainbot )^* . \barS )^*$ and $w_2 \in (Q_L . \Domainbot)^* . F_L$.
Note that $w$ can still be of unbounded length.

In order to find a witness of bounded length, we compress $w_1$ and $w_2$ to $w'_1$ and $w'_2$.
Between two first writes $\bar{a}$ and $\bar{b}$ in $w_1$, the leader can perform an unbounded number of transitions, represented by a word in $( Q_L . \Domainbot )^*$.
Hence, there are states $q \in Q_L$ repeating between $\bar{a}$ and $\bar{b}$.
We contract the word between the first and the last occurrence of $q$ into just a single state $q$.
This state now represents a loop on $P_L$. 
Since there are $\sizeL$ states in the leader, this bounds the number of contractions.  
Furthermore, we know that the number of first writes is bounded by $\sizeM$, each symbol can be written for the first time at most once. 
Thus, the compressed string $w'_1$ is a word in the language $( (Q_L.\Domainbot)^{\leq \sizeL} . \barS )^{\leq \sizeM}$.

The word $w_2$ is of the form $w_2 = q.u$ for a state $q \in Q_L$ and a word $u$.
We truncate the word $u$ and only keep the state $q$.
Then we know that there is a computation leading from $q$ to a state in $F_L$ where $P_L$ can potentially write any symbol but read only those symbols which occurred as a first write in $w'_1$.
Altogether, we are left with a word of bounded length.

\begin{definition}
	The set of witness candidates is $\expr = ( (Q_L.\Domainbot)^{\leq \sizeL} . \barS )^{\leq \sizeM} . Q_L$.
\end{definition}

Before we elaborate on the precise relation between witness candidates and computations, we turn to an example.
It shows how an actual computation is compressed to a witness candidate following the above steps.

\begin{example}\label{Example:ComputationToWitness}
	Consider the program $\asms = (\Domain, \initmem, (P_L,P_C))$ with domain $\Domain$, leader thread $P_L$, and contributor thread $P_C$ given in Figure \ref{Figure:CompressComputation}.
	We follow a computation in $\asms^2$ that reaches the unsafe state $q_4$ of the leader.
	Note that the transitions are labeled by $L$ and $C$, depending on whether the leader or a contributor moved.
	\begin{align*}
		(q_0,p_0,p_0,a^0) \xrightarrow{\textcolor{red}{!a}}_C (q_0,p_1,p_0,a) &\xrightarrow{?a}_L (q_1,p_1,p_0,a) \xrightarrow{!b}_L \\ 
		(q_2,p_1,p_0,b) \xrightarrow{?b}_C (q_2,p_1,p_2,b) &\xrightarrow{\textcolor{red}{!c}}_C (q_2,p_1,p_2,c) \xrightarrow{?c}_L \\ 
		(q_3,p_1,p_2,c) \xrightarrow{!a}_C (q_3,p_1,p_2,a) &\xrightarrow{?a}_L (q_4,p_1,p_2,a).
	\end{align*}
	
	We construct a witness candidate out of the computation.
	To this end, we only keep the first writes of the contributors.
	These are the write $!a$ in the first transition and the write $!c$ in the fifth transition. 
	Both are marked red.
	They will be represented in the witness candidate by the symbols $\bar{a}, \bar{c} \in \bar{\Domain}$.
	
	Now we map the transitions of the leader to words.
	Writes are preserved, reads are mapped to $\bot$.
	Then we obtain the witness candidate
	\begin{align*}
		\bar{a} \, . \, q_0 \, . \, \bot \, . \, q_1 \, . \, b \, . \, \bar{c} \, . \, q_2 .
	\end{align*}
	
	Note that we omit the last two transitions of the leader.
	The reason is as follows.
	After the first write $\bar{c}$, the leader is in state $q_2$.
	From this state, the leader can reach $q_4$ while only reading from first writes that have already appeared in the candidate, namely $a$ and $c$.
	Hence, we can truncate the witness candidate at that point and do not have to keep the remaining computation to~$q_4$.
\end{example}

\begin{figure}
	\begin{center}
		\begin{tikzpicture}
			\node[state, draw, inner sep = 2pt, minimum size = 0pt, initial, initial text = ] (q0) {$q_0$};
			\node[circle, draw, inner sep = 2pt, right of = q0, node distance = 1.25cm] (q1) {$q_1$};
			\node[circle, draw, inner sep = 2pt, right of = q1, node distance = 1.25cm] (q2) {$q_2$};
			\node[circle, draw, inner sep = 2pt, right of = q2, node distance = 1.25cm] (q3) {$q_3$};
			\node[circle, draw, inner sep = 2pt, right of = q3, node distance = 1.25cm, accepting] (q4) {$q_4$};
			
			\node[left of = q0, node distance = 0.75cm, yshift = 0.5cm] {$P_L$};
			
			\path[->, shorten >= 1pt] (q0) edge 
				node[above] {$?a$}
			(q1);
			
			\path[->, shorten >= 1pt] (q1) edge[bend left]
				node[above] {$!b$}
			(q2);
			\path[->, shorten >= 1pt] (q1) edge[bend right]
				node[below] {$\varepsilon$}
			(q2);
			
			\path[->, shorten >= 1pt] (q2) edge 
				node[above] {$?c$}
			(q3);
			
			\path[->, shorten >= 1pt] (q3) edge 
				node[above] {$?a$}
			(q4);
			
			\node[circle, draw, inner sep = 2pt, right of = q4, node distance = 2.5cm, initial, initial text = ] (p0) {$p_0$};
			\node[circle, draw, inner sep = 2pt, right of = p0, node distance = 1.5cm, yshift = 0.5cm] (p1) {$p_1$};
			\node[circle, draw, inner sep = 2pt, below of = p1] (p2) {$p_2$};
			
			\node[left of = p0, node distance = 0.75cm, yshift = 0.5cm] {$P_C$};
			
			\path[->, shorten >= 1pt] (p0) edge
				node[above] {$!a$}
			(p1);
			
			\path[->, shorten >= 1pt] (p1) edge[loop right]
				node[right] {$!a$}
			(p1);
			
			\path[->, shorten >= 1pt] (p0) edge
				node[below] {$?b$}
			(p2);
			
			\path[->, shorten >= 1pt] (p2) edge[loop right]
				node[right] {$!c$}
			(p2);
		\end{tikzpicture}
	\end{center}
	\caption{Leader thread $P_L$ (left) and contributor thread $P_C$ (right) over the data domain $\Domain = \setcon{a^0, a, b, c}$.
		The only unsafe state is given by $F_L = \setcon{q_4}$.}
	\label{Figure:CompressComputation}
\end{figure}
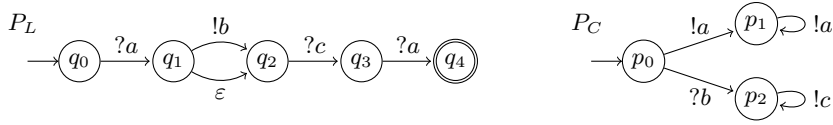

To characterize computations in terms of witness candidates, we define the notion of validity. 
This needs some notation.
Consider a word $w = w_1 \dots w_\ell$ over some alphabet~$\Gamma$.
For $i \in [1..\ell]$, we set $w[i] = w_i$ and $w[1..i] = w_1 \dots w_i$.
If $\Gamma' \subseteq \Gamma$, we use $\prjauto{w}{\Gamma'}$ for the projection of $w$ to the letters in $\Gamma'$. 

Consider a witness candidate $w \in \expr$ and let $i \in [1..\abs{w}]$.
We use $\barS(w, i)$ for the set of all first writes that occurred in $w$ up to position $i$. 
Formally, we define it to be \mbox{$\barS(w,i) = \Setcon{ a }{ \bar{a} \text{ is a letter in } \prjauto{w[1..i]}{\barS} }$}.
We abbreviate  $\barS(w, \abs{w})$ as $\barS(w)$.
Let $q \in Q_L$ and $S \subseteq \Domain$.
Recall that the state represents a loop in $P_L$. 
The set of all letters written within a loop from $q$ to $q$ when reading only symbols from $S$ is 
\mbox{$\Loop(q,S) = \Setcon{a}{ a \in \Domain \text{ and } \exists v_1, v_2 \in (\Writes{\Domain} \cup \Reads{S})^* : q \xrightarrow{v_1 !a v_2}_L q }$.}

The definition of validity is given next.
Technical details of the three requirements are made precise in the text below.
\begin{definition}\label{Definition:Validity}
	A witness candidate $w \in \expr$ is \emph{valid} if it satisfies the following properties: (1) First writes are unique. (2) The word $w$ encodes a run on $P_L$. (3) There are supportive computations on the contributors. 
\end{definition}

(1) If $\prjauto{w}{\barS} \, = \bar{c_1} \dots \bar{c_\ell}$, then the $\bar{c_i}$ are pairwise different.

(2) Let $\prjauto{w}{Q_L \cup \Domainbot} = q_1 a_1 q_2 a_2 \dots a_\ell q_{\ell+1}$. 
If $a_i \in \Domain$, then $q_i \xrightarrow{!a_i}_L q_{i+1} \in \delta_L$ is a write transition of $P_L$. 
If $a_i = \bot$, then we have an $\varepsilon$-transition \mbox{$q_i \xrightarrow{\varepsilon}_L q_{i+1}$.} 
Alternatively, there is a read $q_i \xrightarrow{?a}_L q_{i+1}$ of a symbol $a \in \barS(w, \pos(a_i))$ that already occurred within a first write (the leader does not read its own writes). 
Here, we use $\pos(a_i)$ to access the position of $a_i$ in $w$. 
State $q_1=q^0_L$ is initial.
There is a run from $q_{\ell + 1}$ to a state $q_f \in F_L$. 
During this run, reading is restricted to symbols that occurred as first writes in $w$.
Formally, there is a word \mbox{$v \in (\Writes{\Domain} \cup \Reads{\barS(w)})^*$} leading to an unsafe state $q_f$.
We have $q_{\ell+1} \xrightarrow{v}_L q_f$.

(3) For each prefix $v\bar{a}$ of $w$ with $\bar{a} \in \barS$ there is a computation $q^0_C \xrightarrow{u !a}_C q$ on $P_C$ so that the reads in $u$ can be obtained from $v$. 
Formally, let $u' = \prjauto{u}{\Reads{\Domain}}$.
Then there is an embedding of $u'$ into $v$, a monotone map $\mu : [1..\abs{u'}] \rightarrow [1..\abs{v}]$ that satisfies the following. 
Let $u'[i] =\ ?a$ with $a \in \Domain$. 
The read is served in one of the following three ways. 
We may have $v[\mu(i)] = a$, which corresponds to a write of $a$ by $P_L$. 
Alternatively, $v[\mu(i)] = q \in Q_L$ and $a \in \Loop(q,\barS(w,\mu(i)))$. 
This amounts to reading from a leader's write that was executed in a loop.
Finally, we may have $a \in \barS(w, \mu(i))$, corresponding to reading from another contributor.

Our goal is to prove that a valid witness candidate exists if and only if there is a computation leading to an unsafe state.
Before we state the corresponding lemma, we provide some intuition for the three requirements along an example.

\begin{example}
	Reconsider the program $\asms$ from Figure \ref{Figure:CompressComputation}.
	We elaborate on why the three requirements for validity are essential.
	To this end, we present three witness candidates, each violating one of the requirements. 
	Thus, these candidates cannot correspond to an actual computation of the program.
	
	The witness candidate $w_1 = \bar{a} \, . \, q_0 \, . \, \bot \, . \, q_1 \, . \, b \, . \, \bar{a} \, . \, q_2$ clearly violates requirement~(1) due to the repetition of $\bar{a}$.
	Since first writes are unique there cannot exist a computation of program $\asms$ following candidate $w_1$.
	
	Requirement (2) asks for a proper run on the leader thread $P_L$.
	Hence, the witness candidate $w_2 = \bar{a} \, . \, q_0 \, . \, a \, . \, q_1 \, . \, b \, . \, \bar{c} \, . \, q_2$ violates the requirement although it satisfies (1).
	The subword $q_0 \, . \, a \, . \, q_1$ of $w_2$ encodes that the leader should take the transition $q_0 \xrightarrow{!a}_L q_1$.
	But this transition does not exist in $P_L$.
	Consequently, there is no computation of $\asms$ which corresponds to the witness candidate $w_2$.
	
	For requirement (3), consider the candidate $w_3 = \bar{a} \, . \, q_0 \, . \, \bot \, . \,q_1 \, . \, \bot \, . \, \bar{c} \, . \, q_2$.
	It clearly satisfies (1).
	Requirement (2) is also fulfilled.
	In fact, the subwords encoding transitions of the leader are $q_0 \, . \, \bot \, . \, q_1$ and $q_1 \, . \, \bot \, . \, q_2$.
	The first subword corresponds to transition $q_0 \xrightarrow{?a}_L q_1$ which can be takes since $a$ already appeared as a first write in~$w_3$.
	The second subword refers to the transition $q_1 \xrightarrow{\varepsilon}_L q_2$.
	
	To explain that $w_3$ does not satisfy requirement (3), we show that $c$ cannot be provided as a first write.
	To this end, assume that $w_3$ satisfies (3).
	Then, for the prefix $v . \bar{c}$ with $v = \bar{a} \, . \, q_0 \, . \, \bot \, . \,q_1 \, . \, \bot$, there is a computation of the form $p_0 \xrightarrow{u!c}_C p_2$.
	The reads in $u$ are either first writes in $v$ or writes provided by the leader (potentially in loops).
	Symbol $b$ is not provided as such: It is neither a first write in $v$ nor a symbol written by the leader (in a loop) along $v$.
	However, a computation $u$ leading to state $p_2$ in $P_C$ needs to read $b$ once.
	Hence, such a computation does not exist and $c$ cannot be provided as a first write.

	The witness candidate $w = \bar{a} \, . \, q_0 \, . \, \bot \, . \, q_1 \, . \, b \, . \, \bar{c} \, . \, q_2$ from Example \ref{Example:ComputationToWitness} satisfies all the requirements.
	In particular (3) is fulfilled since $b$ is written by the leader in the transition $q_1 \xrightarrow{!b} q_2$.
	Hence, in this case, $c$ can be provided as a first write.
\end{example}

\begin{lemma}\label{Lemma:LCR_Leader_Mem_Upper_Correctness}
	There is a $t \in \Naturals$ so that $c^0\rightarrow^*_{\asms^t}c$ with $c\in C^f$ if and only if there is a valid witness candidate $w \in \expr$.
\end{lemma}

Our algorithm iterates over all witness candidates $w \in \expr$ and tests whether $w$ is valid.
The number of candidates $\Wit(\sizeL, \sizeM )$ is $\left( \sizeL \cdot (\sizeM + 1) \right)^{\sizeL \cdot \sizeM} \cdot \sizeM^{\sizeM} \cdot \sizeL$.
This is due to the fact that we can force a witness candidate to have maximum length via inserting padding symbols.
Hence, the number of candidates constitutes the first factor of the complexity estimation stated in Theorem $\ref{Theorem:LCR_Leader_Mem_Upper}$.
The polynomial factor $\Valid(\sizeL, \sizeM, \sizeC)$ is due to the following lemma.

\begin{lemma}\label{Lemma:LCR_Leader_Mem_Upper_Validity}
	Validity of $w \in \expr$ can be checked in time $\bigO( \sizeL^3 \cdot \sizeM^{\, 2} \cdot \sizeC^{\, 2} \, )$. 
\end{lemma}

%% file: content/LCR_Leader_Mem_Practical.tex
\subsubsection*{Practical Algorithm}
\label{Section:LCR_Leader_Mem_Upper_Practical}

We improve the above algorithm so that it should work well on practical instances.
The idea is to factorize the leader along its \emph{strongly connected components} (SCCs), the number of which is assumed to be small in real programs. 
Technically, our improved algorithm works with \emph{valid SCC-witnesses}. 
They symbolically represent SCCs rather than loops in the leader.   
To state the complexity, we first define the \emph{straight line depth}, the number of SCCs the leader may visit during a computation.
The definition needs a graph construction.

Let $\validw \subseteq \barS^{\leq \sizeM}$ contain only words that do not repeat letters. 
Take an element \mbox{$r = \bar{c}_1 \dots \bar{c}_\ell \in \validw$} and let $i \in [0..\ell]$. By $\prjauto{P_L}{i}$ we denote the automaton obtained from $P_L$ by removing all transitions that read a value outside $\setcon{c_1, \dots, c_i}$.
Let $\scc(\prjauto{P_L}{i})$ denote the set of all SCCs in this automaton. 
We construct the directed graph $\sccg(P_L,r)$ as follows. 
The vertices are the SCCs of all $\prjauto{P_L}{i}$ where $i \in [0..\ell]$. 
There is an edge between $S, S' \in \scc(\prjauto{P_L}{i})$, if there are states $q \in S, q' \in S'$ with $q \xrightarrow{?a / !a / \varepsilon} q'$ in $\prjauto{P_L}{i}$.
If \mbox{$S \in \scc(\prjauto{P_L}{i-1})$} and $S' \in \scc(\prjauto{P_L}{i})$, we only get an edge if we can get from $S$ to $S'$ by reading $c_i$.
Note that the resulting graph is acyclic. 

The depth $d(r)$ of $P_L$ relative to $r$ is the length of the longest path in $\sccg(P_L,r)$. 
The \emph{straight line depth} is 
$\sld = \max \Setcon{d(r)}{ r \in \validw}$. 
The \emph{number of SCCs} $\scccnt$ is the size of $\scc(\prjauto{P_L}{0})$. 
With these values at hand, the number of SCC-witness candidates (the definition of which can be found in Appendix \ref{Appendix_LCR_Leader_Mem}) can be bounded by $\WitSCC(\scccnt, \sizeM, \sld) \leq (\scccnt \cdot (\sizeM + 1))^{\sld} \cdot \sizeM^{\sizeM} \cdot 2^{\sizeM + \sld}$. 
The time needed to test whether a candidate is valid is $\ValidSCC(\sizeL, \sizeM, \sizeC, \sld) = \sizeL^2 \cdot \sizeM \cdot \sizeC^2 \cdot \sld^2$.

\begin{theorem}\label{Theorem:LCR_Leader_Mem_Upper_SCC}
	$\LCR$ can be solved in time $\bigO(\WitSCC(\scccnt, \sizeM, \sld) \cdot \ValidSCC(\sizeL, \sizeM, \sizeC, \sld))$.
\end{theorem}
For this algorithm, what matters is that the leader's state space is strongly connected. 
The number of states has limited impact on the runtime.

%% file: content/LCR_Leader_Mem_Lower.tex
\subsubsection*{Lower bound}
\label{Section:LCR_Leader_Mem_Lower}

We prove that the algorithm from Theorem \ref{Theorem:LCR_Leader_Mem_Upper} is only a root-factor away from being optimal:
A $2^{o( \sqrt{\sizeL \cdot \sizeM} \cdot \log( \sizeL \cdot \sizeM ) )}$-time algorithm for $\LCR$ would contradict $\ETH$.
We achieve the lower bound by a reduction from $\kkClique$, the problem of finding a clique of size $k$ in a graph the vertices of which are elements of a $k \times k$ matrix.
Moreover, the clique has to contain one vertex from each row.
Unless $\ETH$ fails, the problem cannot be solved in time $2^{o(k \cdot \log(k))}$ \cite{Lokshtanov2011}.

Technically, we construct from an instance $(G,k)$ of $\kkClique$ an instance $(\asms = (\Domain, \initmem, (P_L, P_C)), F_L)$ of $\LCR$ such that $\sizeM = \bigO(k)$ and $\sizeL = \bigO(k)$.
Furthermore, we show that $G$ contains the desired clique of size $k$ if and only if there is a $t \in \Naturals$ such that $c^0 \rightarrow^*_{\asms^t} c$ with $c \in C^f$.
Suppose we had an algorithm for $\LCR$ running in time $2^{o( \sqrt{\sizeL \cdot \sizeM} \cdot \log( \sizeL \cdot \sizeM ) )}$.
Combined with the reduction, this would yield an algorithm for $\kkClique$ with runtime $2^{o(\sqrt{k^2} \cdot \log (k^2))} = 2^{o(k \cdot \log k)}$.
But unless the exponential time hypothesis fails, such an algorithm cannot exist.

\begin{proposition}\label{Proposition:LCR_Leader_Mem_Lower}
	$\LCR$ cannot be solved in time $2^{o( \sqrt{\sizeL \cdot \sizeM} \cdot \log( \sizeL \cdot \sizeM ) )}$ unless $\ETH$ fails.
\end{proposition}

We assume that the vertices $V$ of $G$ are given by tuples $(i,j)$ with $i,j \in [1..k]$, where $i$ denotes the row and $j$ denotes the column in the matrix.
In the reduction, we need the leader and the contributors to communicate on the vertices of $G$. 
However, we cannot store tuples $(i,j)$ in the memory as this would cause a quadratic blow-up $\sizeM = \bigO(k^2)$. 
Instead, we communicate a vertex $(i,j)$ as a string $\row(i).\col(j)$.
We distinguish between row- and column-symbols to avoid stuttering, the repeated reading of the same symbol.
With this, it cannot happen that a thread reads a row-symbol twice and takes it for a column.

The program starts its computation with each contributor choosing a vertex $(i,j)$ to store.
For simplicity, we denote a contributor storing the vertex $(i,j)$ by $P_{(i,j)}$.
Note that there can be copies of $P_{(i,j)}$.

Since there are arbitrarily many contributors, the chosen vertices are only a superset of the clique we want to find.
To cut away the false vertices, the leader $P_L$ guesses for each row the vertex belonging to the clique.
Contributors storing other vertices than the guessed ones will be switched off bit by bit.
To this end, the program performs for each $i \in [1..k]$ the following steps:
If $(i,j_i)$ is the vertex of interest, $P_L$ first writes $\row(i)$ to the memory.
Each contributor that is still active reads the symbol and moves on for one state.
Then $P_L$ communicates the column by writing $\col(j_i)$.
Again, the active contributors $P_{(i',j')}$ read.

Upon transmitting $(i,j_i)$, the contributors react in one of the
following three ways:
(1) If $i' \neq i$, the contributor $P_{(i',j')}$ stores a vertex of a different row.
The computation in $P_{(i',j')}$ can only go on if $(i',j')$ is connected to $(i,j_i)$ in $G$.
Otherwise it will stop.
(2) If $i' = i$ and $j' = j_i$, then $P_{(i',j')}$ stores exactly the vertex guessed by $P_L$.
In this case, $P_{(i',j')}$ can continue its computation.
(3) If $i' = i$ and $j' \neq j$, thread $P_{(i',j')}$ stores a different vertex from row $i$.
The contributor has to stop.

After $k$ such rounds, there are only contributors left that store vertices guessed by $P_L$.
Furthermore, each two of these vertices are connected.
Hence, they form a clique.
To transmit this information to $P_L$, each $P_{(i,j_i)}$ writes $\#_i$ to the memory, a special symbol for row $i$.
After $P_L$ has read the string $\#_1 \dots \#_k$, it moves to its final state.
A formal construction is given in Appendix \ref{Appendix_LCR_Leader_Mem}.

Note that the size $\bigO(k)$ of the data domain cannot be avoided, even if we encoded the row and column symbols in binary.
The reason is that $P_L$ needs a confirmation of $k$ contributors that were not stopped during the guessing and terminated correctly.
Since contributors do not have final states, we need to transmit this information in the form of $k$ different memory symbols.

%% file: content/LCR_Leader_Mem_Kernel.tex
\subsubsection*{Absence of a Polynomial Kernel}
\label{Section:LCR_LEader_Mem_Kernel_Lower}

A kernelization of a parameterized problem is a compression algorithm. 
Given an instance, it returns an equivalent instance the size of which is bounded by a function only in the parameter. 
From an algorithmic perspective, kernels put a bound on the number of hard instances.
Indeed, the search for small kernels is a key interest in algorithmics, similar to the search for $\FPT$-algorithms.
It can be shown that kernels exist if and only if a \mbox{problem admits an $\FPT$-algorithm~\cite{Cygan2015}.}

Let $Q$ be a parameterized problem.
A \emph{kernelization} of $Q$ is an algorithm that given an instance $(B, k)$, runs in polynomial time in $B$ and $k$,  and outputs an equivalent instance $(B', k')$ such that $\abs{B'} + k' \leq g(k)$. 
Here, $g$ is a computable function.
If $g$ is a polynomial, we say that $Q$ admits a \emph{polynomial kernel}.

Unfortunately, for many problems the community failed to come up with polynomial kernels. 
This lead to the contrary approach, namely disproving their existence \cite{Fortnow2011,Bodlaender2009,Bodlaender2014}.  
The absence of a polynomial kernel constitutes an exponential lower bound on the number of hard instances.
Like computational hardness results, such a bound is seen as an indication of general hardness of the problem. 
Technically, the existence of a polynomial kernel for the problem of interest is shown to imply $\NP \subseteq \co\NP / \poly$.
However, the inclusion is considered unlikely as it would cause a collapse of the polynomial hierarchy to the third level \cite{Yap1983}.

In order to link the existence of a polynomial kernel for $\LCR(\sizeM, \sizeL)$ with the above inclusion, we follow the framework developed in \cite{Bodlaender2014}.
Let $\Gamma$ be an alphabet.
A \emph{polynomial equivalence relation} is an equivalence relation $\polyrel$ on $\Gamma^*$ with the following properties:
Given $x,y \in \Gamma^*$, it can be decided in time polynomial in $\abs{x} + \abs{y}$ whether $(x,y) \in \polyrel$.
Moreover, for $n \in \Naturals$ there are at most polynomially many \mbox{equivalence classes in $\polyrel$ restricted to $\Gamma^{\leq n}$.}

The key tool for proving kernel lower bounds are cross-compositions.
Let \mbox{$L \subseteq \Gamma^*$} be a language and $Q \subseteq \Gamma^* \times \mathbb{N}$ be a parameterized language. 
We say that $L$ \emph{cross-composes} into $Q$ if there exists a polynomial equivalence relation $\polyrel$ and an algorithm $\crosscomp$, together called the \emph{cross-composition}, with the following properties:
$\mathcal{C}$ takes as input $\varphi_1, \dots, \varphi_I \in \Gamma^*$, all equivalent under $\polyrel$.
It computes in time polynomial in $\sum_{\ell=1}^{I} \abs{\varphi_\ell}$ a string $(y,k) \in \Gamma^* \times \Naturals$ such that $(y,k) \in Q$ if and only if there is an $\ell \in [1..I]$ with $\varphi_\ell \in L$.
Furthermore, parameter $k$ is bounded by $p(\max_{\ell \in [1..I]} \abs{\varphi_\ell} + \log(I)  )$, where $p$ is a polynomial.

It was shown in \cite{Bodlaender2014} that a cross-composition of any $\NP$-hard language into a parameterized language $Q$ prohibits the existence of a polynomial kernel for $Q$ unless $\NP \subseteq \co\NP / \poly$.
In order to make use of this result, we show how to cross-compose $\kSAT{3}$ into $\LCR(\sizeM, \sizeL)$.
This yields the following:

\begin{theorem}\label{Theorem:LCR_Leader_Mem_Kernel_Lower}
	$\LCR(\sizeM, \sizeL)$ does not admit a poly. kernel unless $\NP \subseteq \co\NP / \poly$.
\end{theorem}

The difficulty in coming up with a cross-composition is the restriction on the size of the parameters.
In our case, this affects $\sizeM$ and $\sizeL$:
Both parameters are not allowed to depend polynomially on $I$, the number of given $\kSAT{3}$-instances. 
We resolve the polynomial dependence by encoding the choice of  such an instance into the contributors via a binary tree.

\begin{proof}[Idea]
	Assume some encoding of Boolean formulas as strings over a finite alphabet.
	We use the polynomial equivalence relation $\polyrel$ defined as follows:
	Two strings $\varphi$ and $\psi$ are equivalent under $\polyrel$ if both encode $\kSAT{3}$-instances, and the numbers of clauses and variables coincide.
	
	Let the given $\kSAT{3}$-instances be $\varphi_1, \dots, \varphi_I$.
	Every two of them are equivalent under $\polyrel$.
	This means all $\varphi_\ell$ have the same number of clauses $m$ and use the same set of variables $\setcon{x_1, \dots, x_n}$.
	We assume that $\varphi_\ell = C^\ell_1 \wedge \dots \wedge C^\ell_m$.
	
	We construct a program proceeding in three phases.
	First, it chooses an instance $\varphi_\ell$, then it guesses an evaluation for all variables, and in the third phase it verifies that the evaluation satisfies $\varphi_\ell$.
	While the second and the third phase do not cause a dependence of the parameters on $I$, the first phase does. 
	It is not possible to guess a number $\ell \in [1..I]$ and communicate it via the memory as this would provoke a polynomial dependence of $\sizeM$ on $I$.
	
	To implement the first phase without a polynomial dependence, we transmit the indices of the $\kSAT{3}$-instances in binary.
	The leader guesses and writes tuples $(u_1, 1), \dots, (u_{\log(I)}, \log(I))$ with $u_\ell \in \setcon{0,1}$ to the memory.
	This amounts to choosing an instance $\varphi_\ell$ with binary representation $\bin(\ell) = u_1 \dots u_{\log(I)}$.
	
	It is the contributors' task to store this choice.
	Each time the leader writes a tuple $(u_i,i)$, the contributors read and branch either to the left, if $u_i = 0$, or to the right, if $u_i = 1$.
	Hence, in the first phase, the contributors are binary trees with $I$ leaves, each leaf storing the index of an instance $\varphi_\ell$.
	Since we did not assume that $I$ is a power of $2$, there may be computations arriving at leaves that do not represent proper indices.
	\mbox{In this case, the computation deadlocks.}
	
	The size of $\Domain$ and $P_L$ in the first phase is $\bigO(\log(I))$. 
	Note that this satisfies the size-restrictions of a cross-composition.
	
	For guessing the evaluation in the second phase, the program communicates on tuples $(x_i, v)$ with $i \in [1..n]$ and $v \in \setcon{0,1}$.
	The leader guesses such a tuple for each variable and writes it to the memory.
	Any participating contributor is free to read one of these.
	After reading, it stores the variable and the evaluation.
	
	In the third phase, the satisfiability check is performed as follows:
	Each contributor that is still active has stored in its current state the chosen instance $\varphi_\ell$, a variable $x_i$, and its evaluation $v_i$.
	Assume that $x_i$ when evaluated to $v_i$ satisfies $C^\ell_j$, the $j$-th clause of $\varphi_\ell$.
	Then the contributor loops in its current state while writing the symbol $\#_j$.
	The leader waits to read the string $\#_1 \dots \#_m$.
	If $P_L$ succeeds, we are sure that the $m$ clauses of $\varphi_\ell$ were satisfied by the chosen evaluation.
	Thus, $\varphi_\ell$ is satisfiable and $P_L$ moves to its final state.
	For details of the construction and a proof of correctness, we refer to Appendix \ref{Appendix_LCR_Leader_Mem}.
	\qed
\end{proof}

%% file: content/LCR_Contributor.tex
\subsection{Parameterization by Contributors}
\label{Section:LCR_Contributor}

The size of the contributors $\sizeC$ has substantial influence on the complexity of $\LCR$. 
We show that the problem can be solved in time $\bigOS(2^\sizeC)$ via dynamic programming.
Moreover, we present a matching lower bound proving it unlikely that $\LCR$ can be solved in time $\bigOS((2 - \delta)^\sizeC)$, for any $\delta > 0$.
The result is obtained by a reduction from $\SetCov$.
Finally, \mbox{we prove the absence of a polynomial kernel.}

\input{content/LCR_Contributor_Upper}
\input{content/LCR_Contributor_Lower}
\input{content/LCR_Contributor_Kernel}

%% file: content/LCR_Contributor_Upper.tex
\subsubsection*{Upper Bound}
\label{Section:LCR_Contributor_Upper}

Our algorithm is based on dynamic programming.
Intuitively, we cut a computation of the program along the states reached by the contributors.
To this end, we keep a table with an entry for each subset of the contributors' states.
The entry of set $S \subseteq Q_C$ contains those states of the leader that are reachable under a computation where the behavior of the contributors is limited to $S$.
We fill the table by a dynamic programming procedure and check in the end whether a final state of the leader occurs in an entry.
The result is as follows.

\begin{theorem}\label{Theorem:LCR_Contributor_Upper_Bound}
	$\LCR$ can be solved in time $\bigO( 2^\sizeC \cdot \sizeC^{\: 4} \cdot \sizeL^2 \cdot \sizeM^{\, 2} )$.
\end{theorem}

To define the table, we first need a compact way of representing computations that allows for fast iteration.
The observation is that keeping one set of states for all contributors suffices.
Let $S \subseteq Q_C$ be the set of states reachable by the contributors in a given computation.
By the \emph{Copycat Lemma} \cite{Esparza13}, we can assume for each $q \in S$ an arbitrary number of contributors that are currently in $q$.
This means that we do not have to distinguish \mbox{between different contributor instances.}

Formally, we reduce the search space to $V = Q_L \times \Domain \times \power(Q_C)$.
Instead of explicit configurations, we consider tuples $(q, a, S)$, where $q \in Q_L$, $a \in \Domain$, and $S \subseteq Q_C$.
Between these tuples, we define an edge relation $E$.
If $P_L$ writes $a \in \Domain$ with transition $q \xrightarrow{!a} q'$, we get $(q, b, S) \rightarrow_E (q', a, S)$ for each $b \in \Domain$ and $S \subseteq Q_C$.
Reads of the leader are similar.
Contributors also change the memory but saturate set $S$ instead of changing the state:
If there is a transition $p \xrightarrow{!a} p'$ in $P_C$ with $p \in S$, then $(q, b, S) \rightarrow_E (q, a, S \cup \setcon{p'})$ for each $b \in \Domain$ and $q \in Q_L$.
Reads are handled similarly.

The set $V$ together with the relation $E$ form a finite directed graph \mbox{$\fingraph = (V,E)$}.
We call the node $\initnode = (q^0_L, \initmem, \setcon{q^0_C})$ the \emph{initial node}.
Computations are represented by paths in $\fingraph$ starting in $\initnode$.
Hence, we reduced $\LCR$ to the problem of checking whether the set of nodes $F_L \times \Domain \times \power(Q_C)$ is reachable from $\initnode$ in $\fingraph$.

\begin{lemma}\label{Lemma:LCR_Contributor_Upper_Bound_Reach}
	There is a $t \in \Naturals$ so that $c^0 \rightarrow^*_{\asms^t} c$ with $c \in C^f$ if and only if there is a path in $\fingraph$ from $\initnode$ to a node in $F_L \times \Domain \times \power(Q_C)$.
\end{lemma}

Before we elaborate on solving reachability on $\fingraph$ we turn to an example.
It shows how $\fingraph$ is constructed from a program and illustrates Lemma~\ref{Lemma:LCR_Contributor_Upper_Bound_Reach}.

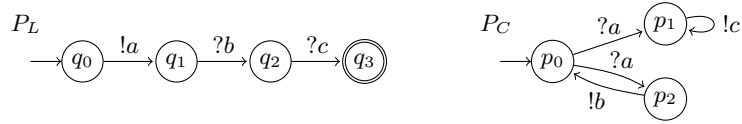
\begin{figure}
	\begin{center}
		\begin{tikzpicture}
		\node[state, draw, inner sep = 2pt, minimum size = 0pt, initial, initial text = ] (q0) {$q_0$};
		\node[circle, draw, inner sep = 2pt, right of = q0, node distance = 1.25cm] (q1) {$q_1$};
		\node[circle, draw, inner sep = 2pt, right of = q1, node distance = 1.25cm] (q2) {$q_2$};
		\node[circle, draw, inner sep = 2pt, right of = q2, node distance = 1.25cm, accepting] (q3) {$q_3$};
		
		\node[left of = q0, node distance = 0.75cm, yshift = 0.5cm] {$P_L$};
		
		\path[->, shorten >= 1pt] (q0) edge 
			node[above] {$!a$}
		(q1);
		
		\path[->, shorten >= 1pt] (q1) edge
			node[above] {$?b$}
		(q2);
		
		\path[->, shorten >= 1pt] (q2) edge 
			node[above] {$?c$}
		(q3);
		
		\node[circle, draw, inner sep = 2pt, right of = q3, node distance = 2.5cm, initial, initial text = ] (p0) {$p_0$};
		\node[circle, draw, inner sep = 2pt, right of = p0, node distance = 1.5cm, yshift = 0.5cm] (p1) {$p_1$};
		\node[circle, draw, inner sep = 2pt, below of = p1] (p2) {$p_2$};
		
		\node[left of = p0, node distance = 0.75cm, yshift = 0.5cm] {$P_C$};
		
		\path[->, shorten >= 1pt] (p0) edge
			node[above] {$?a$}
		(p1);
		
		\path[->, shorten >= 1pt] (p1) edge[loop right]
			node[right] {$!c$}
		(p1);
		
		\path[->, shorten >= 1pt] (p0) edge[out = 350, in = 150]
			node[above, pos = 0.65] {$?a$}
		(p2);
		
		\path[->, shorten >= 1pt] (p2) edge[out = 170, in = 330]
			node[below, pos = 0.65] {$!b$}
		(p0);
		\end{tikzpicture}
	\end{center}
	\caption{Leader thread $P_L$ (left) and contributor thread $P_C$ (right). The data domain is given by $\Domain = \setcon{a^0, a, b, c}$ and the only unsafe state is $F_L = \setcon{q_3}$.}
	\label{Figure:ComputationGraphSetup}
\end{figure}

\begin{example}
	We consider the program $\asms = (\Domain, \initmem, (P_L,P_C))$ from Figure \ref{Figure:ComputationGraphSetup}.
	The nodes of the corresponding graph $\fingraph$ are given by $V = Q_L \times \Domain \times \power(\setcon{p_0,p_1,p_2})$.
	Its edges $E$ are constructed following the above rules.
	For instance, we get an edge $(q_1,a,\setcon{p_0}) \rightarrow_E (q_1,a,\setcon{p_0,p_1})$ since $P_C$ has a read transition $p_0 \xrightarrow{?a} p_1$.
	Intuitively, the edge describes that currently, the leader is in state $q_1$, the memory holds $a$, and an arbitrary number of contributors is waiting in $p_0$.
	Then, some of these read $a$ and move to $p_1$.
	Hence, we might assume an arbitrary number of contributors in the states $p_0$ and $p_1$.
	
	The complete graph $\fingraph$ is presented in Figure \ref{Figure:ComputationGraph}.
	For the purpose of readability, we only show the nodes reachable from the initial node \mbox{$\initnode = (q_0,a^0,\setcon{p_0})$}.
	Moreover, we omit self-loops and we present the graph as a collection of subgraphs.
	The latter means that for each subset $S$ of $\power(\setcon{p_0,p_1,p_2})$, we consider the induced subgraph $\fingraph[Q_L \times \Domain \times \setcon{S}]$.
	It contains the set of nodes $Q_L \times \Domain \times \setcon{S}$ and all edges that start and end in this set.
	Note that we omit the last component from a node $(q,a,S)$ in $\fingraph[Q_L \times \Domain \times \setcon{S}]$ since it is clear from the context.
	The induced subgraphs are connected by edges that saturate $S$.
	
	The red marked nodes are those which contain the unsafe state $q_3$ of the leader.
	Consider a path from $\initnode$ to one of these nodes.
	It starts in the subgraph $\fingraph[Q_L \times \Domain \times \setcon{\setcon{p_0}}]$.
	To reach one of the red nodes, the path has to traverse via $\fingraph[Q_L \times \Domain \times \setcon{\setcon{p_0,p_1}}]$ to $\fingraph[Q_L \times \Domain \times \setcon{Q_C}]$ or via $\fingraph[Q_L \times \Domain \times \setcon{\setcon{p_0,p_2}}]$.
	Phrased differently, the states of the contributors need to be saturated two times along the path.
	This means that in an actual computation, there must be contributors in $p_0$, $p_1$, and $p_2$.
	These can then provide the symbols $b$ and $c$ which are needed by the leader to reach the state $q_3$.
\end{example}

\begin{figure}
	\begin{center}
		\begin{tikzpicture}[
			dot/.style = {circle, fill = black, inner sep = 2pt}]
		
			\node[dot] (q0a0-p0) {};
			\node[above of = q0a0-p0, node distance = 0.3cm] (q0a0-p0Name) {\scriptsize $(q_0,a^0)$};
			
			\node[dot, below of = q0a0-p0] (q1a-p0) {};
			\node[below of = q1a-p0, node distance = 0.3cm] (q1a-p0Name) {\scriptsize $(q_1,a)$};
			
			\path[->, shorten >= 1pt] (q0a0-p0) edge (q1a-p0);
			
			\node[fit = (q0a0-p0)(q0a0-p0Name)(q1a-p0)(q1a-p0Name)] (graph-p0) {};
			\node[above of = graph-p0, node distance = 1.5cm] (graph-p0Name) {$\fingraph[Q_L \times \Domain \times \setcon{\setcon{p_0}}]$};
			
			\begin{pgfonlayer}{background}
				\path[fill = lightgray!20, rounded corners] (graph-p0.south west) rectangle (graph-p0.north east);
			\end{pgfonlayer}
			
			\node[dot, right of = q0a0-p0, node distance = 3cm] (q1a-p1) {};
			\node[below of = q1a-p1, node distance = 0.3cm] (q1a-p1Name) {\scriptsize $(q_1,a)$};
			
			\node[dot, above of = q1a-p1] (q1c-p1) {};
			\node[above of = q1c-p1, node distance = 0.3cm] (q1c-p1Name) {\scriptsize $(q_1,c)$};
			
			\path[->, shorten >= 1pt] (q1a-p1) edge (q1c-p1);
			\path[->, shorten >= 1pt] (q1a-p0) edge (q1a-p1);
			
			\node[fit = (q1a-p1)(q1a-p1Name)(q1c-p1)(q1c-p1Name)] (graph-p1) {};
			\node[above of = graph-p1, node distance = 1.5cm] (graph-p1Name) {$\fingraph[Q_L \times \Domain \times \setcon{\setcon{p_0,p_1}}]$};
			
			\begin{pgfonlayer}{background}
				\path[fill = lightgray!20, rounded corners] (graph-p1.south west) rectangle (graph-p1.north east);
			\end{pgfonlayer}
			
			\node[dot, right of = q1a-p0, node distance = 3cm, yshift = -1cm] (q1a-p2) {};
			\node[above of = q1a-p2, node distance = 0.3cm] (q1a-p2Name) {\scriptsize $(q_1,a)$};
			
			\node[dot, below of = q1a-p2, xshift = -0.5cm] (q1b-p2) {};
			\node[below of = q1b-p2, node distance = 0.3cm] (q1b-p2Name) {\scriptsize $(q_1,b)$};
			
			\node[dot, below of = q1a-p2, xshift = 0.5cm] (q2b-p2) {};
			\node[below of = q2b-p2, node distance = 0.3cm] (q2b-p2Name) {\scriptsize $(q_2,b)$};
			
			\path[->, shorten >= 1pt] (q1a-p2) edge (q1b-p2);
			\path[->, shorten >= 1pt] (q1b-p2) edge (q2b-p2);
			\path[->, shorten >= 1pt] (q1a-p0) edge (q1a-p2);
			
			\node[fit = (q1a-p2)(q1a-p2Name)(q1b-p2)(q1b-p2Name)(q2b-p2)(q2b-p2Name)] (graph-p2) {};
			\node[below of = graph-p2, node distance = 1.5cm] (graph-p2Name) {$\fingraph[Q_L \times \Domain \times \setcon{\setcon{p_0,p_2}}]$};
			
			\begin{pgfonlayer}{background}
				\path[fill = lightgray!20, rounded corners] (graph-p2.south west) rectangle (graph-p2.north east);
			\end{pgfonlayer}
		
			\node[dot, right of = q1a-p0, node distance = 6cm] (q1a-F) {};
			\node[above of = q1a-F, node distance = 0.3cm] (q1a-FName) {\scriptsize $(q_1,a)$};

			\node[dot, above of = q1a-F, xshift = 1cm] (q1b-F) {};
			\node[above of = q1b-F, node distance = 0.3cm] (q1b-FName) {\scriptsize $(q_1,b)$};
			
			\node[dot, right of = q1b-F] (q2b-F) {};
			\node[above of = q2b-F, node distance = 0.3cm] (q2b-FName) {\scriptsize $(q_2,b)$};
			
			\node[dot, fill = red, right of = q2b-F] (q3b-F) {};
			\node[above of = q3b-F, node distance = 0.3cm] (q3b-FName) {\scriptsize $(q_3,b)$};

			\node[dot, below of = q1a-F, xshift = 1cm] (q1c-F) {};
			\node[below of = q1c-F, node distance = 0.3cm] (q1c-FName) {\scriptsize $(q_1,c)$};
			
			\node[dot, right of = q1c-F] (q2c-F) {};
			\node[below of = q2c-F, node distance = 0.3cm] (q2c-FName) {\scriptsize $(q_2,c)$};
			
			\node[dot, fill = red, right of = q2c-F] (q3c-F) {};
			\node[below of = q3c-F, node distance = 0.3cm] (q3c-FName) {\scriptsize $(q_3,c)$};
			
			\path[->, shorten >= 1pt] (q1a-F) edge[out = 10, in = 190] (q1b-F);
			\path[->, shorten >= 1pt] (q1a-F) edge[out = 350, in = 170] (q1c-F);
			
			\path[->, shorten >= 1pt] (q1b-F) edge (q2b-F);
			\path[->, shorten >= 1pt] (q1b-F) edge[out = 290, in = 70] (q1c-F);
			
			\path[->, shorten >= 1pt] (q1c-F) edge[out = 110, in = 250] (q1b-F);
			
			\path[->, shorten >= 1pt] (q2b-F) edge[out = 290, in = 70] (q2c-F);
			
			\path[->, shorten >= 1pt] (q2c-F) edge (q3c-F);
			\path[->, shorten >= 1pt] (q2c-F) edge[out = 110, in = 250] (q2b-F);
			
			\path[->, shorten >= 1pt] (q3b-F) edge[out = 290, in = 70] (q3c-F);
			
			\path[->, shorten >= 1pt] (q3c-F) edge[out = 110, in = 250] (q3b-F);
			
			\path[->, shorten >= 1pt] (q1a-p1) edge (q1a-F);
			\path[->, shorten >= 1pt] (q1a-p2) edge (q1a-F);
			
			\node[fit = (q1a-F)(q1a-FName)(q1b-F)(q1b-FName)(q2b-F)(q2b-FName)(q3b-F)(q3b-FName)(q1c-F)(q1c-FName)(q2c-F)(q2c-FName)(q3c-F)(q3c-FName)] (graph-F) {};
			\node[above of = graph-F, node distance = 2cm] (graph-FName) {$\fingraph[Q_L \times \Domain \times \setcon{Q_C}]$};
			
			\begin{pgfonlayer}{background}
				\path[fill = lightgray!20, rounded corners] (graph-F.south west) rectangle (graph-F.north east);
			\end{pgfonlayer}
			
			\node[fit = (graph-p0)(graph-p0Name)(graph-p1)(graph-p1Name)] (slice) {};
			
			\coordinate[yshift = 1cm] (newsoutheast) at (slice.south east);
			
			\begin{pgfonlayer}{backbackground}
				\path[fill = blue!20, rounded corners] (slice.south west) -- (slice.south) -- ++ (0,1) -- (newsoutheast) -- (slice.north east) -- (slice.north west) -- cycle;
			\end{pgfonlayer}
			
		\end{tikzpicture}
	\end{center}
	\caption{Graph $\fingraph$ summarizing the computations of the program given in Figure \ref{Figure:ComputationGraphSetup}.
		Self-loops and nodes not reachable from the initial node $\initnode = (q_0,a^0,\setcon{p_0})$ are omitted.
		We further omit the third component of nodes since it is clear from the context.
		Nodes that are marked red involve the unsafe state $q_3$ of the leader.
		The blue highlighted area shows the slice $\fingraph_{\setcon{p_0},\setcon{p_0,p_1}}$.}
	\label{Figure:ComputationGraph}
\end{figure}
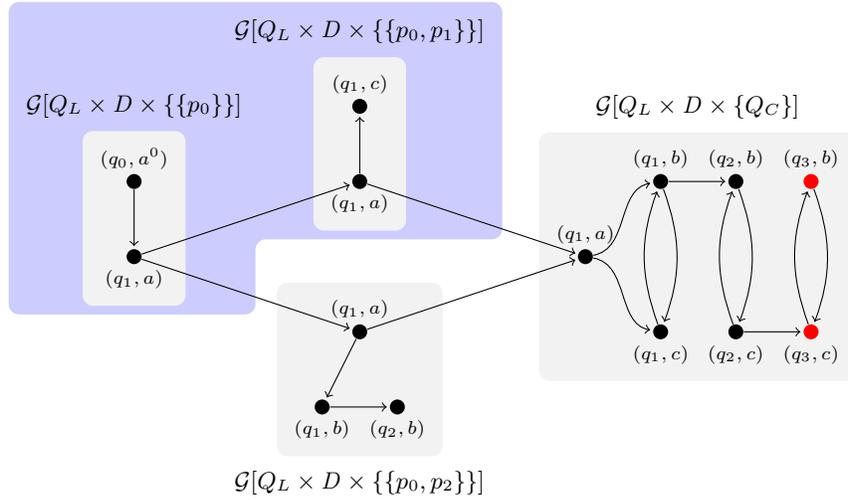

Constructing $\fingraph$ for a program and solving reachability takes time~$\bigOS(4^\sizeC)$~\cite{Saivasan2018}.
Hence, we have to solve reachability without constructing $\fingraph$ explicitly.
Our algorithm computes a table $T$ which admits a recurrence relation that simplifies the reachability query:
Instead of solving reachability directly on $\fingraph$, we can restrict to so-called \emph{slices} of~$\fingraph$.
These are subgraphs of polynomial size where reachability queries can be decided efficiently.

We define the table $T$.
For each set $S \subseteq Q_C$, \mbox{we have an entry $T[S]$ given by:}
\begin{align*}
	T[S] = \Setcon{ (q,a) \in Q_L \times \Domain }{ v^0 \rightarrow^*_E (q,a,S) }.
\end{align*}
Intuitively, $T[S]$ contains all nodes from the induced subgraph $\fingraph[Q_L \times \Domain \times \setcon{S}]$ that are reachable from the initial node $\initnode$.

Assume we have already computed $T$.
By Lemma \ref{Lemma:LCR_Contributor_Upper_Bound_Reach} we get:
There is a $t \in \Naturals$ so that $c^0 \rightarrow^*_{\asms^t} c \in C^f$ if and only if there is a set $S \subseteq Q_C$ such that $T[S] \cap F_L \times D \neq \emptyset$.
The latter can be checked in time $\bigO(2^\sizeC \cdot \sizeL^2 \cdot \sizeM^2)$ as there are $2^\sizeC$ candidates for a set $S$.

It remains to compute the table.
Our goal is to employ a dynamic programming based on a recurrence relation over $T$.
To formulate the relation, we need the notion of \emph{slices} of $\fingraph$.
Let $W \subseteq Q_C$ be a subset and $p \in Q_C \setminus W$ be a state.
We denote by $S$ the union $S = W \cup \setcon{p}$.
The \emph{slice} $\fingraph_{W,S}$ is the induced subgraph $\fingraph[ Q_L \times \Domain \times \setcon{W,S} ]$.
We denote its set of edges by $E_{W,S}$.

The main idea of the recurrence relation is saturation.
When traversing a path $\pi$ in $\fingraph$, the set of contributor states gets saturated over time.
Assume we cut $\pi$ each time after a new state gets added.
Then we obtain subpaths, each being a path in a slice:
If $p \in Q_C$ gets added to $W \subseteq Q_C$, the corresponding subpath is in $\fingraph_{W,W\cup\setcon{p}}$.
This means that for a set $S \subseteq Q_C$, the entry $T[S]$ contains those nodes that are reachable from $T[S \setminus \! \setcon{p}]$ in the slice $\fingraph_{S\setminus\!\setcon{p},S}$, for some state $p \in S$.

Formally, we define the set $R(W,S)$ for each $W \subseteq Q_C$, $p \in Q_C \setminus W$, and $S = W \cup \setcon{p}$.
These sets contain the nodes reachable from $T[W]$ in $\fingraph_{W,S}$:
\begin{align*}
	R(W,S) = \Setcon{ (q,a) \in Q_L \times \Domain }{ \exists (q',a') \in T[W] \text{ with } (q',a',W) \rightarrow^*_{E_{W,S}} (q,a,S)}.
\end{align*}

\begin{lemma}\label{Lemma:LCR_Contributor_Upper_Bound_Recursion}
	Table $T$ admits the recurrence relation $T[S] = \bigcup_{p \in S} R(S \setminus \! \setcon{p},S)$.
\end{lemma}

We illustrate the lemma and the introduced notions on an example.
Afterwards, we show how to compute the table $T$ by exploiting the recurrence relation.

\begin{example}
	Reconsider the program given in Figure \ref{Figure:ComputationGraphSetup}.
	The table $T$ has eight entries, one for each subset of $Q_C$.
	The entries that are non-empty can be seen in the graph of Figure \ref{Figure:ComputationGraph}.
	Each of the subgraphs contains exactly those nodes that are reachable from $\initnode$.
	For instance $T[\setcon{p_0,p_1}] = \setcon{(q_1,a),(q_1,c)}$.
	
	Let $W = \setcon{p_0}$ and $S = \setcon{p_0,p_1}$.
	Then, the slice $\fingraph_{W,S}$ is shown in the figure as blue highlighted area.
	Note that it also contains the edge from $(q_1,a,\setcon{p_0})$ to $(q_1,a,\setcon{p_0,p_1})$, leading from $\fingraph[Q_L \times \Domain \times \setcon{\setcon{p_0}}]$ to $\fingraph[Q_L \times \Domain \times \setcon{\setcon{p_0,p_1}}]$.
	
	The set $R(W,S)$ contains those nodes in $\fingraph[Q_L \times \Domain \times \setcon{S}]$ that are reachable from $T[W]$ in the slice $\fingraph_{W,S}$.
	According to the graph, these are $(q_1,a,S)$ and $(q_1,c,S)$ and hence we get $T[S] = R(W,S)$.
	
	In general, not all nodes in $T[S]$ are reachable from a single $T[W]$.
	But if a node is reachable, then it is reachable from some set $T[S\setminus\setcon{p}]$ with $p \in S$.
	Note that this is covered by the recurrence relation in Lemma \ref{Lemma:LCR_Contributor_Upper_Bound_Recursion}.
	It branches over all $S \setminus \setcon{p}$ and hence collects all nodes that are reachable from an entry $T[S \setminus \setcon{p}]$.
\end{example}

We apply the recurrence relation in a bottom-up dynamic programming to fill the table $T$.
Let $S \subseteq Q_C$ be a subset and assume we already know $T[S \setminus \! \setcon{p}]$, for each $p \in S$.
Then, for a fixed $p$, we compute $R(S \setminus \! \setcon{p},S)$ by a fixed-point iteration on the slice $\fingraph_{S \setminus \setcon{p},S}$.
The number of nodes in the slice is $\bigO(\sizeL \cdot \sizeM)$.
Hence, the iteration takes time at most $\bigO(\sizeL^2 \cdot \sizeM^2)$.
It is left to construct $\fingraph_{S \setminus \setcon{p}, S}$.
We state the time needed in the following lemma.
The proof is postponed so as to finish the complexity estimation of Theorem \ref{Theorem:LCR_Contributor_Upper_Bound}.

\begin{lemma}\label{Lemma:LCR_Contributor_Upper_Bound_Construction_Slice}
	Slice $\fingraph_{S \setminus \setcon{p}, S}$ can be constructed in time $\bigO(\sizeC^{\: 3} \cdot \sizeL^2 \cdot \sizeM^{\,2})$.
\end{lemma}

Wrapping up, we need $\bigO( \sizeC^3 \cdot \sizeL^2 \cdot \sizeM^2 )$ time for computing a set $R(S \setminus \! \setcon{p},S)$.
Due to the recurrence relation of Lemma \ref{Lemma:LCR_Contributor_Upper_Bound_Recursion}, we have to compute at most $\sizeC$ sets $R(S \setminus \! \setcon{p},S)$ for a given $S \subseteq Q_C$.
Hence, an entry $T[S]$ can be computed in time $\bigO(\sizeC^4 \cdot \sizeL^2 \cdot \sizeM^2)$.
The estimation also covers the base case $S = \setcon{p_C^0}$, where $T[S]$ can be computed by a fixed-point iteration on the induced subgraph $\fingraph[Q_L \times \Domain \times \setcon{S}]$.
Since the table $T$ has $2^\sizeC$ entries, the complexity estimation of Theorem \ref{Theorem:LCR_Contributor_Upper_Bound} follows.
It is left to prove Lemma~\ref{Lemma:LCR_Contributor_Upper_Bound_Construction_Slice}.

\begin{proof}
	The slice $\fingraph_{S \setminus \setcon{p},S}$ consists of the two subgraphs 
	\begin{align*}
		\fingraph_{S \setminus \setcon{p}} = \fingraph[Q_L \times \Domain \times \setcon{S \setminus \! \setcon{p}}]~\text{and}~\fingraph_{S} = \fingraph[Q_L \times \Domain \times \setcon{S}],
	\end{align*}
	and the edges leading from $\fingraph_{S \setminus \setcon{p}}$ to $\fingraph_{S}$.
	We elaborate on how to construct $\fingraph_{S}$.
	The construction of $\fingraph_{S \setminus \setcon{p}}$ is similar.
	
	First, we write down the nodes of $\fingraph_{S}$.
	This can be done in time $\bigO(\sizeL \cdot \sizeM)$.
	Edges in the graph are either induced by transitions of the leader or by the contributor.
	The former ones can be added in time $\bigO(\abs{\delta_L} \cdot \sizeM) = \bigO(\sizeL^2 \cdot \sizeM^2)$ since a single transition of $P_L$ may lead to $\sizeM$ edges.
	To add the latter edges, we browse $\delta_C$ for transitions of the form $s \xrightarrow{!a} s'$ with $s,s' \in S$.
	Each such transition may induce $\sizeL \cdot \sizeM$ edges.
	Adding them takes time $\bigO(\abs{\delta_C} \cdot \sizeC \cdot \sizeL \cdot \sizeM) = \bigO(\sizeC^3 \cdot \sizeL \cdot \sizeM^2)$ since we have to test membership of $s,s'$ in $S$.
	Note that we can omit transitions $s \xrightarrow{?a} s'$ with $s,s' \in S$ as their induced edges are self-loops in $\fingraph_{S}$.
	
	To complete the construction, we add the edges from $\fingraph_{S \setminus \setcon{p}}$ to $\fingraph_{S}$.
	These are induced by transitions $r \xrightarrow{?a / !a} p \in \delta_C$ with $r \in S \setminus \! \setcon{p}$.
	Since each of these may again lead to $\sizeL \cdot \sizeM$ different edges, adding all of them takes time $\bigO(\sizeC^3 \cdot \sizeL \cdot \sizeM^2)$.
	In total, we estimate the time for the construction by $\bigO(\sizeC^3 \cdot \sizeL^2 \cdot \sizeM^2)$.
	\qed
\end{proof}

%% file: content/LCR_Contributor_Lower.tex
\subsubsection*{Lower bound}
\label{Section:LCR_Contributor_Lower}

We prove it unlikely that $\LCR$ can be solved in $\bigOS((2-\delta)^\sizeC)$ time, for any $\delta > 0$.
This shows that the algorithm from Section \ref{Section:LCR_Contributor_Upper} has an optimal runtime.
The lower bound is achieved by a reduction from $\SetCov$, one of the~21 original $\NP$-complete problems by Karp \cite{Karp1972}.
We state its definition.

\begin{myproblem}
	\problemtitle{$\SetCov$}
	\probleminput{A family of sets $\mathcal{F} \subseteq \power(U)$ over a universe $U$, and $r \in \Naturals$.}
	\problemquestion{Are there sets $S_1, \dots, S_r$ in $\mathcal{F}$ such that $U = \bigcup_{i \in [1..r]} S_i$?}
\end{myproblem}

Besides its $\NP$-completeness, it is known that $\SetCov$ admits an $\bigOS(2^n)$-time algorithm \cite{Fomin2004}, where $n$ is the size of the universe $U$.
However, no algorithm solving $\SetCov$ in time $\bigOS((2-\delta)^n)$ for a $\delta > 0$ is known so far.
Actually, it is conjectured in \cite{Cygan2016} that such an algorithm cannot exist unless the $\SETH$ breaks.

While a proof for the conjecture in \cite{Cygan2016} is still missing, the authors provide evidence in the form of relative hardness.
They obtain lower bounds for prominent problems by tracing back to the \emph{assumed} lower bound of $\SetCov$.
These bounds were not known before since $\SETH$ is hard to apply:
No suitable reductions from $\SAT$ to these problems are known so far.
Hence, $\SetCov$ can be seen as an alternative source for lower bounds whenever $\SETH$ seems out of reach.
This made the problem a standard assumption for hardness~\cite{Cygan2016,Bjorklund2016,Chini2017}.

To obtain the desired lower bound for $\LCR$, we establish a polynomial time reduction from $\SetCov$ that strictly preserves the parameter $n$.
Formally, if $(\mathcal{F},U,r)$ is an instance of $\SetCov$, we construct \mbox{$(\asms = (\Domain, \initmem, (P_L, P_C)), F_L)$}, an instance of $\LCR$ where $\sizeC = n + c$ with $c$ a constant.
Note that even a linear dependence on $n$ is not allowed.
Moreover, the instance satisfies the equivalence:
There is a set cover if and only if there is a $t \in \Naturals$ such that $c^0 \rightarrow^*_{\asms^t} c$ with $c \in C^f$.
Assume we had an $\bigOS((2-\delta)^\sizeC)$-time algorithm for $\LCR$.
With the reduction, this would immediately yield an $\bigOS((2-\delta)^{n+c}) = \bigOS((2-\delta)^n)$-time algorithm for $\SetCov$ breaking its hardness.

\begin{proposition}\label{Proposition:LCR_Contributor_Lower_Bound}
	If $\LCR$ can be solved in $\bigOS((2-\delta)^\sizeC)$ time for a $\delta > 0$, then $\SetCov$ can be solved in $\bigOS((2-\delta)^n)$ time.
\end{proposition}

For the proof of the proposition, we elaborate on the aforementioned reduction.
The main idea is the following: 
We let the leader guess $r$ sets from $\mathcal{F}$.
The contributors store the elements that got covered by the chosen sets.
In a final communication phase, the leader verifies that it has chosen a valid cover by querying whether all elements of $U$ have been stored by the contributors.

Leader and contributors essentially communicate over the elements of $U$.
For guessing $r$ sets from $\mathcal{F}$, the automaton $P_L$ consists of $r$ similar phases.
Each phase starts with $P_L$ choosing an internal transition to a set $S \in \mathcal{F}$.
Once $S$ is chosen, the leader writes a sequence of all $u \in S$ to the memory.

A contributor in the program consists of $\sizeC = n+1$ states: 
An initial state and a state for each $u \in U$.
When $P_L$ writes an element $u \in S$ to the memory, there is a contributor storing this element in its states by reading $u$.
Hence, each element that got covered by $S$ is recorded in one of the contributors.

After $r$ rounds of guessing, the contributors hold those elements of $U$ that are covered by the chosen sets.
Now the leader verifies that it has really picked a cover of $U$.
To this end, it needs to check whether all elements of $U$ have been stored by the contributors.
Formally, the leader can only proceed to its final state if it can read the symbols $u^\#$, for each $u \in U$.
A contributor can only write $u^\#$ to the memory if it stored the element $u$ before.
Hence, $P_L$ reaches its final state if and only if a valid cover of $U$ was chosen.

%% file: content/LCR_Contributor_Kernel.tex
\subsubsection*{Absence of a Polynomial Kernel}
\label{Section:LCR_Contributor_Kernel}

We prove that $\kSAT{3}$ can be cross-composed into $\LCR(\sizeC)$.
This shows that the problem is unlikely to admit a polynomial kernel.
The result is the following.

\begin{proposition}\label{Proposition:LCR_Contributor_Kernel_Lower}
	$\LCR(\sizeC\, )$ does not admit a poly. kernel unless $\NP \subseteq \co\NP / \poly$.
\end{proposition}

For the cross-composition, let $\varphi_1, \dots, \varphi_I$ be the given $\kSAT{3}$-instances, each two equivalent under $\polyrel$, where $\polyrel$ is the polynomial equivalence relation from Theorem~\ref{Theorem:LCR_Leader_Mem_Kernel_Lower}.
Then, each formula has the same number of clauses $m$ and variables $x_1, \dots, x_n$.
Let us fix the notation to be $\varphi_\ell = C^\ell_1 \wedge \dots \wedge C^\ell_m$.

The basic idea is the following.
Leader $P_L$ guesses the formula $\varphi_\ell$ and an evaluation for the variables.
The contributors store the latter.
At the end, leader and contributors verify that the chosen evaluation indeed satisfies formula $\varphi_\ell$.

For guessing $\varphi_\ell$, the leader has a branch for each instance.
Note that we can afford the size of the leader to depend on $I$ since the cross-composition only restricts parameter $\sizeC$.
Hence, we do not face the problem we had in Theorem \ref{Theorem:LCR_Leader_Mem_Kernel_Lower}.

Guessing the evaluation of the variables is similar to Theorem \ref{Theorem:LCR_Leader_Mem_Kernel_Lower}:
The leader writes tuples $(x_i,v_i)$ with $v_i \in \setcon{0,1}$ to the memory.
The contributors store the evaluation in their states.
After this guessing-phase, the contributors can write the symbols $\#^\ell_j$, depending on whether the currently stored variable with its evaluation satisfies clause $C^\ell_j$.
As soon as the leader has read the complete string $\#^\ell_1 \dots \#^\ell_m$, it moves to its final state, showing that the evaluation \mbox{satisfied all clauses of $\varphi_\ell$.}

Note that parameter $\sizeC$ is of size $\bigO(n)$ and does not depend on $I$ at all.
Hence, the size-restrictions of a cross-composition are met.

%% file: content/LCR_Intractability.tex
\subsection{Intractability}
\label{Section:LCR_Intractability}

We show the $\W[1]$-hardness of $\LCR(\sizeM)$ and $\LCR(\sizeL)$.
Both proofs rely on a parameterized reduction from $\kClique$, the problem of finding a clique of size $k$ in a given graph.
This problem is known to be $\W[1]$-complete~\cite{Downey2013}.
\mbox{We state our result.}

\begin{proposition}\label{Proposition:LCR_Intractability}
	Both parameterizations, $\LCR(\sizeM)$ and $\LCR(\sizeL)$, are $\W[1]$-hard.
\end{proposition}

We first reduce $\kClique$ to $\LCR(\sizeL)$.
To this end, we construct from an instance $(G,k)$ of $\kClique$ in polynomial time an instance \mbox{$(\asms = (\Domain, \initmem, (P_L, P_C)), F_L)$} of $\LCR$ with $\sizeL = \bigO(k)$.
This meets the requirements of a parameterized reduction.

Program $\asms$ operates in three phases.
In the first phase, the leader chooses $k$ vertices of the graph and writes them to the memory.
Formally, it writes a sequence $(v_1,1). (v_2,2) \dots (v_k,k)$ where the $v_i$ are vertices of $G$. 
During this selection, the contributors non-deterministically choose to store a suggested vertex~$(v_i,i)$ in their state space.

In the second phase, the leader again writes a sequence of vertices using different symbols: $(w^\#_1,1) (w^\#_2,2) \dots (w^\#_k,k)$.
Note that the vertices $w_i$ do not have to coincide with the vertices from the first phase.
It is then the contributor's task to verify that the new sequence constitutes a clique.
To this end, for each $i$, the program does the following:
If a contributor storing $(v_i,i)$ reads the value $(w^\#_i,i)$, the computation on the contributor can only continue if $w_i = v_i$.
If a contributor storing $(v_j,j)$ with $j \neq i$ reads $(w^\#_i,i)$, the computation can only continue if $v_j \neq w_i$ and if there is an edge between $v_j$ and $w_i$.

Finally, in the third phase, we need to ensure that there was at least one contributor storing $(v_i,i)$ and that the above checks were all positive.
To this end, a contributor that has successfully gone through the second phase and stores $(v_i,i)$ writes the symbol $\#_i$ to the memory.
The leader pends to read the sequence of symbols $\#_1 \dots \#_k$.
This ensures the selection of $k$ different vertices, where each two are adjacent.

For proving $\W[1]$-hardness of $\LCR(\sizeM)$, we reuse the above construction.
However, the size of the data domain is $\abs{V} \cdot k$, where $V$ is the set of vertices of $G$.
Hence, it is not a parameterized reduction for parameter $\sizeM$.
The factor $\abs{V}$ appears since leader and contributors communicate on the pure vertices.
The main idea of the new reduction is to decrease the size of $\Domain$ by transmitting the vertices in binary.
To this end, we add binary branching trees to the contributors that decode a binary encoding.
We omit the details and refer to Appendix \ref{Appendix_LCR_Intractability}.

%% file: content/BSR.tex
\section{Bounded-Stage Reachability}
\label{Section:BSR}

The \emph{bounded-stage reachability problem} is a simultaneous reachability problem.
It asks whether all threads of a program can reach an unsafe state when restricted to $\nrs$-stage computations.
These are computations where the write permission changes $\nrs$ times.
The problem was first analyzed in \cite{Atig14} and shown to be $\NP$-complete for finite state programs.
We give matching upper and lower bounds in terms of fine-grained complexity and prove the absence of a polynomial kernel.

Let $\asms = (\Domain, \initmem, (P_i)_{i \in [1..\nrt]})$ be a program. 
A \emph{stage} is a computation in $\asms$ where only one of the threads writes.
The remaining threads are restricted to reading the memory.
An \emph{$\nrs$-stage computation} is a computation that can be split into $\nrs$ parts, each of which forming a stage.
We state the decision problem.

\begin{myproblem}
	\problemtitle{Bounded-Stage Reachability}
	\problemshort{($\BSR$)}
	\probleminput{A program $\asms = (\Domain, \initmem, (P_i)_{i \in [1..\nrt]})$, a set $C^f \subseteq C$, and $\nrs\in \Naturals$.}
	\problemquestion{Is there an $\nrs$-stage computation $c^0 \rightarrow^*_\asms c$ for some $c \in C^f$?}
\end{myproblem}

We focus on a parameterization of $\BSR$ by $\sizeP$, the maximum number of states of a thread, and $\nrt$, the number of threads.
Let it be denoted by $\BSR(\sizeP, \nrt)$.
We prove that the parameterization is $\FPT$ and present a matching lower bound.
The main result in this section is the absence of a polynomial kernel for $\BSR(\sizeP, \nrt)$.
The result is technically involved and shows what makes the problem hard.

Parameterizations of $\BSR$ involving only $\sizeM$ and $\nrs$ are intractable.
We show that $\BSR$ remains $\NP$-hard even if both, $\sizeM$ and $\nrs$, are constants.
This proves the existence of an $\FPT$-algorithm for those cases unlikely.

\input{content/BSR_States_Threads}
\input{content/BSR_Intractability}

%% file: content/BSR_States_Threads.tex
\subsection{Parameterization by Number of States and Threads}
\label{Section:BSR_States_Threads}

We first give an algorithm for $\BSR$, based on a product construction of automata.
Then, we present a lower bound under $\ETH$.
Interestingly, the lower bound shows that we cannot avoid building the product.
We conclude with proving the absence of a polynomial kernel.
As before, we cross-compose from $\kSAT{3}$ but now face the problem that two important parameters in the construction, $\sizeP$ and $\nrt$, are not allowed to depend polynomially on the number of $\kSAT{3}$-instances.

\input{content/BSR_States_Threads_Upper}
\input{content/BSR_States_Threads_Lower}
\input{content/BSR_States_Threads_Kernel}

%% file: content/BSR_States_Threads_Upper.tex
\subsubsection*{Upper Bound}
\label{Section:BSR_States_Threads_Upper}

We show that $\BSR(\sizeP, \nrt)$ is fixed-parameter tractable.
The idea is to reduce to reachability on a product automaton.
The automaton stores the configurations, the current writer, and counts up to the number of stages $\nrs$.
To this end, it has $\bigOS(\sizeP^{\nrt})$ many states.
Details can be found in Appendix \ref{Appendix_BSR_States_Threads}.

\begin{proposition}\label{Proposition:BSR_States_Threads_Upper}
	$\BSR$ can be solved in time $\bigOS(\sizeP^{\, 2 \nrt})$.
\end{proposition}

%% file: content/BSR_States_Threads_Lower.tex
\subsubsection*{Lower Bound}
\label{Section:BSR_States_Threads_Lower}

By a reduction from $\kkClique$, we show that a $2^{o(\nrt \cdot \log(\sizeP))}$-time algorithm for $\BSR$ would contradict $\ETH$.
The above algorithm is optimal.

\begin{proposition}\label{Proposition:BSR_States_Threads_Lower}
	$\BSR$ cannot be solved in time $2^{o(\nrt \cdot \log(\sizeP))}$ unless $\ETH$ fails.
\end{proposition}

The reduction constructs from an instance of $\kkClique$ an equivalent instance \mbox{$(\asms = (\Domain, \initmem,  (P_i)_{i\in[1..\nrt]}), C^f, \nrs)$} of $\BSR$.
Moreover, it keeps the parameters small.
We have that $\sizeP = \bigO(k^2)$ and $\nrt = \bigO(k)$.
As a consequence, a $2^{o(\nrt \cdot \log (\sizeP))}$-time algorithm for $\BSR$ would yield an algorithm for $\kkClique$ running in time \mbox{$2^{o(k \cdot \log (k^2))} = 2^{o(k \cdot \log(k))}$}.
But this contradicts $\ETH$.

\begin{proof}[Idea]
	For the reduction, let $V = [1..k] \times [1..k]$ be the vertices of $G$. 
	We define \mbox{$\Domain = V \cup \setcon{\initmem}$} to be the domain of the memory.
	We want the threads to communicate on the vertices of $G$.
	For each row we introduce a reader thread $P_i$ that is responsible for storing a particular vertex of the row.
	We also add one writer $P_{ch}$ that is used to steer the communication between the $P_i$.
	Our program $\asms$ is given by the tuple $(\Domain, \initmem, ( (P_i)_{i \in [1..k]}, P_{ch} ))$.
	
	Intuitively, the program proceeds in two phases.
	In the first phase, each $P_i$ non-deterministically chooses a vertex from the $i$-th row and stores it in its state space.
	This constitutes a clique candidate $(1,j_1), \dots, (k, j_k) \in V$.
	In the second phase, thread $P_{ch}$ starts to write a random vertex $(1,j_1')$ of the first row to the memory.
	The first thread $P_1$ reads $(1,j_1')$ from the memory and verifies that the read vertex is actually the one from the clique candidate.
	The computation in $P_1$ will deadlock if $j_1' \neq j_1$.
	The threads $P_i$ with $i \neq 1$ also read $(1,j_1')$ from the memory.
	They have to check whether there is an edge between the stored vertex $(i,j_i)$ and $(1,j_1')$.
	If this fails in some $P_i$, the computation in that thread will also deadlock.
	After this procedure, the writer $P_{ch}$ guesses a vertex $(2,j_2')$, writes it to the memory, and the verification steps repeat.
	In the end, after $k$ repetitions of the procedure, we can ensure that the guessed clique candidate is indeed a clique.
	Formal construction and proof are given in Appendix \ref{Appendix_BSR_States_Threads}.
	\qed
\end{proof}

%% file: content/BSR_States_Threads_Kernel.tex
\subsubsection*{Absence of a Polynomial Kernel}
\label{Section:BSR_States_Threads_Kernel_Lower}

We show that $\BSR(\sizeP, \nrt)$ does not admit a polynomial kernel.
To this end, we cross-compose $\kSAT{3}$ into $\BSR(\sizeP, \nrt)$.

\begin{theorem}\label{Theorem:BSR_States_Threads_Kernel_Lower}
	$\BSR(\sizeP, \nrt)$ does not admit a poly. kernel unless $\NP \subseteq \co\NP / \poly$.
\end{theorem}

In the present setting, coming up with a cross-composition is non-trivial.
Both parameters, $\sizeP$ and $\nrt$, are not allowed to depend polynomially on the number $I$ of given $\kSAT{3}$-instances.
Hence, we cannot construct an NFA that distinguishes the $I$ instances by branching into $I$ different directions.
This would cause a polynomial dependence of $\sizeP$ on $I$.
Furthermore, it is not possible to construct an NFA for each instance as this would cause such a dependence of $\nrt$ on $I$.
To circumvent the problems, some deeper understanding of the model is needed.

\begin{proof}[Idea]
	Let $\varphi_1, \dots, \varphi_I$ be given $\kSAT{3}$-instances, where each two are equivalent under $\polyrel$, the polynomial equivalence relation of Theorem \ref{Theorem:LCR_Leader_Mem_Kernel_Lower}.
	Then each $\varphi_\ell$ has $m$ clauses and $n$ variables $\setcon{x_1, \dots, x_n}$.
	We assume \mbox{$\varphi_\ell = C^\ell_1 \wedge \dots \wedge C^\ell_m$.}
	
	In the program that we construct, the communication is based on $4$-tuples of the form $(\ell, j, i, v)$.
	Intuitively, such a tuple transports the following information: 
	The $j$-th clause in instance $\varphi_\ell$, $C^\ell_j$, can be satisfied by variable $x_i$ with evaluation~$v$.
	Hence, our data domain is $\Domain = ([1..I] \times [1..m] \times [1..n] \times \setcon{0,1}) \cup \setcon{\initmem}$.
	
	For choosing and storing an evaluation of the $x_i$, we introduce so-called variable threads $P_{x_1}, \dots, P_{x_n}$.
	In the beginning, each $P_{x_i}$ non-deterministically chooses an evaluation for $x_i$ and stores it in its state space.
	
	We further introduce a writer $P_w$.
	During a computation, this thread guesses exactly $m$ tuples $(\ell_1, 1, i_1, v_1), \dots, (\ell_m, m, i_m, v_m)$ in order to satisfy $m$ clauses of potentially different instances.
	Each $(\ell_j, j, i_j, v_j)$ is written to the memory by $P_w$.
	All variable threads then start to read the tuple.
	If $P_{x_{i}}$ with $i \neq i_j$ reads it, then the thread will just move one state further since the suggested tuple does not affect the variable $x_i$. 
	If $P_{x_i}$ with $i = i_j$ reads the tuple, the thread will only continue its computation if $v_j$ coincides with the value that $P_{x_i}$ guessed for $x_i$ and, moreover, $x_i$ with evaluation $v_j$ satisfies clause $C^{\ell_j}_{j}$.
	
	Now suppose the writer did exactly $m$ steps while each variable thread did exactly $m+1$ steps.
	This proves the satisfiability of $m$ clauses by the chosen evaluation.
	But these clauses can be part of different instances: 
	It is not ensured that the clauses were chosen from one formula $\varphi_\ell$. 
	The major difficulty of the cross-composition lies in how to ensure exactly this.
	
	We overcome the difficulty by introducing so-called bit checkers $P_b$, where $b \in [1..\log(I)]$.
	Each $P_b$ is responsible for the $b$-th bit of $\bin(\ell)$, the binary representation of $\ell$, where $\varphi_\ell$ is the instance we want to satisfy.
	When $P_w$ writes a tuple $(\ell_1, 1, i_1, v_1)$ for the first time, each $P_b$ reads it and stores either $0$ or $1$, according to the $b$-th bit of $\bin(\ell_1)$.
	After $P_w$ has written a second tuple $(\ell_2, 2, i_2, v_2)$, the bit checker $P_b$ tests whether the $b$-th bit of $\bin(\ell_1)$ and $\bin(\ell_2)$ coincide, otherwise it will deadlock.
	This will be repeated any time $P_w$ writes a new tuple to the memory.
	
	Assume the computation does not deadlock in any of the $P_b$.
	Then we can ensure that the $b$-th bit of $\bin(\ell_j)$ with $j \in [1..m]$ never changed during the computation.
	This means that $\bin(\ell_1) = \dots = \bin(\ell_m)$.
	Hence, the writer $P_w$ has chosen clauses of just one instance $\varphi_{\ell}$.
	Moreover, the current evaluation satisfies the formula.
	Since the parameters are bounded, $\sizeP \in \bigO(m)$ and \mbox{$\nrt \in \bigO(n + \log(I))$}, the construction constitutes a proper cross-composition.
	For a formal construction and proof, we refer to Appendix \ref{Appendix_BSR_States_Threads}.
	\qed
\end{proof}

Variable threads and writer thread are needed for testing satisfiability of clauses.
The need for bit checkers comes from ensuring that all clauses stem from the same formula.
We illustrate the notion with an example.

\begin{example}
	Let four formulas $\varphi_1, \varphi_2, \varphi_3, \varphi_4$ with two clauses each be given.
	We show how the bit checkers are constructed.
	To this end, we first encode the index of the instances as binary numbers using two bits.
	The encoding is shown in Figure~\ref{Figure:BitChecker} on the right hand side.
	Note the offset by one in the encoding.
	
	We focus on the bit checker $P_{b_1}$ responsible for the first bit.
	It is illustrated in Figure \ref{Figure:BitChecker} on the left hand side.
	Note that the label $\ell = 1, \ell = 3$ refers to transitions of the form $?(\ell,j,i,v)$ with $\ell$ either $1$ or $3$ and arbitrary values for $i$,$j$, and $v$.
	On reading the first of these tuples, $P_{b_1}$ stores the first bit of $\ell$ in its state space.
	The blue marked states store that $b_1 = 0$, the red states store $b_1 = 1$.
	Then, the bit checker can only continue on reading tuples $(\ell,j,i,v)$ where the first bit of $\ell$ matches the stored bit.
	In the case of $b_1 = 0$, this means that $P_{b_1}$ can only read tuples $(\ell,j,i,v)$ with $\ell$ either $1$ or $3$.
	
	Assume the writer thread has output two tuples $(\ell_1,1,i_1,v_1)$ and $(\ell_2,2,i_2,v_2)$ and the bit checker $P_{b_1}$ has reached a \emph{last} state. 
	Since the computation did not deadlock on $P_{b_1}$, we know that the first bits of $\ell_1$ and $\ell_2$ coincide.
	If the bit checker for the second bit does not deadlock as well, we get that $\ell_1 = \ell_2$.
	Hence, the writer has chosen two clauses from one instance $\varphi_{\ell_1}$.
\end{example}

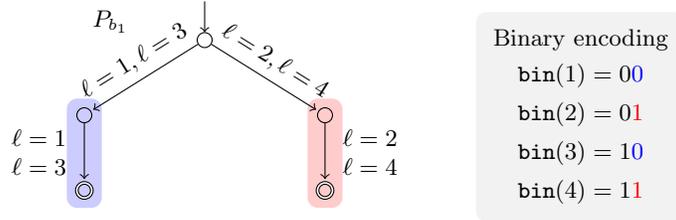
\begin{figure}
	\begin{center}
		\begin{tikzpicture}
			\node[state, initial above, initial text =, inner sep = 2pt, minimum size = 0pt] (p0) {};
			\coordinate[below of = p0] (mid);
			\node[left of = mid, node distance = 1.6cm, state, inner sep = 2pt, minimum size = 0pt] (p1) {};
			\node[right of = mid, node distance = 1.6cm, state, inner sep = 2pt, minimum size = 0pt] (p2) {};
			\node[below of = p1, state, inner sep = 2pt, minimum size = 0pt, accepting] (p3) {};
			\node[below of = p2, state, inner sep = 2pt, minimum size = 0pt, accepting] (p4) {};
			
			\path[->, shorten >= 1pt] (p0) edge 
				node[sloped, above] {$\ell = 1, \ell = 3$}
			(p1);
			\path[->, shorten >= 1pt] (p0) edge 
				node[sloped, above] {$\ell = 2, \ell = 4$}
			(p2);
			\path[->, shorten >= 1pt] (p1) edge 
				node[align = left, xshift = -0.6cm] {$\ell = 1$ \\ $\ell = 3$}
			(p3);
			\path[->, shorten >= 1pt] (p2) edge 
				node[align = right, xshift = 0.6cm] {$\ell = 2$ \\ $\ell = 4$}
			(p4);
			
			\node[left of = p0, node distance = 1.25cm, yshift = 0.25cm] {$P_{b_1}$};
			
			\node[fit = (p1)(p3)] (b0) {};
			\node[fit = (p2)(p4)] (b1) {};
			\begin{pgfonlayer}{background}
				\path[fill = blue!20, rounded corners] (b0.south west) rectangle (b0.north east);
				\path[fill = red!20, rounded corners] (b1.south west) rectangle (b1.north east);
			\end{pgfonlayer}
			
			\node[right of = p0, node distance = 5cm] (encoding) {Binary encoding};
			\node[below of = encoding, node distance = 1.25cm, align = left] (bin) {$\bin(1) = 0\textcolor{blue}{0}$ \\[0.4em] $\bin(2) = 0\textcolor{red}{1}$ \\[0.4em] $\bin(3) = 1\textcolor{blue}{0}$ \\[0.4em] $\bin(4) = 1\textcolor{red}{1}$};
			
			\node[fit = (encoding)(bin)] (enc) {};
			\begin{pgfonlayer}{background}
				\path[fill = lightgray!20, rounded corners] (enc.south west) rectangle (enc.north east);
			\end{pgfonlayer}
			
		\end{tikzpicture}
		\caption{A binary encoding (right) of the numbers $1$ up to $4$ using two bits.
			First bits are either marked blue, if they are $0$, or red if they are $1$.
			The bit checker $P_{b_1}$ (left) focuses on the first bit. 
			The label $\ell = 1, \ell = 3$ means that $P_{b_1}$ has transitions on $?(\ell,j,i,v)$ for $\ell = 1,3$ and arbitrary values for $i,j$, and $v$.
			The blue marked states store that the first bit $b_1$ is $0$.
			Red marked states store that $b_1$ is $1$.
		}
		\label{Figure:BitChecker}
	\end{center}
\end{figure}

%% file: content/BSR_Intractability.tex
\subsection{Intractability}
\label{Section:BSR_Intractability}

We show that parameterizations of $\BSR$ involving only $\nrs$ and $\sizeM$ are intractable.
To this end, we prove that $\BSR$ remains $\NP$-hard even if both parameters are constant.
This is surprising as the number of stages $\nrs$ seems to be a powerful parameter.
Introducing such a bound in simultaneous reachability lets the complexity drop from $\PSPACE$ to $\NP$.
But it is not enough \mbox{to guarantee an $\FPT$-algorithm.}

\begin{proposition}\label{Proposition:BSR_Intractability}
	$\BSR$ is $\NP$-hard even if both $\nrs$ and $\sizeM$ are constant.
\end{proposition}

The proposition implies intractability:
Assume there is an $\FPT$-algorithm $A$ for $\BSR$ running in time $f(\nrs,\sizeM) \cdot poly(\abs{x})$, where $x$ denotes the input.
Then $\BSR'$, the variant of $\BSR$ where $\nrs$ and $\sizeM$ are constant, can also be solved by $A$.
In this case, the runtime of $A$ is $\bigO(poly(\abs{x}))$ since $f(\nrs,\sizeM)$ is a constant on every instance of $\BSR'$.
But this contradicts the $\NP$-hardness of $\BSR'$.

\begin{proof}[Idea]
	We give a reduction from $\kSAT{3}$ to $\BSR$ that keeps both parameters constant.
	Let $\varphi$ be a $\kSAT{3}$-instance with $m$ clauses and variables $x_1, \dots, x_n$.
	We construct a program $\asms = (\Domain, \initmem, P_1, \dots, P_n, P_v)$ with $\sizeM = 4$ different memory symbols that can only run $1$-stage computations.
	
	The program cannot communicate on literals directly, as this would cause a blow-up in parameter $\sizeM$.
	Instead, variables and evaluations are encoded in binary in the following way.
	Let $\ell$ be a literal in $\varphi$.
	It consists of a variable $x_i$ and an evaluation $v \in \setcon{0,1}$.
	The padded binary encoding $\bin_\#(i) \in (\setcon{0,1} . \#)^{\log(n)+1}$ of $i$ is the usual binary encoding where each bit is separated by a $\#$.
	The string \mbox{$\enc(\ell) = v \# \bin_\#(i)$} encodes that variable $x_i$ has evaluation $v$.
	We need the padding symbol $\#$ to prevent the threads in $\asms$ from reading the same symbol more than once.
	Program $\asms$ communicates by passing messages of the form $\enc(\ell)$.
	To this end, we need the data domain $\Domain = \setcon{a^0,\#,0,1}$.
	
	The program contains threads $P_i, i \in [1..n]$, called variable threads.
	Initially, these threads choose an evaluation for the variables and store it:
	Each $P_i$ can branch on reading $\initmem$ and choose whether it assigns $0$ or $1$ to $x_i$.
	Then, a verifier thread $P_v$ starts to iterate over the clauses.
	For each clause $C$, it picks a literal $\ell \in C$ that should evaluate to true and writes its encoding $\enc(\ell)$ to the memory.
	Each of the $P_i$ reads $\enc(\ell)$.
	Note that reading and writing $\enc(\ell)$ needs a sequence of transitions.
	In the construction, we ensure that all the needed states and transitions are provided.
	It is the task of each $P_i$ to check whether the chosen literal $\ell$ is conform with the chosen evaluation for $x_i$.
	We distinguish two cases.
	
	(1) If $\ell$ involves a variable $x_j$ with $j \neq i$, variable thread $P_i$ just continues its computation by reading the whole string  $\enc(\ell)$.
	
	(2) If $\ell$ involves $x_i$, $P_i$ has to ensure that the stored evaluation coincides with the one sent by the verifier.
	To this end, $P_i$ can only continue its computation if the first bit in $\enc(\ell)$ shows the correct evaluation.
	Formally, there is only an outgoing path of transitions on $\enc(x_i)$ if $P_i$ stored $1$ as evaluation \mbox{and on $\enc(\neg x_i)$ if it stored $0$.}
	
	Note that each time $P_v$ picks a literal $\ell$, all $P_i$ read $\enc(\ell)$, even if the literal involves a different variable.
	This means that the $P_i$ count how many literals have been seen already.
	This is important for correctness:
	The threads will only terminate if they have read a word of fixed length and did not miss a single symbol.
	There is no loss in the communication between $P_v$ and the $P_i$.
	
	Now assume $P_v$ iterated through all $m$ clauses and none of the variable threads got stuck.
	Then, each of them read exactly $m$ encodings without running into a deadlock.
	Hence, the picked literals were all conform with the evaluation chosen by the $P_i$.
	This means that a satisfying assignment for $\varphi$ is found.
	
	During a computation of $\asms$, the verifier $P_v$ is the only thread that has write permission.
	Hence, each computation of $\asms$ consists of a single stage.
	For a formal construction, we refer to Appendix \ref{Appendix_BSR_Intractability}.
	\qed
\end{proof}

%% file: content/conclusion.tex
\section{Conclusion}
\label{Section:Conclusion}

We have studied several parameterizations of $\LCR$ and $\BSR$, two safety verification problems for shared memory concurrent programs.
In $\LCR$, a designated leader thread interacts with a number of equal contributor threads.
The task is to decide whether the leader can reach an unsafe state.
The problem $\BSR$ is a generalization of bounded context switching.
A computation gets split into stages, periods where writing is restricted to one thread.
Then, $\BSR$ asks whether all threads can reach a final state simultaneously during an $\nrs$-stage computation.

For $\LCR$, we identified the size of the data domain $\sizeM$, the size of the leader $\sizeL$ and the size of the contributors $\sizeC$ as parameters.
Our first algorithm showed that $\LCR(\sizeM, \sizeL)$ is $\FPT$.
Then we modified the algorithm to obtain a verification procedure valuable for practical instances.
The main insight was that due to a factorization along strongly connected components, the impact of $\sizeL$ can be reduced to a polynomial factor in the time complexity.
We also proved the absence of a polynomial kernel for $\LCR(\sizeM, \sizeL)$ and presented an $\ETH$-based lower bound which shows that the upper bound is a root-factor away from being optimal.

For $\LCR(\sizeC)$ we presented a dynamic programming, running in $\bigOS(2^\sizeC)$ time.
The algorithm is based on slice-wise reachability.
This reduces a reachability problem on a large graph to reachability problems on subgraphs (slices) that are solvable in polynomial time.
Moreover, we gave a tight lower bound based on $\SetCov$ and proved the absence of a polynomial kernel.

Parameterizations different from $\LCR(\sizeM,\sizeL)$ and $\LCR(\sizeC)$ were shown to be intractable.
We gave reductions from $\kClique$ and proved $\W[1]$-hardness.

The parameters of interest for $\BSR$ are the maximum size of a thread $\sizeP$ and the number of threads $\nrt$.
We have shown that a parameterization by both parameters is $\FPT$ and gave a matching lower bound.
The main contribution was to prove it unlikely that a polynomial kernel exists for $\BSR(\sizeP, \nrt)$.
The proof relies on a technically involved cross-composition that avoids a polynomial dependence of the parameters on the number of given $\kSAT{3}$-instances.

Parameterizations involving other parameters like $\nrs$ or $\sizeM$ were proven to be intractable for $\BSR$.
We gave an $\NP$-hardness proof where $\nrs$ and $\sizeM$ are constant.

\subsubsection*{Extension of the Model}
\label{Section:Extension}

In this work, the considered model for programs allows the memory to consist of a single cell.
We discuss whether the presented results carry over when the number of memory cells increases.
Having multiple memory cells is referred to as supporting \emph{global variables}.
Extending the definition of programs in Section \ref{Section:Prelim} to global variables is straightforward.

For the problem $\LCR$, allowing global variables is a rather powerful mechanism.
Let $\LCR_\Var$ denote the problem $\LCR$ where the input is a program featuring global variables.
The interesting parameters for the problem are $\sizeM$, $\sizeL$, $\sizeC$, and $\nrvars$, the number of global variables.
It turns out that $\LCR_\Var$ is $\PSPACE$-hard, even when $\sizeC$ is constant.
One can reduce the intersection emptiness problem for finite automata to $\LCR_\Var$.
The reduction makes use only of the leader, contributors are not needed.

A program $\asms$ with global variables can always be reduced to a program $\asms'$ with a single memory cell \cite{Muscholl2016}.
Roughly, the reduction constructs the leader of $\asms'$ in such a way that it can store the memory contents of $\asms$ and manage contributor accesses to the memory.
This means the new leader needs exponentially many states since there are  $\sizeM^\nrvars$ many possible memory valuations.
The domain and the contributor of $\asms'$ are of polynomial size.
In fact, we can then apply the algorithm from Section~\ref{Section:LCR_Leader_Mem} to the program $\asms'$.
The runtime depends exponentially only on the parameters $\sizeM$, $\sizeL$, and $\nrvars$.
This shows that $\LCR_\Var(\sizeM, \sizeL, \nrvars)$ is fixed-parameter tractable.
It is an interesting question whether this algorithm can be improved.
Moreover, it is open whether there are other parameterizations of $\LCR_\Var$ that have an $\FPT$-algorithm.
A closer investigation is considered future work.

For $\BSR$, allowing global variables also leads to $\PSPACE$-hardness.
The problem $\BSR_\Var$, defined similarly to $\LCR_\Var$, is $\PSPACE$-hard already for a constant number of threads.
In fact, the proof is similar to the hardness of $\LCR_\Var$ where only one thread is needed.
To obtain an algorithm for the problem, we modify the construction from Proposition \ref{Proposition:BSR_States_Threads_Upper}. 
The resulting product automaton then also maintains the values of the global variables.
This shows membership in $\PSPACE$.
But the size of the product now also depends exponentially on $\sizeM$ and $\nrvars$.
The interesting question is whether we can find an algorithm that avoids an exponential dependence on one of the parameters $\sizeP, \nrt, \sizeM$ or $\nrvars$.
It is a matter of future work to examine the precise complexity of the different parameterizations.

%% file: content/appendix.tex
\input{content/appendix_LCR_Leader_Mem}
\input{content/appendix_LCR_Contributor}
\input{content/appendix_LCR_Intractability}

\input{content/appendix_BSR_States_Threads}
\input{content/appendix_BSR_Intractability}

%% file: content/appendix_LCR_Leader_Mem.tex
\section{Proofs for Section \ref{Section:LCR_Leader_Mem}}
\label{Appendix_LCR_Leader_Mem}

We give the missing constructions and proofs for Section \ref{Section:LCR_Leader_Mem}.

\subsection*{Proof of Lemma \ref{Lemma:LCR_Leader_Mem_Upper_Correctness}}

\begin{proof}
	We will first show that for each computation leading to an unsafe state, there is a corresponding valid witness candidate.
	To this end, assume there is a $t \in \Naturals$ and a computation  $\pi = c^0 \rightarrow^*_{\asms^t} c$ with $ c \in C^f $.
	The computation $\pi$ acts on configurations but we want to work with transitions of the leader and contributor instead.
	To this end, let $\sigma$ be the sequence of transitions appearing in $\pi$.
	Without loss of generality, we assume that the last transition in $\sigma$ is due to the leader.
	In the following we show how to construct a valid witness candidate out of the sequence $\sigma$.
	It is useful to assume that each transition in $\sigma$ is uniquely identifiable.
	We use $\Pos{\tau}$ to access the position of a certain transition $\tau$ in $\sigma$.
	Hence, we have $\sigma[\Pos{\tau}] = \tau$.
	
	The first step to construct the witness candidate is collecting the first writes from $\sigma$.
	Identifying these is simple.
	One only needs to iterate over $\sigma$ and mark those write transitions of the contributors that write a symbol for the first time.
	Then, the transitions of the contributors that are not marked are removed from $\sigma$.
	Moreover, each marked transition is replaced by the symbol that it writes.
	Formally, if a marked transition is of the form $(q,!a,q')$, it is replaced by $\bar{a} \in \bar{\Domain}$.
	The resulting sequence is of the form
	\begin{align*}
		\sigma_1 \bar{c_1} \sigma_2 \bar{c_2} \dots \sigma_n \bar{c_n} \sigma_{n+1},
	\end{align*}
	where the $\bar{c_i}$ are the symbols written by the first writes and $\sigma_i$ is the sequence of transitions performed by the leader between the first writes $\bar{c_i}$ and $\bar{c_{i+1}}$.
	Note that we have $\bar{c_i} \neq \bar{c_j}$ for $i \neq j$ and $n \leq \abs{\Domain}$ since first writes can only be written once and there are at most $\abs{\Domain}$ many of them.
	
	In order to define a witness candidate in $\expr = ((Q_L.\Domainbot)^{\leq \sizeL} . \barS )^{\leq \sizeM} . Q_L$, we need to cut out loops in the $\sigma_i$ and map the resulting sequences to a word.
	We define a procedure \emph{Shrinkmap} that performs these two operations.
	As input, it takes a tuple $(\alpha,c)$ where $\alpha$ is a sequence of transitions of the leader and $c$ is a natural number.
	The procedure computes a tuple $(v,\varphi)$ where the word $v \in (Q_L.\Domainbot)^{\leq \sizeL}$ is obtained by cutting out the loops in $\alpha$ and mapping writes of a symbol $a$ to $a$ and reads of any symbol to $\bot$. 
	The function
	\begin{align*}
		\varphi : \Setcon{\tau}{\tau ~\text{a transition in}~ \alpha} \rightarrow [c ..\abs{v} + c]
	\end{align*}
	maps the transitions of the given sequence $\alpha$ into the word $v$.
	It is used to recover the sequence $\alpha$ from $v$.
	Moreover, the constant $c$ is needed to right shift the map $\varphi$.
	This gets important when we append different words obtained from applying \emph{Shrinkmap}.
	The procedure is explained in Algorithm \ref{Algorithm:Shrinkmap}.
	
	\begin{algorithm}
		\caption{\emph{Shrinkmap}}
		\label{Algorithm:Shrinkmap}
		
		\textbf{Input:} $(\alpha = \tau_1 \dots \tau_k, c)$ where $\alpha$ is a sequence of leader transitions and $c$ is a constant.\\
		\textbf{Output:} $(v,\varphi)$ with $v \in (Q_L.\Domainbot)^{\leq \sizeL}$ and $\varphi : \setcon{\tau_1, \dots, \tau_k} \rightarrow [c..\abs{v} + c]$.
		
		\begin{algorithmic}
			\STATE $i = 1$;
			\STATE $v = \varepsilon$;
			\WHILE {$i \leq k$}
			
			\STATE let $\tau_i = (q,\op, p)$;
			\IF {$\exists j : \tau_j = (q,\op',p')$}
			\STATE $\forall \ell \in [i..j-1]$, set $\varphi(\tau_\ell) = c + \abs{v} + 1$; $~\backslash\backslash$ Cutting out detected loop
			\STATE $i = j$;
			
			\ELSE
			\IF{$\op = \, !b$ for some $b \in \Domain$}
			\STATE $v = v.q.b$; $~\backslash\backslash$ $\op$ is a write of symbol $b$
			\ELSE
			\STATE $v = v.q.\bot$; $~\backslash\backslash$ $\op$ is a read or $\varepsilon$
			\ENDIF
			\STATE{$ \varphi(\tau_i) = c + \abs{v}$};
			\STATE{$i = i + 1$};
			\ENDIF
			
			\ENDWHILE
			\RETURN{(v,$\varphi$)};
		\end{algorithmic}
	\end{algorithm}
	
	We consecutively apply \emph{Shrinkmap}.
	We begin with the input $(\sigma_1,0)$ and obtain the tuple $(w_1,\varphi_1)$.
	In the $i$-th step, we run the procedure on the input $(\sigma_i,\sum_{j\in[1..i-1]} \abs{w_j})$ and get the output $(w_i,\varphi_i)$.
	We do not apply \emph{Shrinkmap} to the last sequence $\sigma_{n+1}$.
	Let this sequence be given by $\sigma_{n+1} = \tau_1 \dots \tau_t$ with transition $\tau_1 = (q, op, q')$.
	Then, the witness candidate is defined by
	\begin{align*}
		w = w_1 . \bar{c_1} . w_2 . \bar{c_2} \dots w_n . \bar{c_n} . q \in \expr.
	\end{align*}
	Moreover, we define the map $\varphi$ to be the concatenation of the $\varphi_i$.
	Formally, 
	\begin{align*}
		\varphi : \Setcon{\tau}{\tau ~\text{a transition in}~ \sigma_1 \dots \sigma_n} \rightarrow [1 ..\abs{w}]
	\end{align*}
	with $\varphi(\tau) = \varphi_i(\tau)$ if $\tau$ is a transition in $\sigma_i$.
	
	We show that the witness $w$ is valid.
	Requirement (1) is clearly satisfied since the symbols $\bar{c_i}$ written by the first writes are pairwise different.
	The second requirement is also fulfilled since we started with a proper run of the leader leading to an unsafe state $q_f \in F_L$.
	Formally, let $\prjauto{w}{Q_L \cup \Domainbot} = q_0 a_0 q_1 a_1 \dots q_m a_m q$.
	Since $\sigma_1 \dots \sigma_{n}$ is a run of the leader starting in $q_0$ and ending in $q$, we get that $q_0$ is indeed the initial state of $P_L$.
	Moreover, the transition sequence $\sigma_{n+1}$ leads from $q$ to the state $q_f$ and reading in this sequence is restricted to the symbols $\bar{c_i}$ that were provided by the first writes.
	This means there is a word $u \in (\Writes{\Domain} \cup \Reads{\barS(w)})^*$ with $q \xrightarrow{v}_L q_f$.
	
	Let $i \in [1..m]$ and consider $a_i$.
	If $a_i \in \Domain$, we know that there is a transition $(q_i,!a_i,q_{i+1})$.
	This follows from the application of \emph{Shrinkmap}.
	Similarly, if \mbox{$a_i = \bot$}, we get a transition of the form $(q_i,\varepsilon,q_{i+1})$ or $(q_i,?a,q_{i+1})$.
	In the latter case, the read symbol $a$ is provided by an earlier first write.
	This is due to the fact that the read transition appears in the computation $\sigma_1 \dots \sigma_{n}$ of the leader.
	Formally, $a \in \barS(w,\pos(a_i))$.
	
	It is left to show that Requirement (3) is satisfied.
	We show that the reads of contributors that are responsible for first writes can be embedded into the witness candidate $w$.
	To this end, consider the $i$-th first write $\bar{c_i}$ and the corresponding prefix $v.\bar{c_i}$ of $w$.
	Since $\pi$ is a computation of the system, we know there is a contributor providing $\bar{c_i}$.
	Formally, there is a computation $\rho$ on this contributor of the form $\rho = q^0_C \xrightarrow{u!c_i}_C p$.
	Let $u' = \prjauto{u}{\Reads{\Domain}}$ be the reads of $u$ and $\tau^C_1 \dots \tau^C_z$ be the read transitions along $\rho$.
	Note that $\abs{u'} = z$.
	Our goal is to define a monotonic function $\mu : [1..z] \rightarrow [1..\abs{v}]$ that maps the reads of $\rho$ into $v$.
	
	We first identify those transitions among the $\tau^C_i$ that read a value written by the leader.
	Let these be $\tau^C_{i_1}, \dots \tau^C_{i_s}$.
	Then, there are writes of the leader in $\pi$ that serve these reads.
	Let $\tau^L_{i_j}$ denote the transition of the leader that writes the symbol read in $\tau^C_{i_j}$.
	This is the transition of the leader (writing the correct symbol) that immediately precedes $\tau^C_{i_j}$.
	We set $\mu(i_j) = \varphi(\tau^L_{i_j})$.
	
	Note that this already covers two cases of Requirement (3).
	If the read is served from a write of the leader that appears in $w$, the map $\mu$ directly points to that write.
	If the corresponding write stems from a loop, the map $\mu$ points to the state in $w$ where the loop starts.
	This is due to the application of \emph{Shrinkmap}.
	When loops are cut out, the procedure ensures that $\varphi$ is changed accordingly.
	
	Let $\tau^C_{j_1}, \dots, \tau^C_{j_r}$ be the read transitions among the $\tau^C_i$ that read symbols not provided by the leader.
	We consider $\tau^C_{j_i}$.
	Let the transition read the symbol $a$.
	Then, we need to ensure that $\mu$ maps $j_i$ to a position in $w$ such that \mbox{$a \in \barS(w,\mu(j_i))$}.
	Moreover, we need to keep $\mu$ monotonic.
	The idea is to map $j_i$ either to the position of the write transition of the leader preceding $\tau^C_{j_i}$ or to the position of the last first write before $\tau^C_{j_i}$, depending on which of the two positions is larger.
	Let $\tau^L$ be the write transition of the leader that precedes $\tau^C_{j_i}$ in $\pi$.
	Moreover, let $\bar{c_h}$ be the last first write before $\tau^C_{j_i}$.
	If $\Pos{\tau^L} > \abs{w_1 . \bar{c_1} \dots w_h . \bar{c_h}}$, we set $\mu(j_i) = \varphi(\tau^L)$.
	Otherwise, we set $\mu(j_i) = \abs{w_1 . \bar{c_1} \dots w_h . \bar{c_h}}$.
	The resulting map $\mu$ is indeed monotonic and satisfies Requirement (3).
	
	For the other direction, we assume the existence of a valid witness
	\begin{align*}
		w = w_1 . \bar{c_1} . w_2 . \bar{c_2} \dots w_n. \bar{c_n} . q \in \expr.
	\end{align*}
	Our goal is to show that there is a $t \in \Naturals$ and a computation $c^0 \rightarrow^*_{\asms^t} c$ leading to a configuration $c \in C^f$.
	
	Since $w$ is valid according to Definition \ref{Definition:Validity}, we get by the first requirement that the $\bar{c_i}$ are pairwise different.
	This shows that the $\bar{c_i}$ are unique and are thus candidates for a sequence of first writes.
	By Requirement (2), we obtain a computation of the leader from the $w_i$.
	Formally, there are $\gamma_1, \dots, \gamma_n$ and $\gamma_{n+1}$ such that $q^0_L \xrightarrow{\gamma = \gamma_1 \dots \gamma_n} q$ and $q \xrightarrow{\gamma_{n+1}} q_f$ for some $q_f \in Q_L$.
	Moreover, the reads in $\gamma_i$ are restricted to $\setcon{\bar{c_1}, \dots, \bar{c}_{i-1}}$.
	
	From Requirement (3) we get for each $\bar{c_i}$ a computation of the contributor of the form $q^0_C \xrightarrow{\alpha_i !c_i} p_i$.
	We let $u_i = \prjauto{\alpha_i}{\Reads{\Domain}}$ be the reads that occur along the computation.
	Then, we also obtain a function $\mu_i$ that maps the positions of $u_i$ into the witness $w$.
	
	Before we construct the computation $c^0 \rightarrow^*_{\asms^t} c$, we need to determine the number of contributors $t$ that are involved.
	Consider a first write $\bar{c_i}$.
	Each time, $c_i$ is read, we need a contributor to provide it.
	Hence, we first give a bound $t(i)$ on how often $c_i$ needs to be provided.
	Summing up all the $t(i)$ then bounds the number of involved contributors.
	Let 
	\begin{align*}
		t(i) = \abs{w} + \abs{\gamma_{n+1}} + \sum_{j \in [1..n]} \sizeL \cdot \abs{\alpha_j}.
	\end{align*}
	Intuitively, the leader $P_L$ can read $c_i$ at most $\abs{w} + \abs{\gamma_{n+1}}$ many times when executing a loop free computation along the witness.
	The loops are taken care of separately.
	During a loop in the leader, it can read $c_i$ at most $\sizeL$ many times.
	Moreover, loops appear at most $\abs{\alpha_j}$ many times for each $j$.
	The latter is true since a contributor currently performing the computation $q^0_C \xrightarrow{\alpha_j !c_j} p_j$ for a $j$, may need the leader to run a complete loop for each step in $\alpha_j$.
	We set the total number of involved contributors to be $t = \sum_{j \in [1..n]} t(i)$.
	
	We introduce some notions needed for the construction of the computation.
	For each $i \in [1..n]$, variable $x_i$ is used to point to a position of the word $u_i$.
	Moreover, variable $x$ points to a position in the witness $w$.
	Initially, these variables are set to zero.
	The computation will involve $t$ contributors captured in the set $\funnyp$.
	Each of these provides a certain symbol $c_i$.
	We partition $\funnyp$ into sets
	\begin{align*}
	\funnyp(i) = \Setcon{P_C \in \funnyp}{P_C ~\text{provides}~ c_i}
	\end{align*} 
	of contributors that provide $c_i$ during the computation.
	Note that $\abs{\funnyp(i)} = t(i)$.
	Given a configuration $c$ of $\asms^t$, we use $c \xrightarrow{\tau}_{\funnyp(i)} c'$ to denote that all contributors in $\funnyp(i)$ execute a transition $\tau$.
	This may happen in any order.
	The definition extends to sequences of transitions.
	Moreover, we write $c \zigmoves{!c_i}_{\funnyp(i)} c'$ if exactly one contributor in $\funnyp(i)$ writes the value $c_i$.
	The corresponding contributor is then removes from $\funnyp(i)$ since it has provided $c_i$ and therefore fulfilled it's duty.
	
	We construct the computation inductively, in such a way that it maintains the following invariants.
	Roughly, these describe that there are always enough contributors in the set $\funnyp(i)$ and those can execute the computation $q^0_C \xrightarrow{\alpha_i} p'_i$ to reach the write transition of $c_i$.
	
	(1) If $w[x] = \bar{c_i}$, we need that all contributors in $\funnyp(i)$ have already executed $q^0_C \xrightarrow{\alpha_i} p'_i$ so that they can provide $c_i$ whenever it is needed.
	To this end, all pending reads in $\alpha_i$ need to be served during the computation.
	We can ensure the latter by the invariant: If $x_i \neq \abs{u_i}$ then $\mu_i(x_i) \geq x$.
	It means that whenever there are still pending reads in $\alpha_i$, the currently pending read at position $x_i$ can still be served since the current position in the witness is not further than $\mu_i(x_i)$.
	
	(2) The number of contributors in $\funnyp(i)$ needs to be large enough in order to provide $c_i$ during the ongoing computation.
	This is ensured by the invariant 
	\begin{align*}
		\abs{\funnyp(i)} \geq k + \sum_{j \in [1..n]} \sizeL \cdot k_j,
	\end{align*}
	where $k = \abs{w} + \abs{\gamma_{n+1}} - x$ and $k_j = \abs{\alpha_j} - x_j$ for $j \in [1..n]$.
	
	(3) We synchronize the contributors in the sets $\funnyp(i)$, $i \in [1..n]$.
	To this end, we demand that after each step of the construction, all contributors from a particular set $\funnyp(i)$ are currently in the same state.
	
	We elaborate on the inductive construction of the computation.
	To this end, assume that a computation $\rho = c^0 \rightarrow^*_{\asms^t} c$ was already constructed and that the variables $x, x_1, \dots, x_n$ admit values such that the invariants (1), (2), and (3) hold.
	We show how to extend $\rho$ to a computation $\rho' = c^0\rightarrow^*_{\asms^t} c \rightarrow^*_{\asms^t} c'$.
	Moreover, $\rho'$ satisfies (1), (2), and (3) along new values $x', x'_1, \dots, x'_n$ for the variables with $x' = x+1$ and $x'_i \geq x_i$ for $i \in [1..n]$.
	Note that the induction basis is simple. 
	The computation $c^0$ along with $x = x_1 = \dots = x_n = 0$ already satisfies the invariants (1), (2), and (3).
	We perform the induction step by distinguishing the following four cases:
	
	\textbf{Case 1: $w[x] = \bot$}.
	Then, the corresponding transition $\tau$ in the computation $\gamma$ of the leader is either an $\varepsilon$-transition or a read of an earlier first write.
	In the premier case, we extend the computation $\rho$ by the $\varepsilon$-transition $\tau$ and increment $x$ by $1$.
	Then, clearly invariants (2) and (3) are still satisfied.
	Moreover, Invariant (1) holds since no map $\mu_i$ points onto a position $x$ with $w[x] = \bot$.
	
	If transition $\tau$ reads the symbol $c_i$ of an earlier first write, we add two transitions to $\rho$.
	First, we add a transition $c \zigmoves{!c_i}_{\funnyp(i)} \hat{c}$ to write $c_i$ to the memory.
	Note that we have a contributor in $\funnyp(i)$ that can perform the transition due to invariants (1) and (2).
	Then, we add the read $\tau$ of the leader, resulting in
	\begin{align*}
		\rho' = c^0 \rightarrow^*_{\asms^t} c \zigmoves{!c_i}_{\funnyp(i)} \hat{c} \xrightarrow{?c_i} c',
	\end{align*}
	and increment $x$ by $1$.
	Invariant (1) is still satisfied.
	Invariant (2) also holds since one contributor is removed from $\funnyp(i)$ and $x$ is incremented by $1$.
	And since no other contributor moved, Invariant (3) is preserved, as well.
	
	\textbf{Case 2: $w[x] = a \in \Domain$}.
	This means that the corresponding transition $\tau$ in the computation $\gamma$ of the leader writes $a$ to the shared memory.
	In this case, we first append $\tau$ to $\rho$ and obtain
	\begin{align*}
		c^0 \rightarrow^*_{\asms^t} c \xrightarrow{!a} \hat{c}.
	\end{align*}
	Now we serve all contributors that need to read the value $a$ in order to reach their first write.
	Let $i \in [1..n]$ and let $P$ be a contributor in $\funnyp(i)$ with $x_i \neq \abs{u_i}$ and $\mu_i(x_i) = x$.
	This means that $P$ needs to read $a$ in order to go on with its computation.
	Hence, we extend the computation by $\hat{c} \xrightarrow{?a}_{\funnyp(i)} \hat{c}'$.
	This ensures that all contributors in $\funnyp(i)$ do the required transition and read $a$.
	Note that since Invariant (3) holds, all these contributors are in the required state to perform the transition.
	After that, we increment $x_i$ by $1$.
	When we appended the required transitions for each $i \in [1..n]$, we increment $x$ by $1$.
	
	Invariant (1) is satisfied by the new values since the maps $\mu_i$ are monotonic.
	Invariant (2) is preserved since we did not remove any of the contributors from the sets $\funnyp(i)$.
	And Invariant (3) also holds since all the contributors in $\funnyp(i)$ do the same transition or do not move at all.
	
	\textbf{Case 3: $w[x] = \bar{c_i}$}.
	We have that $x_i = \abs{u_i}$ since $\mu_i$ maps the positions of $u_i$ to positions of $w$ that occur before $\bar{c_i}$.
	Hence, by invariants (1) and (3) we get that all contributors in $\funnyp(i)$ are in the same state from which they can write $c_i$.
	The transitions that we need to add to $\rho$ in this case stem from those contributors in $\funnyp(j)$ with $j > i$ that need to read $c_i$ in order to reach their first write $c_j$.
	We serve these reads with a single contributor from $\funnyp(i)$.
	Hence, we first add a corresponding write transition, resulting in:
	\begin{align*}
		c^0 \rightarrow^*_{\asms^t} c \zigmoves{!c_i}_{\funnyp(i)} \hat{c}.
	\end{align*}
	Then, if $x_j \neq \abs{u_j}$ and $\mu_j(x_j) = x$, we add the read transitions $c \xrightarrow{?c_i}_{\funnyp(j)} \hat{c}$ to the computation and increase $x_j$ by $1$.
	After adding the transitions for each $j$, we increment $x$ by $1$.
	
	Now, Invariant (1) is satisfied since the $\mu_i$ are monotonic.
	Invariant (2) holds, since we remove only one contributor from $\funnyp(i)$ and increase $x$ (and potentially some $x_j$) by $1$.
	Invariant (3) is fulfilled since the contributors from $\funnyp(j)$ that move, all do the same transition.
	And the one contributor from $\funnyp(i)$ is removed after moving to a different state.
	
	\textbf{Case 4: $w[x] = q \in Q_L$}.
	Let $i \in [1..n]$ and suppose $x_i \neq \abs{u_i}$ $\mu_i(x_i) = x$.
	Hence, the contributors in $\funnyp(i)$ need to read either a first write $c_j \in \barS(w,x)$ that appeared before $x$ or a value that is written in a simple loop of the leader, $u_i[x_i] \in \Loop(q,\barS(w,x))$ while reading is restricted to earlier first writes.
	For the premier case, we can append the same transitions to $\rho$ as in Case 3 above.
	We focus on the latter.
	
	Assume that $u_i[x_i] = a$ is the value that the contributors in $\funnyp(i)$ need to read.
	Moreover, according to Requirement (3) of validity, the symbol is written in a loop $q \xrightarrow{\beta . !a . \beta'}_L q$ of the leader.
	The reads in $\beta$ and $\beta'$ are only from earlier first writes $\barS(w,x)$.
	Since the loop is simple, the leader can read at most $\sizeL$ first writes along it.
	Let $c_{j_1} \dots c_{j_\ell}$ be the sequence of first writes, the leader reads along the loop.
	Note that there might be repetition among the $c_{j_k}$.
	Since $x_i < \abs{u_i} \leq \abs{\alpha_i}$, there are at least $\sizeL$ many contributors in each $\funnyp(j)$ according to Invariant (2).
	Hence, we can provide enough contributors to execute the loop.
	
	We add the following transitions to $\rho$.
	The leader executes along $\beta_i$ until it needs a first write $c_{j_k}$.
	Then, we let a contributor perform the transition $\zigmoves{!c_{j_k}}_{\funnyp(j)}$, followed by the leader reading $c_{j_k}$.
	When the leader reaches the transition for writing $a$, it performs the transition, followed by the contributors in $\funnyp(i)$ reading the value: $\xrightarrow{?a}_{\funnyp(i)}$.
	In the same manner, $\beta'_i$ is processed.
	After adding the loop and contributor transitions to $\rho$, we increment $x_i$ by $1$ since we have served the read request of $\funnyp(i)$.
	Once we did this for each such $i$, we increment $x$ by $1$.
	
	Invariant (1) is satisfied since we increased the corresponding $x_i$ by $1$ and $\mu_i$ is monotonic.
	Invariant (2) holds since we use at most $\sizeL$ many contributors from $\funnyp(j)$ in a loop that writes a symbol required by $\funnyp(i)$.
	After the loop, $x_i$ is increased by $1$, preserving the inequality for $\funnyp(j)$ from Invariant (2).
	Moreover, Invariant (3) also holds since the contributors that move, all perform the same transition and contributors writing first writes are removed from $\funnyp(j)$.
	
	From the induction we get a computation $\pi' : c^0 \rightarrow^*_{\asms^t} c'$ which satisfies the invariants and where the leader arrives in state $q$.
	Now we add the transitions of $q \xrightarrow{\gamma_{n+1}} q_f$ to $\pi'$.
	Reading in $\gamma_{n+1}$ is restricted to first writes. 
	For each first write $\bar{c_i}$, by Invariant (3), we have that $\abs{S(i)} \geq \abs{\gamma_{n+1}}$ since $x \leq \abs{w}$.
	This means, that each time $c_i$ is required, we can add a contributor transition $\zigmoves{!c_i}_{\funnyp(i)}$ followed by the corresponding read of the leader.
	This way we construct a computation $\pi: c^0 \rightarrow^*_{\asms^t} c$ with $c \in C^f$.
	\qed
\end{proof}

\subsection*{Proof of Lemma \ref{Lemma:LCR_Leader_Mem_Upper_Validity}}

\begin{proof}
	Note that our complexity estimation are conservative.
	We do not assume that states of an automaton are stored in special lists which allow for faster iteration.
	
	It is clear that Property (1) can be checked in time $\bigO(\sizeL \cdot \sizeM)$.
	
	We claim that Property (2) can be tested in time $\bigO(\sizeL^3 \cdot \sizeM^2)$.
	To this end, we check for every adjacent pair of states $q,q'$ and letter $a \in \Domain$ between them if there is a transition of type $(q, !a, q') \in \delta_L$. 
	Similarly if $\bot$ is the symbol between $q$ and $q'$.
	Then we look for a transition $(q, \varepsilon, q')$ or $(q, ?c_i, q')$ for an $i$.
	Each such transition can be found in time $\abs{\delta_L} \leq \sizeL^2 \cdot \sizeM$. 
	Finally, checking whether there is a run from the last state of $w$ to a state in $F_L$, can be decided in time $\bigO(\sizeL^2)$, as it is just a reachability query on an NFA.
	
	Property (3) can be checked in time $\bigO(\sizeC^2 \cdot \sizeL^3 \cdot \sizeM^2)$.
	We reduce to reachability in a finite state automata $N$ constructed as follows.
	The states of $N$ are given by $Q_N = Q_C \times [1..\abs{v}] \cup \setcon{f}$.
	The initial state is $(q^0_C, 1)$.
	
	We set up the transition relation:
	For all $(q, !a, q') \in \delta_C$, we add the transitions $(q, i) \rightarrow_N (q', i)$. 
	For each read transitions $(q, ?a, q') \in \delta_C$ and state $(q,i)$, we add the transition $(q,i) \rightarrow_N (q',i)$ if one of the following three options hold:
	If $v[i] = a$.
	If $v[i] = p \in Q_L$ and $a \in \Loop(p, \barS(w,i))$.
	If $a \in \barS(w,i)$.
	Further, we add $(q, i) \rightarrow_N (q, i+1)$.
	Finally, for all states $(q, i)$ in $N$ such that $(q, !c_j, q') \in \delta_C$, we add the transition $(q,i) \rightarrow_N f$.
	This ends the computation since $c_j$ was written.
	
	Now we have that Property (3) is satisfied if and only if there is a computation $(q^0_C, 1) \rightarrow^*_N f$.
	The construction of $N$ is the dominant factor and takes time $\bigO(\sizeC^2 \cdot \sizeL^3 \cdot \sizeM^2)$.
	
	If we now combine the three results, we get that validity can be tested in time $\bigO( \sizeL^3 \cdot \sizeM^{\, 2} \cdot \sizeC^{\, 2} \, )$.
	\qed
\end{proof}

\subsection*{Formal Definition of SCC-witness Candidates and Validity}
 
 We give a formal definition of \emph{SCC-witness candidates}.
 Let \mbox{$r = \bar{c}_1 \dots \bar{c}_\ell \in \validw$.}
 We denote by $\sccexpr(r)$ the set of all strings in 
 \begin{align*}
	 ( \scc(\prjauto{P_L}{0}) . ( \Domain \cup \setcon{\bot} ))^{k_0} . \bar{c}_1 \dots \bar{c}_\ell . (\scc(\prjauto{P_L}{\ell}) . ( \Domain \cup \setcon{\bot} ))^{k_\ell} . (\scc(\prjauto{P_L}{\ell})
 \end{align*} 
 such that $\sum_{i \in [0..\ell]}k_i \leq d-1$. 
 The set of \emph{SCC-witness candidates} is the union $\sccexpr = \bigcup_{r \in \validw} \sccexpr(r)$.
 A SCC-witness candidate $w \in \sccexpr$ is called \emph{valid} if it satisfies the following  properties:
 \begin{enumerate}
\item  The sequence in $w$ without the barred symbol induces a valid run in the leader. For this we need to find appropriate entry and exit states in each $\scc$ such that the exit state of each $\scc$ is connected to its adjacent $\scc$ through a transition.  Let $v = \smscc_1^1a^1_1 \dots \smscc_{k_1}^1a_{k_1}^1  \smscc_1^2 a_1^2 \dots  \smscc_{k_2}^2 a_{k_2}^2  \discretionary{}{}{} \dots  \smscc^1_\ell $ be the sequence obtained by projecting out the barred symbols. Further for any symbol $\alpha$ appearing in $v$, let $\pos(\alpha)$ denote the position of $\alpha$ in $v$. For any $i \in [1.. |v|]$, We will also use $\bar{\Domain}(v,i)$ to refer the the set of all barred symbols appearing before the position $i$: $\bar{\Domain}(v,i) = \{ a \mid \bar{a} \text{ appears in } v[1..i] \} $. Now corresponding to any sub-sequence of $v$ of the form $\smscc_1 a \smscc_2$ that appears in $v$, one of the following is true.
\begin{itemize}
\item If $a \in \Domain$, then we can find states $q \in \smscc_1$, $q' \in \smscc_2$ such that there is a transition of the form $q \moves{a} q'$ in $P_L$.

\item If $a = \bot$ then,  then we can find states $q \in \smscc_1$, $q' \in \smscc_2$ such that there is a transition of the form $q \moves{\varepsilon / c?} q'$ in $P_L$ for some $c \in \bar{\Domain}(v, \pos(a))$.

\end{itemize}
We also require that $q^0_L \in \smscc^1_1$. Finally  we require that from any state $q$ in the final scc $\smscc^1_\ell$, a run to the final state only involving writes, internal or a read of the barred symbol occurring in $w$. That is a run of the form $q \moves{\sigma} q_f$ such that $\sigma \in (\Writes{\Domain} \cup \Reads{ \bar{\Domain} (v, |v| )})^* $.

\item We can construct supportive computations on the contributors.
For each prefix $v\bar{a}$ of $w$ with $\bar{a} \in \barS$ there is a computation $q^0_C \xrightarrow{u !a}_C q$ on $P_C$ to some $q \in Q_C$ such that the reads within $u$ can be obtained from $w$. Formally, let $u' = \prjauto{u}{\Reads{\Domain}}$. Then there is an embedding of $u'$ into $v$, a map $\mu : [1..\abs{u'}] \rightarrow [1..\abs{v}]$ with $\mu(i) \leq \mu(j)$ for $i < j$ and $v[\mu(i)] \notin \barS \cup \setcon{\bot}$. Hence, $u'$ is only mapped to elements from $ \biguplus_{i \in [0..\ell]} \scc(\prjauto{P_L}{i}) \cup \Domain$. Let $u'[i] = ?a$ with $a \in \Domain$. Then either $v[\mu(i)] = a$, which corresponds to a write of $a$ by $P_L$, or $v[\mu(i)] = \smscc \in  \scc(\prjauto{P_L}{i})$ for some $i \in [0..\ell]$.
In the latter case, we have $a \in \Domain(\smscc)$, this corresponds to the letter being a write by the leader through the scc or write of a value by a contributor that was already seen.

\item Let $w$ be of the form $v_1 \bar{c}_1 v_2 \bar{c}_2 \dots v_\ell \bar{c}_\ell \smscc^1_\ell$, then the scc dont repeat in each of $v_1,v_2 \dots v_\ell$.

\end{enumerate}
  
 We refer to the above Properties as \emph{(SCC) validity}.
 Instead of stating a characterization of computations in terms of SCC-witnesses directly, we relate the SCC-witnesses to the  witnesses as defined in Section \ref{Section:LCR_Leader_Mem}.
 
 \begin{lemma}\label{Lemma:PRSCCCorrectness}
 	There is a valid SSC-witness candidate in $\sccexpr$ if and only if there is a valid witness candidate in $\expr$.
 \end{lemma}
 
\begin{proof}
	First we will prove the $\Rightarrow$ direction. For this, we will assume a valid witness string $w \in \expr$. Further let $w = v_0 \bar{c}_1 v_1 \bar{c}_2 \dots v_{k-1} \bar{c}_k q$. Now consider the decomposition ($v'_i = \smscc^i_1 \dots \smscc^i_{k_i}$) of each $v_i$ according to their $\scc$ in $ \prjauto{P_L}{i} $. This we do by taking the maximum subsequence and replace it with the $\scc$ to which it belongs. We can be sure that none of the $\scc$ in this sequence repeats (otherwise they will already form a bigger $\scc$). Secondly notice that such a $\scc$ sequence will be a path in the $\scc$ graph $\sccg(P_L, \bar{c}_1\bar{c}_2 \dots \bar{c}_k) $. This implies that the decomposition thus obtained is a scc-witness. However we need it to be a valid scc-witness. The fact that such a scc-witness satisfies property-1 of the scc-validity property follows from the fact that we started with a valid-witness that satisfies property-2 of the validity-property, this already provides us with the requires states in the adjacent scc and a transition relation between them.
	
	Proving  property-2 of the scc-validity is slightly more complicated. For this, we need to construct a run in the contributor for each prefix that ends in a barred symbol and show that there is an appropriate mapping from this run to the scc-witness string. But notice that since we started with a valid-witness string, we are guaranteed a computation in the contributor and a mapping from such a run into the witness string. For this, let us fix one such prefix to be $\alpha' \bar{c}_j$. Let the corresponding prefix of $w$ be $\alpha \bar{c}_j$. Let $q^0_c \moves{u\cdot c!} q $ be such a run and  $\seqmap: [1..|u'|] \mapsto [1..|\alpha'|]$ (where $u'$ is obtained from $u$ by deleting all the values that do not correspond to a read transition) be the corresponding mapping. We construct from this mapping, another mapping $ \seqmap': [1..|u'|] \mapsto [1..|\alpha|]$. Observe that $\alpha = v_0\bar{c}_1 v_1 \bar{c}_2 v_2 \dots \bar{c}_j$, let $\alpha' =  v'_0\bar{c}_1 v'_1 \bar{c}_2 v'_2 \dots \bar{c}_j$, where $v'_i$ are the corresponding decomposition of $v_i$ into $\scc$. The required mapping $\mu'$ is constructed as follows. Suppose the original mapping $\mu$ mapped a value in $[1..|u|]$ to a state or letter that got replaced by an $\scc$ then we let $\mu'$ to map such a position to the $\scc$ that replaced the state. If the original mapping mapped the position in $[1..|u'|]$ to a letter that survied, then we also let $\mu'$ to do the same. That is, suppose for some $v_i = q_1a_1 \dots q_ka_k$ if its corresponding decomposition be $\smscc_1a_{i_1} \dots \smscc_ja_k$ and $\mu(i) = j$, where $j$ is a position into the string $v_i$. If $j$ points to one of $a_{i_1} \dots a_k$ (the letters that survived the decomposition), then we let $\mu'(i) = j'$, where $j'$ is the position of such a letter in $\alpha'$. If $j$ points to a state or a letter in $\Domain$ that was decomposed into $\smscc_i$, then $\mu'(i)$ gets the value of the position of $\smscc_i$ in $\alpha'$. The correctness of such a mapping follows from the following reasoning. Clearly $\mu'$ thus constructed is monotonic since $\mu$ already was. Suppose $\mu'$ maps a position in $u'$ to a letter, such a letter is guaranteed to be the same since it was the same for the mapping $\mu$. Suppose $\mu$ mapped a position in $u'$ to a letter $a$ that got decomposed into  $\smscc_i$, then clearly $a \in \Domain(\smscc_i)$. Suppose  $\mu$ mapped a position $i$ in $u'$ to a state $q$ that got decomposed into  $\smscc_i$, we observe that for any $a \in \Loop(q, \bar{\Domain}(w,\mu(i)))$, we also have $a \in \Domain(\smscc_i)$. With this, we get the result for one direction of the proof.

	\vspace{5pt}
	
	For the ($\Leftarrow$) direction, we will assume a valid scc-witness $w$ and show how to construct a valid witness from this. Let $w$ be of the form
	\begin{align*}
		\smscc_1^1a^1_1 \dots \smscc_{k_1}^1a_{k_1}^1 \bar{c}_1 \smscc_1^2a_1^2 \dots \smscc_{k_2}^2 a_{k_2}^2\bar{c}_2 \dots \bar{c}_\ell \smscc^1_\ell.
	\end{align*}
	From the scc-validity property-1, we have $q^0_L \in \smscc^1_1$. 
	Further for every subsequence of the form $\smscc_1 a \smscc_2$, we have the states $q_1 \in \smscc_1$ and $q_2 \in \smscc_2$ such that there is a transition of the form $q_1 \moves{b} q_2$, where $b$ is either $\varepsilon/ c!$ or in $a?$ depending on the type of $a$. We will call the state $q_1$ the exit state of $\smscc_1$ and $q_2$ the entry state of $\smscc_2$. 
	From this,  corresponding to each $\smscc^i_j$, there are entry and exit states $en^i_j, ex^i_j$.
	 In  each such scc, corresponding to its entry and exit states, there is also a shortest path between them. 
	 That is, corresponding to scc of the form $\smscc^i_j \in \scc(\prjauto{P_L}{i})$ there is a shortest path of the form $en^i_j \moves{a_1}_{\prjauto{P_L}{i}}q^{(i,j)}_1 \moves{a_2}_{\prjauto{P_L}{i}} q^{(i,j)}_2 \dots \moves{a_k}_{\prjauto{P_L}{i}} q^{(i,j)}_k \moves{a_{k+1}}_{\prjauto{P_L}{i}} ex^i_j$. To get the valid witness string, we replace each scc in the valid scc-witness by its shortest path as follows. There is one exception to this replacement, the final scc occuring in the scc-witness is simply replaced by its entry state. The scc $\smscc^i_j$ is replaced by a string of the form $en^i_j b_1 q_1^{(i,j)} b_2 q_2^{(i,j)} \dots   ex^i_j$, where $b_i = a$ if $a_i = a!$ and $b_i = \bot$ otherwise. Let $w' = v_0\bar{c}_1 v_1 \bar{c}_2 \dots \bar{c}_\ell q_\ell$ be the string thus obtained. Notice that in each $v_i$, none of the states repeat, this is because none of the $\scc$ repeats in $w$ and $\scc$ partitions the state space. We still need to show that the string thus obtained satisfies the validity properties.
	
	Property-1 of validity is satisfied since all the barred symbols in valid scc-witness do not repeat. Property-2 is satisfied because between any adjacent $\scc$ in $w$, there is a valid move by definition. Further the way we constructed the valid witness, we replaced each $\smscc$ in the scc-witness by a valid shortest path in that scc. Finally, by definition there is  also a run to final state  from any state in the final scc .
	
	Finally we need to prove that property-3 of validity holds. For this, we need to show existence of  an appropriate computation in the contributor along with a mapping to the witness string for every prefix of $w' $ such that the prefix ends in a barred symbol. Let us fix one such prefix to be $\sigma = v_0 \bar{c}_1 \dots v_j \bar{c}_{j+1} $, let the corresponding prefix in $w$ be $\sigma' = v'_0 \bar{c}_1 \dots v'_j \bar{c}_{j+1} $.  From the scc-validity property, we already have a computation of the form $q^0_c \moves{u\cdot c_{j+1}!} q $ and a mapping $\seqmap$ from positions of $u'$ ($u'$ is the projection of $u$ onto only read operation) to positions in $\sigma'$. We provide a mapping from the positions in $u'$ to positions in $\sigma$ as follows. If $\seqmap$ maps a position to a symbol in $\Domain$, then $\seqmap'$ maps such a position to the corresponding position of the same letter in $\sigma'$. Otherwise, the position is mapped to the first state of the shortest path of the corresponding $\scc$. If a letter is in $\Domain(\smscc)$, then one can always find a loop that visits such a letter. Hence such a mapping satisfies the validity property. \qed
\end{proof}
 
 As above, the algorithm iterates over all SCC-witness candidates and each is tested to be valid. The validity check can be performed in polynomial time, as shown in the below lemma.
 
 \begin{lemma} \label{Lemma:PRSCCValid}
 	Validity of $w \in \sccexpr$ can be checked in time $\bigO(\sizeL^2 \cdot \sizeM \cdot \sizeC^{\, 2} \cdot \sld^2)$.
 \end{lemma}
 
 \begin{proof}
 
Checking whether a SCC-witness candidate $w$ satisfies Property (1) can be tested in time $\bigO(\sld  \cdot \sizeL^2 \cdot \sizeM )$.
For this, between any two SCC in $w$ of the form the  $\smscc_1 a \smscc_2$, we need to check if there are states $q_1 \in \smscc_1$ and $q_2 \in \smscc_2$ such that $q_1 \rightarrow q_2$. This can easily be done in time $\bigO(|\delta_L|) \leq \bigO(\sizeL^2 \cdot \sizeM)$.  This operation, we need to perform between every pair of adjacent SCC, there are at most $d$ many such adjacent pairs. Hence overall the time required to check whether there are proper transitions between each scc is $ \bigO(\sld  \cdot \sizeL^2 \cdot \sizeM ) $ Checking whether there is a path from the final SCC to the final state reduces to reachability in an automata and can be easily done in time $\bigO(\sizeL^2)$.

Property (2) can be tested in time $\bigO(\sld^2 \cdot \sizeL^2 \cdot \sizeM \cdot \sizeC^2 )$.
This can be achieved by reducing the problem to reachability in a NFA. The idea is similar to the one we saw in proof of lemma-\ref{Lemma:LCR_Leader_Mem_Upper_Validity}. The automaton we construct will have the states of the contributor and an index into the given witness string. 

For constructing this automaton, we need to check for any given $\smscc$ and a letter $a \in \Domain$ , whether $a \in \Domain(\smscc)$. This can easily be achieved in time $\bigO( \sizeL^2)$ by reducing it to reachability problem in a graph. We add a transition of the form $(c_1,i) \moves{} (c_2,i)$ if there is a transition of the form $c_1 \moves{a?}_{P_C} c_2$, $i$ refers to $\smscc$ in $w$ and $a \in \smscc$. So for each $i$, we go through each transition in $\delta_C$ and perform the scc-check. This takes time $\bigO(\sld \cdot \sizeC^2 \cdot \sizeM \cdot \sizeL^2)$.

Finally, the size of such an automaton is simply $\bigO(\sizeC \times \sld)$ and we need to perform a reachability check which is quadratic and hence in $\bigO (\sizeC^2 \times \sld^2)$ . From this, we get the required complexity.

Finally, it is obvious that Property (3) can be checked in time $\bigO(\sld)$.
This completes the proof.
\qed
\end{proof}
 
\subsection*{Comparison between introduced Methods}

Compared to the formerly introduced witnesses, there are less candidates to test.
 We have that $\WitSCC(\scccnt, \sizeM, \sld) \leq \Wit(\sizeL, \sizeM)$.

\begin{lemma}
	 $\WitSCC(\scccnt, \sizeM, \sld) \leq \Wit(\sizeL, \sizeM)$.
\end{lemma}
 
\begin{proof}
	For this, we show how to construct an injective function from every string in $\WitSCC(\scccnt, \sizeM, \sld)$ to strings in $\Wit(\sizeL, \sizeM)$. This would already give us the result. The idea here is to represent each scc by an unique state present in that scc. For this, we will assume an arbitrary ordering on the states of the leader $Q_L$. Recall that the scc  partitions  the states space of the leader, hence for any $\smscc_1,\smscc_2 \in \scc(\prjauto{P_L}{i})$, for some $i$ we have $\smscc_1 \cap \smscc_2 = \emptyset$. Now given any string $w = \smscc_1^1a^1_1 \dots \smscc_{k_1}^1a_{k_1}^1 \bar{c}_1 \smscc_1^2 a_1^2 \dots  \smscc_{k_2}^2 a_{k_2}^2 \bar{c}_2 \discretionary{}{}{} \dots \bar{c}_\ell \smscc^1_\ell   \in \WitSCC(\scccnt, \sizeM, \sld)$, let $f(w)=v$ be  obtained by replacing each $\smscc_j^i$ by the minimum state $q_j^i$ in that scc. Since no two scc share a state, clearly $v \in \Wit(\sizeL, \sizeM)$. Quite clearly such a mapping is injective since each of the replacing states uniquely represents the scc.
	\qed
\end{proof}

\subsection*{A Bound on the Number of SCC-witness Candidates}

We will here show that $\WitSCC(\scccnt, \sizeM, \sld) \leq (\scccnt \cdot (\sizeM + 1))^{\sld} \cdot \sizeM^{\sizeM} \cdot 2^{\sizeM + \sld}$. For this, we note that for any $w \in \WitSCC(\sizeL, \sizeM, \sld)$, the size of such a string is bounded by $2\sld + \sizeM$. Further if we remove the last appearing $\scc$ from $w$, all other $\scc$ appearing there has an letter from $\Domain \cup \{ \bot \}$ appearing next to it. Hence if we consider such a pair as a single cell and the symbols from $\barS$ as a single cell, there are at-most $\sld + \sizeM$ many cells. Now if we fix the positions and the value of the barred symbols in the sequence of cells, then each of the unoccupied cells ($\sld$ many of them) have at-most $\scccnt \cdot (\sizeM + 1 )$ choices. Hence there are $(\scccnt \cdot (\sizeM + 1 ))^\sld$ many such strings. Now for any fixed $v \in \validw$, there are ${\sizeM + \sld \choose \sizeM }$ many ways to pick $\sizeM$ positions from $\sizeM + \sld$ positions. This is upper bounded by $2^{\sizeM + \sld}$. Since we have $\sizeM^\sizeM$ many strings in $\validw$, we get that $\WitSCC(\scccnt, \sizeM, \sld) \leq (\scccnt \cdot (\sizeM + 1))^{\sld} \cdot \sizeM^{\sizeM} \cdot 2^{\sizeM + \sld}$.

\subsection*{Formal Construction and Proof of Proposition \ref{Proposition:LCR_Leader_Mem_Lower}}

The domain of the memory is $\Domain = \Setcon{\row(i), \col(i), \#_i}{ i \in [1..k]} \cup \setcon{\initmem}$.
The contributor threads are defined by $P_C = (\Ops{\Domain}, Q_C, q^0_C,  \delta_C)$ with set of states $Q_C = \Setcon{ q^{(r,\ell)}_{(i,j)}, q^{(c,\ell')}_{(i,j)} }{ i,j, \ell \in [1..k], \ell' \in [0..k] } \cup \setcon{q^0_C, q^f_C}$.
Intuitively, we use the states $q^{(r, \ell)}_{(i,j)}$, $q^{(c,\ell')}_{(i,j)}$ to indicate that the contributor has chosen $(i,j)$ to store and to count the number of vertices that the contributor has read so far.
More precise, the state $q^{(r, \ell)}_{(i,j)}$ reflects that the last symbol read was $\row(\ell)$, the row of the $\ell$-th vertex.
The state $q^{(c,\ell)}_{(i,j)}$ indicates that the last read symbol was the column-symbol belonging to the $\ell$-th vertex.
Note that this can be any column and thus different from $\col(\ell)$.

The transition relation $\delta_C$ is defined by the following rules.
We have a rule to choose a vertex: $q^0_C \xrightarrow{? \initmem}_C q^{(c,0)}_{(i,j)}$ for any $i,j \in [1..k]$.
To read the $\ell$-th row-symbol, we have $q^{(c,\ell-1)}_{(i,j)} \xrightarrow{? \row(\ell)}_C q^{(r,\ell)}_{(i,j)}$ for any $\ell \in [1..k]$.
For reading the $\ell$-th column-symbol, we get the transition $q^{(r,\ell)}_{(i,j)} \xrightarrow{\col(j')}_C q^{(c,\ell)}_{(i,j)}$, but only if one of the following is satisfied:
(1) We have that $i \neq \ell$ and there is an edge between $(\ell,j')$ and $(i,j)$ in $G$.
Intuitively, the contributor stores a vertex $(i,j)$ from a row, different than $\ell$.
But then it can only continue its computation if $(i,j)$ and $(\ell, j')$ share an edge.
(2) We have that $i = \ell$ and $j' = j$.
This means that the contributor stores the vertex that it has read.
Note that with this, we rule out all contributors storing other vertices from the $i$-th row.

To end the computation in a contributor, we get the rule $q^{(c,k)}_{(i,j)} \xrightarrow{! \#_i}_C q^f_C$ for any $i,j \in [1..k]$.
The rule writes $\#_i$ to the memory, reflecting the \emph{correctness} of the clique in row $i$.

The leader $P_L$ is given by the tuple $P_L = (\Ops{\Domain}, Q_L, q^0_L, \delta_L)$ with set of states $Q_L = \Setcon{q^{(r,i)}_L, q^{(c,i)}_L, q^{(\#,i)}_L}{ i \in [1..k]} \cup \setcon{q^0_L}$.
Unlike for the contributor, the state $q^{(r,i)}_L$ indicates that the last written symbol was $\row(i)$, the row of the $i$-th vertex guessed by $P_L$.
State $q^{(c,i)}_L$ reflects that the last written symbol was the column-symbol belonging to the $i$-th vertex.
The remaining states $q^{(\#,i)}_L$ are used to receive the verification symbols and count up to $k$.

The transition relation $\delta_L$ is described by the below rules.
For transmitting the row-symbols we have $q^{(c, i-1)}_L \xrightarrow{! row(i)}_L q^{(r,i)}_L$ for $i \in [1..k]$.
Note that we identify $q^0_L$ as $q^{(c,0)}_L$.
For writing the column-symbols we have the rules $q^{(r,i)}_L \xrightarrow{! col(j)}_L q^{(c,i)}_L$, where $i,j \in [1..k]$.
After sending $k$ row- and column-symbols, we receive the symbols $\#_i$ via the following transitions:
$q^{(\#,i-1)}_L \xrightarrow{? \#_i}_L q^{(\#,i)}_L$ for $i \in [1..k]$.
Here, we denote by $q^{(\#,0)}_L$ the state $q^{(c,k)}_L$.

We further set $F_L = \setcon{q^{(\#,k)}_L}$, the state after receiving all $\#$-symbols.
Then our program is defined to be $\asms = (\Domain, \initmem, (P_L, P_C))$.
Note that the parameters are $\sizeM = \sizeL = 3k+1$.
The correctness of the construction is proven in the following lemma.

\begin{lemma}
	There is a $t \in \Naturals$ so that $c^0\rightarrow^*_{\asms^t}c$ with $c\in C^f$ if and only if there is a clique of size $k$ in $G$ with one vertex from each row.
\end{lemma}
\begin{proof}
	First assume that $G$ contains the desired clique and let $(1,j_1), \dots, (k,j_k)$ be its vertices.
	For $t = k$ we construct a computation from $c^0$ to a configuration in $C^f$.
	We have $k$ copies of the contributor $P_C$ in the system $\asms^t$.
	We shall denote them by $P_1, \dots, P_k$.
	
	The computation starts with $P_i$ choosing the vertex $(i,j_i)$ to store, it performs the step $q^0_{C_i} \xrightarrow{?\initmem}_i q^{(c,0)}_{(i,j_i)}$.
	Hence, we get as a computation on $\asms^t$: 
	\begin{align*}
		c^0 \rightarrow^*_{\asms^t} (q^0_L, q^{(c,0)}_{(1,j_1)}, \dots, q^{(c,0)}_{(k,j_k)}, \initmem) = c_0.
	\end{align*}
	Then $P_L$ writes the symbol $\row(1)$ and each contributor reads it.
	We get 
	\begin{align*}
		c_0 \rightarrow^*_{\asms^t} (q^{(r,1)}_L, q^{(r,1)}_{(1,j_1)}, \dots, q^{(r,1)}_{(k,j_k)}, \row(1)) = c_{(r,1)}.
	\end{align*}
	After transmitting the first row to all contributors, $P_L$ then communicates the first column by writing $\col(j_1)$ to the memory.
	Again, each contributor reads it.
	Note that $P_1$ can read $\col(j_1)$ since it stores exactly the vertex $(1,j_1)$.
	A contributor $P_i$ with $i \neq k$ can also read $\col(j_1)$ since $P_i$ stores $(i,j_i)$, a vertex which shares an edge with $(1,j_1)$ due to the clique-assumption.
	Hence, we get the computation $c_{(r,1)} \rightarrow^*_{\asms^t} (q^{(c,1)}_L, q^{(c,1)}_{(1,j_1)}, \dots, q^{(c,1)}_{(k,j_k)}, \col(j_1)) = c_1$.
	
	Similarly, we can construct a computation leading to a configuration $c_2$.
	By iterating this process, we get $c^0 \rightarrow^*_{\asms^t} c_k = (q^{(c,k)}_L, q^{(c,k)}_{(1,j_1)}, \dots, q^{(c,k)}_{(k,j_k)}, \col(j_k))$.
	Then each contributor $P_i$ can write the symbol $\#_i$.
	This is done in ascending order:
	First, $P_1$ writes $\#_1$ and $P_L$ reads it.
	Then, it is $P_2$'s turn and it writes $\#_2$.
	Again, the leader reads the symbol.
	After $k$ rounds, we reach the configuration $(q^{(\#,k)}_L, q^f_{C_1}, \dots, q^f_{C_k}, \#_k)$ which lies in $C^f$.
	
	Now let a $t \in \Naturals$ together with a computation $\rho = c^0 \rightarrow^*_{\asms^t} c$ with $c$ from $C^f$ be given.
	Let $\rho_L$ be the subcomputation of $\rho$ carried out by $P_L$.
	Technically, the projection of $\rho$ to $P_L$.
	Then $\rho_L$ has the form $\rho_L = \rho^1_L . \rho^2_L$ with
	\begin{align*}
	\rho^1_L &= q^0_L \xrightarrow{!\row(1)}_L q^{(r,1)}_L \xrightarrow{!\col(j_1)}_L q^{(c,1)}_L \xrightarrow{!\row(2)}_L \dots \xrightarrow{!\col(j_k)}_L q^{(c,k)}_L \text{ and} \\
	\rho^2_L &= q^{(c,k)}_L \xrightarrow{?\#_1}_L q^{(\#,1)}_L \xrightarrow{?\#_2}_L \dots \xrightarrow{?\#_k}_L q^{(\#,k)}_L.
	\end{align*}
	We show that the vertices $(1,j_1), \dots, (k,j_k)$ form a clique in $G$.
	
	Since in $\rho^2_L$, the leader is able to read the symbols $\#_1$ up to $\#_k$, there must be at least $k$ contributors writing them.
	Due to the structure of $P_C$, it is not possible to write different $\#_i$ symbols.
	Hence, we get one contributor for one symbol and thus, $t \geq k$.
	
	Let $P_{C_i}$ be a contributor writing $\#_i$.
	Then $P_{C_i}$ stores the vertex $(i,j_i)$, it performs the initial move $q^0_{C_i} \xrightarrow{?\initmem}_{i} q^{(c,0)}_{(i,j_i)}$.
	Assume $P_{C_i}$ stores the vertex $(i',j')$.
	Since the thread writes $\#_i$ in the end, we get $i' = i$ due to the structure of $P_{C_i}$.
	During the computation $\rho$, the thread performs the step $q^{(r,i)}_{(i,j')} \xrightarrow{?\col(j_i)}_{i} q^{(c,i)}_{(i,j')}$ since $P_L$ writes the symbol $\col(j_i)$ to the memory and the computation on $P_{C_i}$ does not deadlock.
	Note that we use the following: 
	The leader writes $\row(i)$ before $\col(j_i)$.
	This ensures that $\col(j_i)$ is indeed the column of the $i$-th transmitted vertex and the above transition is correct.
	However, the contributor $P_{C_i}$ can only do the transition if $j' = j_i$.
	Thus, we get that $(i',j') = (i,j_i)$.
	
	Let $P_{(i,j_i)}$ denote a contributor that writes $\#_i$ during $\rho$.
	Since the contributor $P_{(i,j_i)}$ stores the vertex $(i,j_i)$, the leader $P_L$ has written $\row(i)$ and $\col(j_i)$ to the memory.
	Now let $P_{(i',j_{i'})}$ be another contributor with $i' \neq i$.
	Then also this thread needs to perform the step $q^{(r,i)}_{(i', j_{i'})} \xrightarrow{?\col(j_i)}_i q^{(c,i)}_{(i',j_{i'})}$ since the computation does not end on $P_{(i',j_{i'})}$ at this point.
	But by definition, the transition can only be carried out if there is an edge between $(i,j_i)$ and $(i',j_{i'})$.
	Hence, each two vertices of $(1,j_1), \dots, (k,j_k)$ share an edge.
	\qed
\end{proof}

\subsection*{Formal Construction and Proof of Theorem \ref{Theorem:LCR_Leader_Mem_Kernel_Lower}}

We define the polynomial equivalence relation $\polyrel$ in more detail.
Assume some encoding of Boolean formulas over some finite alphabet $\Gamma$.
Let $\validforms \subseteq \Gamma^*$ be the encodings of proper $\kSAT{3}$-instances.
We say that two encodings $\varphi, \psi \in \Gamma^*$ are equivalent under $\polyrel$ if either $\varphi, \psi \in \validforms$ and the formulas have the same number of clauses and variables or if $\varphi, \psi$ are both not in $\validforms$.
Then $\polyrel$ is a polynomial equivalence relation.

We define $\Domain$ to be the union of the sets \mbox{$\Setcon{(u,\ell)}{u \in \setcon{0,1}, \ell \in [1..\log(I)]}$}, $\Setcon{(x_i, v)}{i \in [1..n], v \in \setcon{0,1}}$, $\Setcon{\#_j}{ j \in [1..m]}$, and $\setcon{\initmem}$.
Thus, we have that $\sizeM = \bigO(\log(I) + n + m)$, which does not exceed the \mbox{bounds of a cross-composition.}

The leader is defined by the tuple $P_L = (\Ops{\Domain}, Q_L, q^{(b,0)}_L, \delta_L)$ with set of states $Q_L = \Setcon{q^{(b,\ell)}_L, q^{(x,i)}_L, q^{(\#,j)}_L}{\ell \in [0..\log(I)], i \in [1..n], j \in [1..m]}$.
Hence, we have $\sizeL = \bigO(\log(I) + n + m)$.
The transition relation is defined by the following rules:
For transmitting the bits in the first phase, we have for each $\ell \in [1..\log(I)]$ the transition $q^{(b, \ell-1)}_L \xrightarrow{!(u,\ell)}_L q^{(b,\ell)}_L$ with $u \in \setcon{0,1}$.
For choosing the evaluation of variables in the second phase, we have for $i \in [1..n]$ the rule \mbox{$q^{(x,i-1)}_L \xrightarrow{!(x_i, v)}_L q^{(x,i)}_L$} with $v \in \setcon{0,1}$.
We denote the state $q^{(b,\log(I))}_L$ by $q^{(x,0)}_L$ to connect the phases.
In the third phase, $P_L$ wants to read the symbols $\#_j$.
Thus, we get for any $j \in [1..m]$ the transition $q^{(\#,j-1)}_L \xrightarrow{? \#_j}_L q^{(\#,j)}_L$.
Here we denote by $q^{(\#,0)}_L$ the state $q^{(x,n)}_L$.
We set $F_L = \setcon{q^{(\#,m)}}$.

The contributor $P_C$ is defined by $P_C = (\Ops{\Domain}, Q_C, q^\varepsilon_C, \delta_C)$.
The set of states $Q_C$ is the union of 
\begin{align*}
	&\Setcon{q^w}{ w \in \setcon{0,1}^{\leq \log(I)}}, \\
	&\Setcon{q^{(ch,\ell)}}{\ell \in [1..I]}, ~\text{and} \\
	&\Setcon{q^\ell_{(x_i,v)}}{\ell \in [1..I], i \in [1..n], v \in \setcon{0,1}}.
\end{align*}
Intuitively, the states $q^w$ with \mbox{$w \in \setcon{0,1}^{\leq \log(I)}$} form the nodes of the tree of the first phase.
The remaining states are needed to store the chosen instance, a variable, and its evaluation.

The transition relation $\delta_C$ contains rules for the three phases.
In the first phase, $P_C$ reads the bits chosen by the leader.
According to the value of the bit, it branches to the next state:
$q^w \xrightarrow{?(u,\abs{w}+1)}_C q^{w.u}$ for $u \in \setcon{0,1}$ and \mbox{$w \in \setcon{0,1}^{\leq \log(I)-1}$}.
Then we get an $\varepsilon$-transition from those leaves of the tree that encode a proper index, lying in $[1..I]$.
We have the rule $q^w \xrightarrow{\varepsilon}_C q^{(ch,\ell)}$ if $w = \bin(\ell)$ and $\ell \in [1..I]$.
For the second phase, we need transitions to store a variable and its evaluation: $q^{(ch,\ell)} \xrightarrow{?(x_i,v)}_C q^\ell_{(x_i,v)}$ for $\ell \in [1..I], i \in [1..n]$, and $v \in \setcon{0,1}$.
In the third phase, the contributor loops.
We have $q^\ell_{(x_i,v)} \xrightarrow{!\#_j}_C q^\ell_{(x_i,v)}$ if $x_i$ evaluated to $v$ satisfies clause $C^\ell_j$.

The program $\asms$ is defined to be $\asms = (\Domain, \initmem, (P_L, P_C))$ and the $\LCR$-instance of interest is thus $(\asms, F_L)$.
The correctness of the construction is shown in the following lemma.
Hence, all requirements of a cross-composition are met.

\begin{lemma}
	There is a $t \in \Naturals$ so that $c^0\rightarrow^*_{\asms^t}c$ with $c\in C^f$ if and only if there is an $\ell \in [1..I]$ such that $\varphi_\ell$ is satisfiable.
\end{lemma}
\begin{proof}
	We first assume that there is an $\ell \in [1..I]$ such that $\varphi_\ell$ is satisfiable.
	Let $v_1, \dots, v_n$ be the evaluation of the variables $x_1, \dots, x_n$ that satisfies $\varphi_\ell$.
	Further, we let $\bin(\ell) = u_1 \dots u_{\log(I)}$ be the binary representation of $\ell$.
	Set $t = n$.
	Then the program $\asms^t$ has $n$ copies of the contributor.
	Denote them by $P_1, \dots, P_n$.
	Intuitively, $P_i$ is responsible for variable $x_i$.
	
	We construct a computation of $\asms^t$ from $c^0$ to a configuration in $C^f$.
	We proceed as in the aforementioned phases.
	In the first phase, $P_L$ starts to guess the first bit of $\bin(\ell)$.
	This is read by all the contributors.
	We get a computation
	\begin{align*}
	c^0 &\xrightarrow{!(u_1,1)}_L (q^{(b,1)}_L, q^\varepsilon, \dots, q^\varepsilon, (u_1,1)) \\
	&\xrightarrow{?(u_1,1)}_1 (q^{(b,1)}_L, q^{u_1}, q^\varepsilon, \dots, q^\varepsilon, (u_1,1)) \\
	\dots &\xrightarrow{?(u_1,1)}_n (q^{(b,1)}_L, q^{u_1}, \dots, q^{u_1}, (u_1,1)) = c^{(b,1)}.
	\end{align*}
	To extend the computation, we let $P_L$ guess the remaining bits, while the contributors read and store.
	Hence, we get
	\begin{align*}
	c^0 \rightarrow^*_{\asms^t} c^{(b,\log(I))} = (q^{(b, \log(I))}_L, q^{\bin(\ell)}, \dots, q^{\bin(\ell)}, (u_{\log(I)}, \log(I))).
	\end{align*}
	
	Before the second phase starts, the contributors perform an $\varepsilon$-transition to the state $q^{(ch,\ell)}$.
	This is possible since $\bin(\ell)$ encodes a proper index in $[1..I]$.
	We get $c^{(b, \log(I))} \rightarrow^*_{\asms^t} (q^{(b, \log(I))}_L, q^{(ch,\ell)}, \dots, q^{(ch,\ell)}, (u_{\log(I)}, \log(I))) = c^{(x,0)}$.
	
	In the second phase $P_L$ chooses the correct evaluation for the variables, it writes $(x_i,v_i)$ for each variable $x_i$.
	Contributor $P_i$ reads $(x_i,v_i)$ and stores it.
	Hence, we get the computation
	\begin{align*}
	c^{(x,0)} &\xrightarrow{!(x_1,v_1)}_{L} (q^{(x,1)}_L, q^{(ch,\ell)}, \dots, q^{(ch, \ell)}, (x_1, v_1)) \\
	&\xrightarrow{?(x_1,v_1)}_1 (q^{(x,1)}_L, q^\ell_{(x_1,v_1)}, q^{(ch,\ell)}, \dots, q^{(ch, \ell)}, (x_1, v_1)) \\
	&\xrightarrow{!(x_2,v_2)}_L (q^{(x,2)}_L, q^\ell_{(x_1,v_1)}, q^{(ch,\ell)}, \dots, q^{(ch, \ell)}, (x_2, v_2)) \\
	&\xrightarrow{?(x_2,v_2)}_2 (q^{(x,2)}_L, q^\ell_{(x_1,v_1)}, q^\ell_{(x_2,v_2)}, q^{(ch,\ell)}, \dots, q^{(ch, \ell)}, (x_2, v_2)) \\
	\dots &\xrightarrow{?(x_n,v_n)}_n (q^{(x,n)}_L, q^\ell_{(x_1,v_1)}, \dots, q^\ell_{(x_n,v_n)}, (x_n, v_n)) = c^{(x,n)}.
	\end{align*}
	
	In the last phase, the contributors write the symbols $\#_j$.
	Since $\varphi_\ell$ is satisfied by the evaluation $v_1, \dots, v_n$, there is a variable $i_1 \in [1..n]$ such that $x_{i_1}$ evaluated to $v_{i_1}$ satisfies clause $C^\ell_1$.
	Hence, due to the transition relation $\delta_C$, we can let $P_{i_1}$ write the symbol $\#_1$.
	After that, the leader reads it and moves to the next state.
	This amounts to the computation
	\begin{align*}
	c^{(x,n)} &\xrightarrow{!\#_1}_{i_1} (q^{(x,n)}_L, q^\ell_{(x_1,v_1)}, \dots, q^\ell_{(x_n,v_n)}, \#_1) \\
	&\xrightarrow{?\#_1}_L (q^{(\#,1)}_L, q^\ell_{(x_1,v_1)}, \dots, q^\ell_{(x_n,v_n)}, \#_1) = c^{(\#,1)}.
	\end{align*}
	Similarly, we can extend the computation to reach the configuration 
	\begin{align*}
		c^{(\#,m)} = (q^{(\#,m)}_L, q^\ell_{(x_1,v_1)}, \dots, q^\ell_{(x_n,v_n)}, \#_m)
	\end{align*}
	which lies in $C^f$.
	This proves the first direction.
	
	Now we assume the existence of a $t \in \Naturals$ such that there is a computation $\rho$ from $c^0$ to a configuration in $C^f$.
	Let $\rho_L$ be the subcomputation of $\rho$ carried out by the leader $P_L$.
	Then $\rho_L$ can be split into $\rho_L = \rho^1_L.\rho^2_L.\rho_L^3$ such that
	\begin{align*}
	\rho^1_L &= q^{(b,0)}_L \xrightarrow{!(u_1,1)}_L q^{(b,1)}_L \xrightarrow{!(u_2,2)}_L \dots \xrightarrow{!(u_{\log(I)}, \log(I))}_L q^{(b,\log(I))}_L, \\
	\rho^2_L &= q^{(b,\log(I))}_L \xrightarrow{!(x_1,v_1)}_L q^{(x,1)}_L \xrightarrow{!(x_2,v_2)}_L \dots \xrightarrow{!(x_n,v_n)}_L q^{(x,n)}_L, \\
	\rho^3_L &= q^{(x,n)}_L \xrightarrow{?\#_1}_L q^{(\#,1)}_L \xrightarrow{?\#_2}_L \dots \xrightarrow{?\#_m}_L q^{(\#,m)}_L.
	\end{align*}
	Let $\ell$ be the natural number such that $\bin(\ell) = u_1 \dots u_{\log(I)}$.
	We show that $\ell \in [1..I]$ and $\varphi_\ell$ is satisfied by evaluating the variables $x_i$ to $v_i$.
	
	In $\rho^3_L$, the leader can read the symbols $\#_1, \dots, \#_m$.
	This means that there is at least one contributor writing them.
	Let $P_C$ be a contributor writing such a symbol.
	Then, after $P_L$ has finished $\rho^1_L$, the contributor $P_C$ is still active and performs the step $q^{\bin(\ell)} \xrightarrow{\varepsilon}_C q^{(ch, \ell)}$.
	This is true since $P_C$ did not miss a bit transmitted by $P_L$ and $P_C$ has to reach a state where it can write the $\#$-symbols.
	Thus, we get that $\ell \in [1..I]$ and $P_C$ stores $\ell$ in its state space.
	
	We denote the number of contributors writing a $\#$-symbol in $\rho$ by $t' \geq 1$.
	Each of these contributors gets labeled by $C(j) = \setcon{\#_{j_1}, \dots, \#_{j_{k_j}}}$, the set of $\#$-symbols it writes during the computation $\rho$.
	Hence, we have the contributors $P_{C(1)}, \dots, P_{C(t')}$ and since each symbol in $\setcon{\#_1, \dots, \#_m}$ is written at least once, we have:
	\begin{align}\label{Equation:PRnopolykernelML}
	\setcon{\#_1, \dots, \#_m} = \bigcup_{j = 1}^{t'} C(j).
	\end{align}
	
	Now we show that each $P_{C(j)}$ with $C(j) = \setcon{\#_{j_1}, \dots, \#_{j_{k_j}}}$ stores a tuple $(x_i,v_i)$ such that $x_i$ evaluated to $v_i$ satisfies the clauses $C^\ell_{j_1}, \dots, C^\ell_{j_{k_j}}$.
	We already know that $P_{C(j)}$ is in state $q^{(ch, \ell)}$ after $P_L$ has executed $\rho^1_L$.
	During $P_L$ executing $\rho^2_L$, the contributor $P_{C(j)}$ has to read a tuple $(x_i,v_i)$ since it has to reach a state where it can write the $\#$-symbols.
	More precise, $P_{C(j)}$ has to perform a transition $q^{(ch,\ell)} \xrightarrow{?(x_i,v_i)}_{C(j)} q^{\ell}_{(x_i,v_i)}$ for some $i$.
	Then the contributor writes the symbols $\#_{j_1}, \dots, \#_{j_{k_j}}$ while looping in the current state.
	But by the definition of the transition relation for the contributors, this means that $x_i$ evaluated to $v_i$ satisfies the clauses $C^\ell_{j_1}, \dots, C^\ell_{j_{k_j}}$.
	
	By Equation (\ref{Equation:PRnopolykernelML}) we can deduce that every clause in $\varphi_\ell$ is satisfied by the chosen evaluation.
	Hence, $\varphi_\ell$ is satisfiable.
	\qed
\end{proof}

%% file: content/appendix_LCR_Contributor.tex
\section{Proofs for Section \ref{Section:LCR_Contributor}}

We give the missing constructions and proofs for Section \ref{Section:LCR_Contributor}.

\subsection*{Proof of Lemma \ref{Lemma:LCR_Contributor_Upper_Bound_Reach}}

Before we elaborate on the proof, we introduce a few notations.
Let $t \in \Naturals$ and $c$ a configuration of $\asms^t$.
To access the components of $c$ we use the following projections:
$\pi_L(c)$ returns the state of the leader in $c$, $\pi_\Domain(c)$ the value of the shared memory.
For $p \in Q_C$, we denote by $\countC(c,p)$ the number of contributors that are currently in state $p$.
Finally, we use $\pi_C(c)$ for the set of contributor states that appear in $c$.
Formally, $\pi_C(c) = \Setcon{p \in Q_C}{ \countC(c,p) > 0}$.

The proof of Lemma \ref{Lemma:LCR_Contributor_Upper_Bound_Reach} is a consequence of the following stronger lemma.
It states that for any reachable configuration in the program, there is a node in $\fingraph$ reachable by $\initnode$ such that:
(1) The state of the leader and the memory value are preserved and (2) the possible states of the contributors can only increase.

\begin{lemma}
There is a $t \in \Naturals$ so that $c^0 \rightarrow^*_{\asms^t} c$ if and only if there is a path $\initnode \rightarrow^*_E (q,a,S)$ in $\fingraph$, where $\pi_L(c) = q$, $\pi_\Domain(c) = a$, and $\pi_C(c) \subseteq S$.
\end{lemma}

First assume that a computation $c^0 \rightarrow^*_{\asms^t} c$ for a $t \in \Naturals$ is given.
We proceed by induction on the length of the computation.
In the base case, the length is $0$.
This means that $c = c^0$ is the initial configuration.
But then $\pi_L(c) = q^0_L$, $\pi_\Domain(c) = a^0$, and $\pi_C(c) = \setcon{q^0_C}$.
This characterizes the initial node of $\fingraph$ and there is a path $\initnode \rightarrow^*_E \initnode$ of length $0$ which proves the base case.

Suppose the statement holds for all computations of length at most $\ell$.
Let $c^0 \rightarrow^*_{\asms^t} c$ be a computation of length $\ell + 1$.
Then, it can be split into $c^0 \rightarrow^*_{\asms^t} c' \rightarrow_{\asms^t} c$, where $c^0 \rightarrow^*_{\asms^t} c' $ is a computation of length $\ell$.
By induction, there is a path $\initnode \rightarrow^*_E (q',a',S')$ in $\fingraph$ such that $q' = \pi_L(c')$, $a' = \pi_\Domain(c')$, and $\pi_C(c') \subseteq S'$.
Now we distinguish two cases:

(1) If $c' \rightarrow_{\asms^t} c$ is induced by a transition of the leader, the leaders' state and the memory value get updated, but the contributor states do not.
We have that $\pi_C(c) = \pi_C(c') \subseteq S'$.
Now we set $q = \pi_L(c)$, $a = \pi_\Domain(c)$ and $S' = S$.
Then, on $\fingraph$ we have an edge $(q',a',S') \rightarrow_E (q,a,S)$.

(2) If $ c' \rightarrow_{\asms^t} c$ is induced by a transition of a contributor, we immediately get that $\pi_L(c) = \pi_L(c') = q'$.
Let the transition of the contributor be of the form $p' \xrightarrow{?a'/ \varepsilon} p$.
Then we have that $\pi_C(c) \subseteq \pi_C(c') \cup \setcon{p} \subseteq S' \cup \setcon{p}$ and $\pi_\Domain(c') = \pi_\Domain(c) = a'$.
Note that it can happen that $p'$ is not an element of $\pi_C(c)$ since there might be just one contributor in state $p'$ which switches to state $p$.
We set $q = q'$, $a = a'$, and $S = S' \cup \setcon{p}$.
Then we have an edge $(q',a',S') \rightarrow_E (q,a,S)$ induced by the transition.
Writes of the contributors are similar.
 
This shows the first direction of the lemma.
For the other direction, we apply induction to prove a slightly stronger statement:
For each path $\initnode \rightarrow^*_E (q,a,S)$, there is a $t \in \Naturals$ and a computation $c^0 \rightarrow^*_{\asms^t} c$ such that $\pi_L(c) = q$, $\pi_\Domain(c) = a$, and $\pi_C(c) = S$.
In the proof we rely on the \emph{Copycat Lemma}, presented in \cite{Esparza13}.
Roughly it states that for a computation where a state $p \in Q_C$ is reached by one of the contributors, there is a similar computation where $p$ is reached by an arbitrary number of contributors.
We restate the lemma in our setting.
 
 \begin{lemma}[Copycat Lemma \cite{Esparza13}]
 	Let $t \in \Naturals$ and $ c^0 \rightarrow^*_{\asms^t} c$ a computation. 
 	Moreover, let $p \in Q_C$ such that $\countC(c,p) > 0$.
 	Then for all $k \in \Naturals$, we have a computation of the form $ c^0 \rightarrow^*_{\asms^{t+k}} d$, where configuration $d$ satisfies the following:
 	$\pi_L(d) = \pi_L(c)$, $\pi_\Domain(d) = \pi_\Domain(c)$, $\countC(d,p) = \countC(c,p) + k$ and for all $p' \neq p$ we have  $\countC(d,p') = \countC(c,p')$.
\end{lemma}

We turn back to the induction on the length of the given path.
In the base case, the length is $0$.
Then, we have that $(q,a,S) = \initnode$.
This means $q = q^0_L$, $a = a^0$, and $S = \setcon{q^0_C}$.
Considering the initial configuration $c^0$ for an arbitrary $t \in \Naturals$, we get the computation $c^0 \rightarrow^*_{\asms^t} c^0$ of length $0$, with $\pi_L(c^0) = q$, $\pi_\Domain(c^0) = a$, and $\pi_C(c^0) = S$.

Assume the statement holds true for all paths of length at most $\ell$.
Let $\initnode \rightarrow^*_E (q,a,S)$ be a path of length $\ell+1$.
We split the path into a subpath $\initnode \rightarrow^*_E (q',a',S')$ of length $\ell$ and an edge $(q',a',S') \rightarrow_E (q,a,S)$.
Invoking the induction hypothesis, we get a $t \in \Naturals$ and a computation $c^0 \rightarrow^*_{\asms^t} c'$ such that $\pi_L(c') = q'$, $\pi_\Domain(c') = a'$, and $\pi_C(c') = S'$.
We distinguish two cases:

(1) The edge $(q',a',S') \rightarrow_E (q,a,S)$ was induced by a transition of the leader.
Since $\pi_L(c') = q'$ and $\pi_\Domain(c') = a'$, the same transition also induces a step $c' \rightarrow_{\asms^t} c$ with $\pi_L(c) = q$, $\pi_\Domain(c) = a$, and $\pi_C(c) = S = S'$.

(2) The edge $(q',a',S') \rightarrow_E (q,a,S)$ was induced by a transition of a contributor.
Suppose, this transition is of the form $\tau = p' \xrightarrow{?a' / \varepsilon} p$.
The case of a write is similar.
Then we get $S = S' \cup \setcon{p}$ and $p' \in S'$.
Since $\pi_C(c') = S'$, we get that $\countC(c',p') > 0$.
By an application of the Copycat Lemma with $k = 1$, we obtain a computation of the form $c^0 \rightarrow^*_{\asms^{t+1}} d$ such that $\countC(d,p') > 1$ and for all $r \neq p'$ we have $\countC(d,r) = \countC(c',r)$.
Furthermore, we get that $\pi_L(d) = \pi_L(c')$ and $\pi_\Domain(d) = \pi_\Domain(c')$.
Hence, transition $\tau$ induces a move $d \rightarrow_{\asms^{t+1}} c$, where $c$ is a configuration with $\pi_L(c) = \pi_L(d) = q' = q$, $\pi_\Domain(c) = \pi_\Domain(d) = a' = a$ and $\pi_C(c) = \pi_C(d) \cup \setcon{p} = S' \cup \setcon{p} = S$.

\subsection*{Formal Construction and Proof of Proposition \ref{Proposition:LCR_Contributor_Lower_Bound}}

The memory domain is defined by $\Domain = U \cup \Setcon{u^\#}{u \in U} \cup \setcon{a^0}$.
The leader thread is the tuple $P_L = (\Ops{\Domain}, Q_L, q^1, \delta_L)$, where the set of states $Q_L$ is the union of $\Setcon{q^i}{i \in [1..r+1]}$, $\Setcon{q^{(i,j)}_S}{S \in \mathcal{F}, j \in [0..\abs{S}-1] \text{ and } i \in [1..r] }$, and $\Setcon{q_\#^i}{i \in [1..n]}$.
Recall that $n = \abs{U}$.
The states $q^i$ are needed  to choose a set $S \in \mathcal{F}$.
Then, the $q^{(i,j)}_S$ are used to iterate over the elements in $S$.
For the final phase in $P_L$, the states $q_\#^i$ are needed to read all elements $u^\#$ for $u \in U$.

The transition relation $\delta_L$ contains the following rules:
For choosing a set, we have transitions of the form $q^i \xrightarrow{\varepsilon} q^{(i,0)}_S$ for each $S \in \mathcal{F}$ and $i \in [1..r]$.
Iterating through a set $S = \setcon{v_1, \dots, v_{\abs{S}}}$ is done via the transitions $q^{(i,j)}_S \xrightarrow{!v_{j+1}} q^{(i,j+1)}_S$ for $j \in [0, \abs{S}-2]$.
For the last element, we have a transition $q^{(i,\abs{S}-1)}_S \xrightarrow{!v_{\abs{S}}} q^{i+1}$ that enters the new phase.
Fix an order on $U = \setcon{u_1, \dots, u_n}$.
The final check is realized by the transitions $q^{r+1} \xrightarrow{?u^\#_1} q^1_\#$ and $q^i_\# \xrightarrow{?u^\#_{i+1}} q^{i+1}_\#$ for $i \in [1..n-1]$.
The leader only reaches a final state after the last check: $F_L = \setcon{q^n_\#}$.

A contributor is defined by the tuple $P_C = (\Ops{\Domain}, Q_C, p^0,  \delta_C)$ where the set of states is given by $Q_C = \Setcon{p_u}{u \in U} \cup \setcon{p^0}$.
The transition relation contains rules to store elements of $U$ in the state space:
$p^0 \xrightarrow{?u} p_u$, for each $u \in U$.
Once an element is stored, the contributor can write it to the memory:
$p_u \xrightarrow{!u^\#} p_u$.
Correctness is proven in the following lemma.

\begin{lemma}
	There is a $t \in \Naturals$ so that $c^0 \rightarrow^*_{\asms^t} c$ with $c \in C^f$ if and only if there are sets $S_1, \dots S_r \in \mathcal{F}$ such that $U = \bigcup_{i \in [1,r]} S_i$.
\end{lemma}
\begin{proof}
	Let $S_1, \dots, S_r \in \mathcal{F}$ be a cover of $U$.
	We can construct a computation with $t = n$ contributors.
	The leader first guesses the set $S_1$.
	It writes all elements $u \in S_1$ to the memory and there is one contributor storing each element in its states by reading the corresponding $u$.
	Then, the leader decides for $S_2$ and writes the elements in the set to the memory.
	Now, only the new elements got stored by a contributor.
	Elements that were seen already are ignored.
	We proceed for $r$ phases.
	Then, the contributors store exactly those elements that got covered by $S_1, \dots, S_r$.
	Since these cover $U$, the contributors can write all symbols $u^\#$ to the memory in any order.
	The leader $P_L$ can thus read the required string and reach its final state.
	
	Now assume there is a $t \in \Naturals$ together with a computation $\rho$ on $\asms^t$ from $c^0$ to a configuration $c \in C^f$.
	Consider $\rho_L$, the projection of $\rho$ to the leader $P_L$.
	Then, the computation $\rho_L$ is of the form $\rho_L = \rho^1_L \dots \rho^r_L . \rho^f_L$ with:
	\begin{align*}
		\rho^i_L = q^i \xrightarrow{\varepsilon} q^{(i,0)}_{S_i}
		\xrightarrow{!u^{S_i}_1} q^{(i,1)}_{S_i}
		\xrightarrow{!u^{S_i}_2} \dots
		\xrightarrow{!u^{S_i}_{n_i-1}} q^{(i,n_i-1)}_{S_i}
		\xrightarrow{!u^{S_i}_{n_i}} q^{i+1}, 
	\end{align*}
	where $S_i = \setcon{u^{S_i}_1, \dots, u^{S_i}_{n_i}}$ is a set in $\mathcal{F}$, and
	\begin{align*}
		\rho^f_L = q^{r+1} \xrightarrow{?u^\#_1} q^1_\#
		\xrightarrow{?u^\#_2} \dots
		\xrightarrow{?u^\#_n} q^n_\#.
	\end{align*}
	The candidate for the cover of $U$ is $S_1, \dots, S_r \in \mathcal{F}$.
	These are the sets selected by $P_L$ during its $r$ initial phases.
	
	A contributor can only read and thus move in its state space during the leader is in a phase $\rho^i_L$.
	This means that contributors can only store symbols that got covered by the chosen sets $S_i$.
	Moreover, they can only write what they have stored.
	Since $\rho^f_L$ can be carried out by the leader, the contributors can write all elements $u \in U$ to the memory.
	Phrased differently, all elements $u \in U$ were stored by contributors and hence covered by $S_1, \dots, S_r$.
	\qed
\end{proof}

\subsection*{Formal Construction and Proof of Proposition \ref{Proposition:LCR_Contributor_Kernel_Lower}}

The construction of Proposition \ref{Proposition:LCR_Contributor_Kernel_Lower} is similar to the construction in the following statement.
It presents a lower bound for $\LCR$ based on $\ETH$ and shows that the Algorithm of Section \ref{Section:LCR_Contributor_Upper} has an optimal exponent.

\begin{proposition}\label{Proposition:LCR_Contributor_Lower_ETH}
	Unless $\ETH$ fails, $\LCR$ cannot be solved in time $2^{o(\sizeC)}$.
\end{proposition}

For the reduction, let $\varphi$ be a given $\kSAT{3}$-instance.
We assume $\varphi$ to have the variables $x_1, \dots, x_n$ and clauses $C_1, \dots, C_m$.
The construction of an $\LCR$-instance relies on the following idea which is similar to Proposition \ref{Proposition:LCR_Contributor_Kernel_Lower}.
The leader $P_L$ will guess an evaluation for each variable, starting with $x_1$.
To this end, it will write a tuple of the form $(x_1,v_1)$, with $v_1 \in \setcon{0,1}$, to the memory.
A contributor will read the tuple and stores it in its state space.
This is repeated for each variable.
After the guessing-phase, the contributors can write the symbols $\#_j$, depending on whether the currently stored variable with its evaluation satisfies clause $C_j$.
As soon as the leader has read the complete string $\#_1 \dots \#_m$, it moves to its final state, showing that the guessed evaluation satisfied all the clauses.

For the formal construction, let 
\begin{align*}
	\Domain = \Setcon{(x_i,v)}{i \in [1..n], v \in \setcon{0,1}} \cup \Setcon{\#_j}{j \in [1..m]} \cup \setcon{\initmem}.
\end{align*}
We define the leader to be the tuple $P_L = (\Ops{\Domain}, Q_L, q^{(x,0)}_L, \delta_L)$, where the states are given by \mbox{$Q_L = \Setcon{q^{(x,i)}_L}{i \in [0..n]} \cup \Setcon{q^{(\#,j)}_L}{j \in [1..m]}$}.
The transition relation $\delta_L$ is defined as follows.
We have rules for guessing the evaluation: \mbox{$q^{(x,i-1)}_L \xrightarrow{!(x_i,v)}_L q^{(x,i)}_L$} for each $i \in [1..n]$ and $v \in \setcon{0,1}$.
And we have rules for verifying that the guessed evaluation is correct: $q^{(\#,j-1)}_L \xrightarrow{?\#_j}_L q^{(\#,j)}_L$ for $j \in [1..m]$.
Note that we identify the state $q^{(x,n)}_L$ by $q^{(\#,0)}_L$.
We further define the set $F_L = \setcon{q^{(\#,m)}_L}$.

The contributor $P_C$ is defined by $P_C = (\Ops{\Domain}, Q_C, q^0, \delta_C)$ with set of states $Q_C = \Setcon{q_{(x_i,v)}}{i \in [1..n], v \in \setcon{0,1}} \cup \setcon{q^0}$.
Then we have $\sizeC = 2n+1 = \bigO(n)$.
The transition relation contains the following rules.
For storing a read evaluation we have: $q^0 \xrightarrow{?(x_i,v)}_C q_{(x_i,v)}$, for $i \in [1..n], v \in \setcon{0,1}$.
And for satisfying clauses, we get a rule $q_{(x_i,v)} \xrightarrow{!\#_j}_C q_{(x_i,v)}$ if variable $x_i$ evaluated to $v$ satisfies clause $C_j$.

The program $\asms$ is defined as the tuple $\asms = (\Domain, \initmem, (P_L, P_C))$ and the $\LCR$-instance is $(\asms, F_L)$.
The correctness of the construction is proven in the following lemma.

\begin{lemma}\label{Lemma:LCR_Contributor_Lower_ETH_Correctness}
	There is a $t \in \Naturals$ so that $c^0\rightarrow^*_{\asms^t}c$ with $c\in C^f$ if and only if $\varphi$ is satisfiable.
\end{lemma}
\begin{proof}
	Let $\varphi$ be satisfiable, by the evaluation $v_1, \dots, v_n$.
	We show how to construct the desired computation.
	First set $t = n$, so we have $n$ copies of $P_C$.
	Let these denoted by $P_1, \dots, P_n$.
	
	The leader $P_L$ first guesses the correct evaluation of the variables.
	Each $P_i$ will store the evaluation for variable $x_i$.
	\begin{align*}
		c^0 &\xrightarrow{!(x_1,v_1)}_L (q^{(x,1)}_L, q^0, \dots, q^0, (x_1,v_1)) \\
		&\xrightarrow{?(x_1,v_1)}_1 (q^{(x,1)}_L, q_{(x_1,v_1)}, q^0, \dots, q^0, (x_1,v_1)) \\
		\dots &\xrightarrow{?(x_n,v_n)}_n (q^{(x,n)}_L, q_{(x_1,v_1)}, \dots, q_{(x_n,v_n)}, (x_n,v_n)) = c^{(x,n)}.
	\end{align*}
	Since $v_1, \dots, v_n$ is a satisfying assignment, there is an index $i_1 \in [1..n]$ such that $x_{i_1}$ evaluated to $v_{i_1}$ satisfies clause $C_1$.
	Hence, the corresponding contributor $P_{i_1}$ can write the symbol $\#_1$.
	This is then read by the leader.
	The process gets repeated for $\#_2, \dots, \#_m$.
	Hence, we get
	\begin{align*}
		c^{(x,n)} &\xrightarrow{!\#_1}_{i_1} (q^{(x,n)}_L, q_{(x_1,v_1)}, \dots, q_{(x_n,v_n)}, \#_1) \\
		&\xrightarrow{?\#_1}_L (q^{(\#,1)}_L, q_{(x_1,v_1)}, \dots, q_{(x_n,v_n)}, \#_1) \\
		\dots &\xrightarrow{?\#_m}_L (q^{(\#,m)}_L, q_{(x_1,v_1)}, \dots, q_{(x_n,v_n)}, \#_1) = c^{(\#,m)},
	\end{align*}
	with $c^{(\#,m)} \in C^f$.
	
	For the other direction, let a $t \in \Naturals$ and a computation $\rho$ from $c^0$ to a configuration in $C^f$ be given.
	Let $\rho_L$ denote the subcomputation of $\rho$ carried out by the leader $P_L$.
	Then $\rho_L$ has the form $\rho_L = \rho^1_L.\rho^2_L$ with
	\begin{align*}
		\rho^1_L &= q^{(x,0)}_L \xrightarrow{!(x_1,v_1)}_L q^{(x,1)}_L \xrightarrow{!(x_2,v_2)}_L \dots \xrightarrow{!(x_n,v_n)}_L q^{(x,n)}_L, \\
		\rho^2_L &= q^{(x,n)}_L \xrightarrow{?\#_1}_L q^{(\#,1)}_L \xrightarrow{?\#_2}_L \dots \xrightarrow{?\#_m}_L q^{(\#,m)}_L.
	\end{align*}
	We show that $v_1, \dots, v_n$ is a satisfying assignment for $\varphi$.
	
	Since $P_L$ can read the symbol $\#_1$ during $\rho^2_L$, there is a contributor $P_{\ell}$ writing the symbol.
	But this can only happen if $P_{\ell}$ has stored a tuple $(x_i,v_i)$, written by $P_L$ during $\rho^1_L$, and if $x_i$ evaluated to $v_i$ satisfies clause $C_1$.
	Since all symbols $\#_1, \dots \#_m$ are read by $P_L$, we get that each clause in $\varphi$ is satisfiable by the evaluation chosen by $P_L$ during $\rho^1_L$.
	\qed	
\end{proof}

To prove Proposition \ref{Proposition:LCR_Contributor_Kernel_Lower}, we change the above construction slightly.
Let $\varphi_1, \dots, \varphi_I$ be the given $\kSAT{3}$-instances, each pair equivalent under $\polyrel$, where $\polyrel$ is the polynomial equivalence relation from Theorem \ref{Theorem:LCR_Leader_Mem_Kernel_Lower}.
Then each formula has the same number of clauses $m$ and uses the set of variables $\setcon{x_1, \dots, x_n}$.
We assume $\varphi_\ell = C^\ell_1 \wedge \dots \wedge C^\ell_m$.

First, we let the leader chose an evaluation of the variables $x_1, \dots, x_n$ as above.
The contributors are used to store it.
Then, instead of writing just $\#_j$, the contributors can write the symbols $\#^\ell_j$ to mention that the currently stored variable with its evaluation satisfies clause $C^\ell_j$.
The leader can now branch into one of the $I$ instances.
It waits to read a string $\#^\ell_1 \dots \#^\ell_m$ for a certain $\ell \in [1..I]$.
If it can succeed, it moves to its final state.

To realize the construction, we need to slightly change the structure of the leader, extend the data domain and add more transitions to the contributors.
The parameter $\sizeC$ will not change in this construction, it is still $\bigO(n)$.
Hence, the size-restrictions of a cross-composition are met.
The correctness of the construction is similar to Lemma \ref{Lemma:LCR_Contributor_Lower_ETH_Correctness}, the only difference is the fact that $P_L$ also chooses the instance $\varphi_\ell$ that should be satisfied.

%% file: content/appendix_LCR_Intractability.tex
\section{Proofs for Section \ref{Section:LCR_Intractability}}
\label{Appendix_LCR_Intractability}

We give the missing constructions and proofs for Section \ref{Section:LCR_Intractability}.

\subsection*{Formal Construction and Proof of Proposition \ref{Proposition:LCR_Intractability}}

We first give construction and proof for the $\W[1]$-hardness of $\LCR(\sizeL)$.
We denote by $V$ the vertices of $G$ and by $E$ the edges.
Set the data domain $\Domain = \Setcon{(v,i), (v^\#,i), \#_i}{ v \in V, i \in [1..k]} \cup \setcon{a^0}$.
The leader $P_L$ is given by the tuple $P_L = (\Ops{\Domain}, Q_L, q^0, \delta_L)$ with set of states $Q_L = \Setcon{q^i_V, q^i_{V^\#}, q^i_\#}{i \in [1..k]} \cup \setcon{q^0}$.
The transition relation $\delta_L$ is defined by the following rules.
(First Phase) For each $i \in [1..k]$ and $v \in V$, we add the rule $q^{i-1}_V \xrightarrow{!(v,i)}_L q^{i}_V$.
We identify the vertex $q^0_V$ by $q^0$.
(Second Phase) For each $i \in [1..k]$ and $v \in V$, we add $q^{i-1}_{V^\#} \xrightarrow{!(v^\#,i)}_L q^{i}_{V^\#}$.
Here, we denote by $q^0_{V^\#}$ the vertex $q^k_V$.
(Third Phase) For each $i \in [1..k]$, add $q^{i-1}_\# \xrightarrow{?\#_i}_L q^i_\#$.
Here, we assume $q^0_\# = q^k_{V^\#}$.
Further, we set $F_L = \setcon{q^k_\#}$ to be the final state of interest.

The contributor is defined by $P_C = (\Ops{\Domain}, Q_C, q^0_C, \delta_C)$.
The states are given by $Q_C = \Setcon{q^j_{(v,i)}}{ i \in [1..k], j \in [0..k] } \cup \setcon{q^0_C, q^f_C}$.
We define the transition relation by the following rules.
(First Phase) For each $i \in [1..k]$ and $v \in V$ we have $q^0_C \xrightarrow{?(v,i)}_C q^0_{(v,i)}$.
(Second Phase) For $i \in [1..k], j \in [0..k]$ and $v,w \in V$ we have $q^{j-1}_{(v,i)} \xrightarrow{?(w,j)}_C q^j_{(v,i)}$ if (1) $j = i$ and $v = w$, or if (2) $i \neq j$, $v \neq w$, and there is an edge between $v$ and $w$ in $E$.
(Third Phase) For any $i \in [1..k]$ and $v \in V$, add the rule $q^k_{(v,i)} \xrightarrow{!\#_i}_C q^f_C$.
The correctness is shown in the next lemma.

\begin{lemma}
	There is a $t \in \Naturals$ so that $c^0 \rightarrow^*_{\asms^t} c$ with $c \in C^f$ if and only if there is a clique of size $k$ in $G$.
\end{lemma}
\begin{proof}
	We first assume that $G$ contains a clique of size $k$.
	Let it be the vertices $v_1, \dots, v_k$.
	We construct a computation on $\asms^t$ with $t = k$ that leads from $c^0$ to a configuration $c$ in $C^f$.
	The program contains $k$ contributors, denoted by $P_1, \dots, P_k$.
	We proceed in three phases, as described above.
	
	In the first phase, the leader writes the values $(v_1,1), \dots, (v_k,k)$ to the memory.
	Contributor $P_i$ reads value $(v_i,i)$ and stores it in its state space.
	We get:
	\begin{align*}
		c^0 &\xrightarrow{!(v_1,1)}_L (q^1_V, q^0_C, \dots, q^0_C, (v_1,1) ) \\
		&\xrightarrow{?(v_1,1)}_1 (q^1_V, q^0_{(v_1,1)}, q^0_C, \dots, q^0_C, (v_1,1) ) \\
		&\xrightarrow{!(v_2,2)}_L (q^2_V, q^0_{(v_1,1)}, q^0_C, \dots, q^0_C, (v_2,2) ) \\
		&\xrightarrow{?(v_2,2)}_2 (q^2_V, q^0_{(v_1,1)}, q^0_{(v_2,2)}, q^0_C, \dots, q^0_C, (v_2,2) ) \\
		\dots &\xrightarrow{?(v_k,k)}_k (q^k_V, q^0_{(v_1,1)}, \dots, q^0_{(v_k,k)}, (v_k,k) ) = c_0.
	\end{align*}
	After reaching $c_0$, the leader starts the second phase.
	It writes $(v^\#_1, 1)$ and each contributor reads it:
	\begin{align*}
		c_0 &\xrightarrow{!(v^\#_1, 1)}_L ( q^1_{V^\#}, q^0_{(v_1,1)}, \dots, q^0_{(v_k,k)}, (v^\#_1, 1) ) \\
		&\xrightarrow{?(v^\#_1, 1)}_1 ( q^1_{V^\#}, q^1_{(v_1,1)}, q^0_{(v_2,2)}, \dots, q^0_{(v_k,k)}, (v^\#_1, 1) ) \\
		\dots &\xrightarrow{?(v^\#_1, 1)}_k ( q^1_{V^\#}, q^1_{(v_1,1)}, \dots, q^1_{(v_k,k)}, (v^\#_1, 1) ) = c_1.
	\end{align*}
	Note that $P_1$ can read $(v^\#_1, 1)$ and move since it stores exactly $(v_1,1)$.
	Any $P_i$ with $i \neq 1$ can read $(v^\#_1, 1)$ and continue its computation since $v_i \neq v_1$ and the two vertices share an edge.
	Similarly, one can continue the computation: $c_1 \rightarrow^*_{\asms^t} c_k = (q^k_{V^\#}, q^k_{(v_1,1)}, \dots, q^k_{(v_k,k)}, (v^\#_k, k))$.
	
	In the third phase, contributor $P_i$ writes the symbol $\#_i$ to the memory.
	The leader waits to read the complete string $\#_1 \dots \#_k$.
	This yields the following computation:
	\begin{align*}
		c_k &\xrightarrow{! \#_1}_1 (q^k_{V^\#}, q^f_C, q^k_{(v_2,2)}, \dots, q^k_{(v_k,k)}, \#_1) \\
		&\xrightarrow{? \#_1}_L (q^1_{\#}, q^f_C, q^k_{(v_2,2)}, \dots, q^k_{(v_k,k)}, \#_1) \\
		\dots &\xrightarrow{? \#_k}_L (q^k_{\#}, q^f_C, \dots, q^f_C, \#_k) \in C^f.
	\end{align*}
	
	For the other direction, let a $t \in \Naturals$ and a computation $\rho = c^0 \rightarrow^*_{\asms^t} c$ with $c \in C^f$ be given.
	We denote by $\rho_L$ the part of the computation that is carried out by the leader $P_L$.
	Then we can factor $\rho_L$ into $\rho_L = \rho^1 . \rho^2 . \rho^3$ with
	\begin{align*}
		\rho^1 &= q^0 \xrightarrow{!(v_1,1)}_L q^1_V \xrightarrow{!(v_2,2)}_L \dots \xrightarrow{!(v_k,k)}_L q^k_V, \\
		\rho^2 &= q^k_V \xrightarrow{!(w^\#_1,1)}_L q^1_{V^\#} \xrightarrow{!(w^\#_2,2)}_L \dots \xrightarrow{!(w^\#_k,k)}_L q^k_{V^\#}, \\
		\rho^3 &= q^k_{V^\#} \xrightarrow{? \#_1}_L q^1_\# \xrightarrow{? \#_2}_L \dots \xrightarrow{? \#_k}_L q^k_\#.
	\end{align*}
	We show that $w_i = v_i$ for any $i \in [1..k]$ and that $v_i \neq v_j$ for $i \neq j$.
	Furthermore, we prove that each two vertices $v_i, v_j$ share an edge.
	Hence, $v_1, \dots, v_k$ form a clique of size $k$ in $G$.
	
	Since $P_L$ is able to read the symbols $\#_1, \dots, \#_k$ in $\rho^3$, there are at least $k$ contributors writing them.
	But a contributor can only write $\#_i$ in its computation if it reads (and stores) the symbol $(v_i,i)$ from $\rho^1$.
	Hence, there is at least one contributor storing $(v_i,i)$.
	We denote it by $P_{v_i}$.
	
	The computation $\rho^2$ starts by writing $(w^\#_1,1)$ to the memory.
	The contributors $P_{v_i}$ have to read it in order to reach a state where they can write the symbol $\#_i$.
	Hence, $P_{v_1}$ reads $(w^\#_1,1)$.
	By the definition of the transition relation of $P_{v_1}$, this means that $w_1 = v_1$.
	Now let $P_{v_i}$ with $i \neq 1$.
	This contributor also reads $(w^\#_1,1) = (v^\#_1,1)$.
	By definition this implies that $v_i \neq v_1$ and the two vertices share an edge.
	
	By induction, we get that $w^\#_i = v_i$, the $v_i$ are distinct and each two of the $v_i$ share an edge.
	\qed
\end{proof}

To prove the $\W[1]$-hardness of $\LCR(\sizeM)$, we go back to our idea to transmit vertices in binary.
Let $t = \log(\abs{V})$ and $\bin : V \rightarrow \setcon{0,1}^t$ be a binary encoding of the vertices.
Instead of a single symbol $(v,i)$ with $v \in V$ and $i \in [1..k]$, the leader will write a string $\#.(\alpha_1,i).\#.(\alpha_2,i).\# \dots (\alpha_t,i).\#$ to the memory, where $t = \log(\abs{V})$,  $\alpha_1.\alpha_2 \dots \alpha_t = \bin(v)$, and $\#$ is a special padding symbol.
We need the padding in order to prevent the contributors from reading a symbol $(\alpha_j,i)$ multiple times.
Note that the new data domain contains only $\bigO(k)$ many symbols.

The idea of the program over the changed data domain is similar to the idea above: It proceeds in three phases.
In the first phase, the leader chooses the vertices of a clique candidate.
This is done by repeatedly writing a string $\#.(\alpha_1,i).\#.(\alpha_2,i).\# \dots (\alpha_t,i).\#$ to the memory, for each $i \in [1..k]$.
Like above, the contributors non-deterministically decide to store a written vertex.
To this end, a contributor that wants to store the $i$-th suggested vertex has a binary tree branching on the symbols $(0,i)$ and $(1,i)$.
Leaves of the tree correspond to binary encodings of vertices.
Hence, a particular vertex can be stored in the contributor's states. 
Note, as we did not assume $\abs{V}$ to the a power of $2$, there might be leaves of the tree that do not correspond to encodings of vertices.
If a computation reaches such a leaf it will deadlock.

In the second phase, the leader again writes the binary encoding of $k$ vertices to the memory.
But this time, it uses a different set of symbols: Instead of $0$ and $1$, the leader uses $0^\#$ and $1^\#$ to separate Phase two from Phase one.
The contributors need to compare the suggested vertices as in the above construction.
To this end, a contributor storing the vertex $(v,i)$ proceeds in $k$ stages.
In stage $j \neq k$ it can only read the encodings of those vertices which are connected to $v$.
Hence, if the leader suggests a wrong vertex, the computation will deadlock.
In stage $i$, the contributor can only read the encoding of the stored vertex $v$.
This allows a verification of the clique as above.

The last phase is identical to the last phase of the above construction.
The contributors write the symbols $\#_i$, while the leader waits to read the string $\#_1 \dots \#_k$.
This constitutes a proper clique.
The formal construction and proof are omitted as they are quite similar to the above case.

%% file: content/appendix_BSR_States_Threads.tex
\section{Proofs for Section \ref{Section:BSR_States_Threads}}
\label{Appendix_BSR_States_Threads}

We give the missing constructions and proofs for Section \ref{Section:BSR_States_Threads}.

\subsection*{Formal Construction and Proof of Proposition \ref{Proposition:BSR_States_Threads_Upper}}

The idea here is to reduce $\BSR$ to the reachability problem on an NFA $N$ of size at most $\bigO( \sizeP^{\nrt} \cdot \poly (k,n,\sizeM) ) $. The states of the NFA $N$ are the product of the states of each $P_i$ along with the current stage number, currently active process and the last value written to the memory i.e. it is of the form $Q_N = Q_1 \times \dots Q_t \times [0..\nrt] \times [0..k] \times \Domain$.  The currently active process records the information about who is allowed to write in that stage.
The initial state of $N$ is given by $q^0_N = (q^0_1, \dots q^0_{\nrt},0,0,\initmem)$. 

From any state of the form $(q_1,\dots,q_n,i,j,a), i\in [0..\nrt], j\in [0..k]$, for all moves of the form $q_{\ell} \moves{a?/\epsilon}_{P_{\ell}} q'_{\ell}$ for some $\ell \in [1..\nrt]$, we have a corresponding move of the form $(q_1,\dots,q_n,i,j,a) \moves{}_N (q_1,\dots,q_{\ell-1},q'_{\ell},\dots,q_n,i,j,a) $, this freely simulates any read move within  a stage. 

Similarly from any state of the form $(q_1,\dots,q_n,i,j,a), a \in \Sigma, i\in [1..\nrt], j\in [1..k]$,  for all moves of the form $q_{i} \moves{b!}_{P_{i}} q'_{i}$, we have a corresponding move of the form $(q_1,\dots,q_n,i,j,a) \moves{}_N (q_1,\dots,q_{i-1},q'_{i},\dots,q_n,i,j,b) $. These set of transitions allow the currently active process $P_i$ during any stage to write values to shared memory. 

Finally from any state of the form $(q_1,\dots,q_n,i,j,a), i\in [0..\nrt], j\in [0..k-1]$  we have moves of the form $(q_1,\dots,q_n,i,j,a) \moves{}_N (q_1,\dots,q_n,m,j+1,a) $, for all $m \in [1..\nrt]$. The correctness of such a construction is guaranteed by the following easy to see lemma. The complexity follows from the fact that reachability is quadratic and size of the automata that we construct is  at-most $\bigO( \sizeP^{\nrt} \cdot \poly (k,n,\sizeM) ) $.

\begin{lemma}
	There is an $\ell$-stage computation $c^0 \rightarrow^*_\asms c$  if and only if there is a computation of the form $q^0_N \rightarrow^*_N q_N$ for some $q_N \in \{ (q_1, \dots,q_n,i,\ell,a) \mid i \in [0..\nrt],  a \in \Domain, (q_1, \dots,q_n) = c \}$.
\end{lemma}

\begin{proof}
	($\Rightarrow$) For this direction, we will assume an $\ell$-stage computation of the form $c^0 \rightarrow^*_\asms c$  and show that there is a computation in $N$ of the form $q^0_N \rightarrow^*_N q_N$.  Since $c^0_\asms \rightarrow^*_\asms c_\asms$  is an $\ell$-stage computation, it can be split as $c_0 \rightarrow^*_\asms c_1 \rightarrow^*_\asms c_2 \dots \rightarrow^*_\asms c_{\ell + 1}$, where $c_0 = c^0, c_{\ell+1} = c$ and each of $\pi_i = c_i \moves{}_{\asms} c_{i+1}$ is an $1$-stage computation. Further for $i \in [1..\ell]$, let $p_i$ be the writer corresponding to the stage $\pi_i$.   It is easy to see that corresponding to every move  in the subcomputation of the form $(q_1,\dots,q_{\nrt},a) \rightarrow_{\asms} (q'_1,\dots,q'_{\nrt},b)$, there is a move of the form $(q_1,\dots,q_{\nrt},p_i,j,a) \moves{}_{\asms} (q'_1,\dots,q'_{\nrt},p_i,j,b)$ for all $j \in [1..k]$. Given a configuration $c = (q_1,\dots,q_{\nrt},a)$, we let $\mu(c,p,j) = (q_1,\dots,q_{\nrt},p,j,a)$. Given a computation $\pi$, we let $\mu(\pi,x,y)$ to be the sequence obtained by replacing each configuration $c$ occurring in $\pi$ by $\mu(c,x,y)$. It is easy to see that $\mu(\pi_0,0,0)$, $\mu(\pi_1,p_1,1)$, \dots, $\mu(\pi_{\ell},p_\ell,\ell)$ are all valid sub-computations in $N$. This is because we have for every move in the program, a corresponding move in the $\fingraph$ we construct. Combining these sub-computations will give us the required run in $N$. For combining these sub-computations, we use the transition of the form $(q_1,\dots,q_n,i,j,a) \moves{}_N (q_1,\dots,q_n,m,j+1,a) $, for all $m \in [1..\nrt]$ that was finally added in the construction.
	
	($\Leftarrow$)
	For this direction, we will assume a computation of the form $\pi = q^0_N \rightarrow^*_N q_N$, where $q_N = (q_1, \dots,q_n,i,\ell,a)$ for some $i \in [1..\nrt], a \in \Sigma$.
	Clearly such a computation can be split as follows
	\begin{align*}
		\pi = (q^0_1,\dots,q^0_{\nrt},0,0,r) &\rightarrow_N (q^1_1,\dots,q^1_{\nrt},0,0,a_1) \\
		\dots &\rightarrow_N (q^1_1,\dots,q^1_{\nrt},i_1,1,a_1) \\
		\dots &\rightarrow_N (q_1^2,\dots,q^2_{\nrt},i_1,1,a_2) \\
		\dots &\rightarrow_N (q^2_1,\dots,q^2_{\nrt},i_2,2,a_2) \\
		\dots &\rightarrow_N (q^{\ell+1}_1,\dots,q^{\ell+1}_{\nrt},i_{\ell},\ell,a_{\ell+1}).
	\end{align*}
	 Define $\pi_0 = (q^0_1,\dots,q^0_{\nrt},0,0,r) \rightarrow_N (q^1_1,\dots,q^1_{\nrt},0,0,a_1)$ and computation $\pi_j = (q^j_1,\dots,q^j_{\nrt},i_j,j,a_j) \rightarrow_N (q_1^{j+1},\dots,q^{j+1}_{\nrt},i_j,j,a_{j+1})$.

	It is easy to see that corresponding to every move of the form
	\begin{align*}
		(q_1, \dots, q_n,i,j,a) \moves{}_N (q'_1, \dots,q'_n,i,j,a'),
	\end{align*} 
	there is a move of the form $(q_1, \dots,q_n,a) \moves{} (q'_1, \dots,q'_n,a')$ such that it is either an read or internal move of some process or a write move of process $P_i$. Notice that the transition in $\fingraph$ was added because of existence of one such move in the program. Using this fact, it is easy to see that corresponding to each $\pi_i$, $i \in [0..\ell]$, there is a $1$-stage sub-computation in $\asms$. Concatenating these sub-computations will now give us the required run. \qed
\end{proof}

\subsection*{Formal Construction and Proof of Proposition \ref{Proposition:BSR_States_Threads_Lower}}

For the formal construction, we define for $i \in [1..k]$ the NFA $P_i$ to be the tuple $P_i = (\Ops{\Domain}, Q_i, q^0_i, \delta_i)$ with \mbox{$Q_i = \Setcon{q_{ij}}{ j \in [1..k]} \cup \Setcon{q^\ell_{ij}}{ j,\ell \in [1..k] } \cup \setcon{q^0_{i}}$}.
The states $q_{ij}$ are used to store one of the $k$ vertices of the $i$-th row, the states $q^\ell_{ij}$ are needed to check the edge relations to other rows and to perform the equality check.
The transition relation $\delta_i$ is given by the following rules:
For choosing and storing a vertex of the $i$-th row we have $q^0_i \xrightarrow{? \initmem}_i q_{ij}$ for any $j \in [1..k]$.
For checking the edge relations to vertices from a different row, we have for $\ell, j \in [1..k]$ the rule $q^{\ell - 1}_{ij} \xrightarrow{?(\ell, m)}_i q^{\ell}_{ij}$ if $\ell \neq i$ and if there is an edge between $(i,j)$ and $(\ell,m)$ in $G$.
Note that we identify $q_{ij}$ as $q^0_{ij}$.
Finally, to test the equivalence in the case of the same row we have $q^{i-1}_{ij} \xrightarrow{?(i,j)}_i q^{i}_{ij}$ for $j \in [1..k]$.

The writer $P_{ch}$ is given by $P_{ch} = (\Ops{\Domain}, Q, q_0, \delta)$, where $Q = \setcon{q_0, \dots, q_k}$ and $\delta$ is given by the rules $q_{i-1} \xrightarrow{!(i,j)}_{ch} q_{i}$ for $i,j \in [1..k]$.
Altogether, we define the program $\asms$ to be $\asms = (\Domain, \initmem, ((P_i)_{i \in [1..k]}, P_{ch}))$ and the set of configurations we want to reach as $C^f = \Setcon{(q^k_{1 j_1}, \dots, q^k_{k j_k}, q_k, a)}{j_1, \dots, j_k \in [1..k], a \in \Domain}$.
The correctness of the construction follows by the lemma below.

\begin{lemma}
	There is a $1$-stage computation $c^0 \rightarrow^*_\asms c$ for a $c \in C^f$ if and only if there is a clique of size $k$ in $G$ with one vertex from each row.
\end{lemma}
\begin{proof}
	First note that any computation in $\asms$ is a $1$-stage computation since the only thread that can write to the memory is $P_{ch}$.
	
	Let a clique of size $k$ in $G$ be given.
	It consists of the vertices $(1,j_1), \dots, (k, j_k)$.
	Then it is easy to construct a computation of $\asms$ starting in $(q^0_1, \dots, q^0_k, q_0, \gamma)$ and ending in $(q^k_{1 j_1}, \dots, q^k_{k j_k}, q_k, (k, j_k)) \in C^f$:
	The system just guesses the right vertices $(1,j_1)$ up to $(k,j_k)$ and then performs the edge-tests which are all positive since the vertices form a clique.
	
	For the other direction, let a computation $\rho$ leading to $(q^k_{1 j_1}, \dots, q^k_{k j_k}, q_0, (k, j_k))$ be given.
	Then we show that the vertices $(1, j_1), \dots, (k, j_k)$ form a clique.
	Since the $P_i$ can only start their computations on the initial memory symbol $\initmem$, they have to perform a step before $P_{ch}$ changes the memory content.
	Thus, we can split $\rho$ into $\rho = \rho_1 . \rho_2$ where $\rho_1$ contains only moves of the $P_i$ on reading $\initmem$, in any order.
	We may assume the following form for $\rho_1$:
	\begin{align*}
	(q^0_1, \dots, q^0_k, q_0, \initmem) &\xrightarrow{?\initmem}_1 (q^0_{1 j_1}, q^0_2, \dots, q^0_k, q_0, \initmem) \\
	\dots &\xrightarrow{?\initmem}_k (q^0_{1 j_1}, \dots, q^0_{k j_k}, q_0, \initmem).
	\end{align*}
	Note that the computation $\rho_1$ corresponds to the choice of $(1, j_1), \dots, (k, j_k)$ as a clique-candidate.
	
	After $\rho_1$, the thread $P_{ch}$ needs to write the symbol $(1, j_1)$ to the memory since otherwise, $P_1$ would deadlock.
	But then each $P_i$ with $i \neq 1$ needs to do a step on reading $(1,j_1)$ and this only happens if $(i,j_i)$ and $(1,j_1)$ share an edge.
	Then the computation $\rho_2$ goes on with $P_{ch}$ writing $(2,j_2)$ since otherwise, $P_2$ would deadlock.
	The other threads again perform a verification step on reading $(2,j_2)$.
	Since the computation reaches the configuration $(q^k_{1 j_1}, \dots, q^k_{k j_k}, q_0, (k, j_k))$ in the end, we can ensure that all chosen vertices indeed share an edge.
	\qed
\end{proof}

\subsection*{Formal Construction and Proof of Theorem \ref{Theorem:BSR_States_Threads_Kernel_Lower}}

Let $\varphi_1, \dots, \varphi_I$ be given $\kSAT{3}$-instances, each two equivalent under $\polyrel$.
We assume that each $\varphi_\ell$ has the form:
$\varphi_\ell = C^\ell_1 \wedge \dots \wedge C^\ell_m$ and the set of variables used by the $\kSAT{3}$-instances is $\setcon{x_1, \dots, x_n}$.
We define the data domain to be $\Domain = [1..I] \times [1..m] \times [1..n] \times \setcon{0,1} \cup \setcon{\initmem}$.

For each $i \in [1..n]$, we introduce a thread $P_{x_i} = (\Ops{\Domain}, Q_{x_i}, q^0_{x_i}, \delta_{x_i})$, where $Q_{x_i} = \Setcon{q^j_{(x_i,0)}, q^j_{(x_i,1)}}{j \in [0..m]} \cup \setcon{q^0_{x_i}}$.
Each state $q^j_{(x_i,v)}$ stores the chosen evaluation $v$ of $x_i$ and the number $j$ of clauses that were already satisfied. 
We have $\abs{Q_{x_i}} = 2(m+1) + 1$.
The transition relation $\delta_{x_i}$ contains the following rules:
For choosing an evaluation of $x_i$ we have $q^0_{x_i} \xrightarrow{?\initmem}_{x_i} q^0_{(x_i, v)}$ for $v \in \setcon{0,1}$.
For checking the satisfiability of clauses we have for $\ell \in [1..I], j \in [1..m]$, and $v \in \setcon{0,1}$:
$q^{j-1}_{(x_i,v)} \xrightarrow{?(\ell, j, i, v)}_{x_i} q^j_{(x_i,v)}$ if clause $C^\ell_j$ is satisfied by variable $x_i$ under evaluation $v$.
If the requested variable is not $x_i$, the thread $P_{x_i}$ just performs a counting step: $q^{j-1}_{(x_i,v)} \xrightarrow{?(\ell, j, i', v')}_{x_i} q^j_{(x_i,v)}$ for $i' \neq i$, $\ell \in [1..I]$, $j \in [1..m]$, and $v' \in \setcon{0,1}$.

Next, we introduce the writer-thread $P_w = (\Ops{\Sigma}, Q_w, q^0_w, \delta_w)$ with set of states $Q_w = \setcon{q^0_w, \dots, q^m_w}$. 
Thus, we have $\abs{Q_w} = m+1$.
The writer picks $m$ clauses of probably different instances that need to be satisfied.
To this end, it will not only guess the clause but also the instance that contains the clause, the variable that should satisfy it, and the evaluation of the variable.
This will then be discarded or verified by the variable-threads.
The transitions that we need in $\delta_w$ are thus of the form $q^{j-1}_w \xrightarrow{!(\ell, j, i, v)}_w q^j_w$ for any $\ell \in [1..I], j \in [1..m], i \in [1..n]$, and $v \in \setcon{0,1}$.
Writing $(\ell, j, i, v)$ to the memory reflects the claim that clause $C^\ell_j$ gets satisfied by variable $x_i$ under evaluation $v$.

The last type of threads that we introduce are the bit-checkers.
For each \mbox{$b \in [1..\log(I)]$} we define the thread $P_b = (\Ops{\Sigma}, Q_b, q^0_b, \delta_b)$ with set of states \mbox{$Q_b = \Setcon{q^j_{(b,0)}, q^j_{(b,1)}}{j \in [1..m]} \cup \setcon{q^0_b}$}.
Hence, we have that $\abs{Q_b} = 2m + 1$.
The task of bit-checker $P_b$ is to verify that along the instances $\varphi_{\ell_1}, \dots, \varphi_{\ell_m}$ guessed by the writer, the $b$-th bit of $\bin(\ell_j)$ for $j \in [1..m]$ does not change.
To this end, we construct $\delta_b$ the following way:
Initially, $P_b$ stores the $b$-th bit of the first guessed instance.
We add for $\ell \in [1..I], i \in [1..n]$, and $v \in \setcon{0,1}$ the rule $q^0_b \xrightarrow{?(\ell, 1, i, v)}_b q^1_{(b,u)}$ if the $b$-th bit of $\bin(\ell)$ is $u \in \setcon{0,1}$.
For the comparison with further guessed instances we have for $j \in [2..m], \ell \in [1..I], i \in [1..n]$, and $v,u \in \setcon{0,1}$ the rule $q^{j-1}_{(b,u)} \xrightarrow{?(\ell, j, i, v)}_b q^j_{(b,u)}$ if the $b$-th bit of $\bin(\ell)$ is $u$.

We define the program $\asms$ as $\asms = (\Domain, \initmem, ( P_w, (P_{x_i})_{i \in [1..n]}, (P_b)_{b \in [1..\log(I)]} ) )$.
We want the writer and the bit-checkers to perform exactly $m$ steps and the variable-threads to move exactly $m+1$ steps.
Intuitively, this amounts to the satisfiability of $m$ clauses that all belong to the same instance.
Hence, the set of configurations $C^f$ that we want to reach is the following:
$$
\Setcon{( q^m_w, q^m_{(x_1, v_1)}, \dots, q^m_{(x_n, v_n)}, q^m_{(1, u_1)}, \dots, q^m_{(\log(I), u_{\log(I)})}, a)}{ v_i, u_\ell \in \setcon{0,1}, a \in \Sigma }.
$$

Since $P_w$ is the only thread which is allowed to write, we are interested in reaching $C^f$ within a $1$-stage computation.
Hence, the $\BSR$-instance of interest is the tuple $(\asms, C^f, 1)$.
Note that the parameters obey the bounds of a cross-composition:
$\sizeP = 2(m+1) + 1$ and $\nrt = 1 + n + \log(I)$.
It is thus left to show that the above construction is correct.
This is proven in the following lemma.

\begin{lemma}
	There is a $1$-stage computation from $c^0 \rightarrow^*_\asms c$ for a $c \in C^f$ if and only if there is an $\ell \in [1..I]$ such that $\varphi_\ell$ is satisfiable.
\end{lemma}
\begin{proof}
	First assume that there is an $\ell \in [1..I]$ such that $\varphi_\ell$ is satisfiable.
	Let $v_1, \dots, v_n$ be the evaluation of the variables $x_1, \dots, x_n$ that satisfies $\varphi_\ell$ and let \mbox{$\bin(\ell) = u_1 \dots u_{\log(I)}$} be the binary representation of $\ell$.
	We construct a $1$-stage computation of $\asms$ from $c^0$ to the configuration
	$$
	c = (q^m_w, q^m_{(x_1, v_1)}, \dots, q^m_{(x_n, v_n)}, q^m_{(1, u_1)}, \dots, q^m_{(\log(I), u_{\log(I)})}, (\ell, m, x_z, v_z)),
	$$
	where $x_z$ is a variable in $C^\ell_m$.
	
	The computation starts with choosing the right evaluation for the variables:
	Each $P_{x_i}$ performs the move $q^0_{x_i} \xrightarrow{?\initmem}_{x_i} q^0_{(x_i,v_i)}$.
	This leads to the computation $c^0 \rightarrow^*_\asms (q^0_w, q^0_{(x_1,v_1)}, \dots, q^0_{(x_n,v_n)}, q^0_1, \dots, q^0_{\log(I)}, \gamma) = c_0$.
	Then $P_w$ writes the tuple $(\ell, 1, i', v_{i'})$ to the memory, where $x_{i'}$ is a variable in clause $C^\ell_1$.
	Furthermore, $x_{i'}$ evaluated to $v_{i'}$ satisfies the clause.
	This is read by all $P_{x_i}$ and each performs the step $q^0_{(x_i, v_i)} \xrightarrow{?(\ell, 1, i', v_{i'})}_{x_i} q^1_{(x_i,v_i)}$.
	Note that due to the definition of $\delta_{x_i}$, in both cases, $i \neq i'$ and $i = i'$, the move can indeed be done.
	The bit-checkers $P_b$ also read the tuple.
	Each $P_b$ does the following step: $q^0_b \xrightarrow{?(\ell, 1, i', v_{i'})}_b q^1_{(b, u_b)}$ since $u_b$ is the $b$-th bit of $\bin(\ell)$.
	If we put the individual moves together, this leads to a new configuration 
	\begin{align*}
		c_0 \rightarrow^*_\asms (q^1_w, q^1_{(x_1, v_1)}, \dots, q^1_{(x_n, v_n)}, q^1_{(1, u_1)}, \dots, q^1_{(\log(I),u_{\log(I)})}, (\ell, 1, i', v_{i'})) = c_1.
	\end{align*}
	
	Similarly, we can construct a computation that leads to a configuration $c_2$.
	If we go on with the construction, we get a computation that leads to $c$.
	
	For the other direction, let $\rho$ be a $1$-stage computation of $\asms$, ending in the configuration $(q^m_w, q^m_{(x_1,v_1)}, \dots, q^m_{(x_n, v_n)}, q^m_{(1,u_1)}, \dots, q^m_{(\log(I), u_{\log(I)})}, a) = c$.
	Let \mbox{$\ell \in [1..I]$} be such that $\bin(\ell) = u_1 \dots u_{\log(I)}$.
	We show that $\varphi_\ell$ is satisfiable.
	More precise, $\varphi_\ell$ is satisfied under $x_i$ evaluating to $v_i$.
	
	Since $P_{x_i}$ can start its computation only on reading $\initmem$, we get that each $P_{x_i}$ performs the step $q^0_{x_i} \xrightarrow{?\initmem}_{x_i} q^0_{(x_i,v_i)}$.
	Note that the chosen $v_i$ is indeed the one appearing in $c$.
	Hence, we get as an initial part of $\rho$ the computation
	\begin{align*}
		\rho^0 = c^0 \rightarrow_\asms (q^0_w, q^0_{(x_1,v_1)}, \dots, q^0_{(x_n, v_n)}, q^0_1, \dots, q^0_{\log(I)}, \initmem) = c_0.
	\end{align*}
	
	The computation can then only continue if $P_w$ changes the content of the memory.
	The thread guesses and writes $(\ell, 1, i', v_{i'})$ to the memory.
	He has to choose the index $\ell$, and thus the instance $\varphi_\ell$, since otherwise there would be a bit-checker $P_b$ not reaching $q^m_{(b,u_b)}$.
	The bit-checkers $P_b$ read $(\ell, 1, i', v_{i'})$ and store $u_b$ in their states: 
	$q^0_b \xrightarrow{?(\ell, 1, i', v_{i'})}_b q^1_{(b,u_b)}$.
	Now all $P_{x_i}$ have to perform a step since the computation does not deadlock.
	This means that especially $P_{x_{i'}}$ performs a step on reading $(\ell, 1, i', v_{i'})$.
	By definition, this is only possible if $x_{i'}$ is a variable that evaluated to $v_{i'}$ satisfies clause $C^\ell_1$.
	Hence, we have that under evaluating each $x_i$ to $v_i$, clause $C^\ell_1$ is satisfied.
	If we combine all the moves done, we get another part $\rho_1$ of the computation $\rho$ that leads from $c_0$ to the configuration $(q^1_w, q^1_{(x_1,v_1)}, \dots, q^1_{(x_n, v_n)}, q^1_{(1,u_1)}, \dots, q^1_{(\log(I), u_{\log(I)})}, (\ell,1, i', v_{i'})) = c_1$.
	
	Similarly, one proves that the second part of $\rho$, the computation $\rho^2$ shows that $C^\ell_2$ is satisfiable under evaluating each $x_i$ to $v_i$.
	Hence, by induction we get that $\varphi_\ell$ is satisfiable.
	\qed
\end{proof}

%% file: content/appendix_BSR_Intractability.tex
\section{Proofs for Section \ref{Section:BSR_Intractability}}
\label{Appendix_BSR_Intractability}

We give the missing constructions and proofs for Section \ref{Section:BSR_Intractability}.

\subsection*{Proof of Proposition \ref{Proposition:BSR_Intractability}}

We show that a memory domain of constant size and a single stage suffice to reduce $\kSAT{3}$ to $\BSR$.
Let $\varphi = C_1 \wedge \dots \wedge C_m$ be a formula in CNF with at most three literals per clause $C_i$.
Let the variables of $\varphi$ be $x_1, \dots, x_n$.
Our goal is to construct a program $\asms = (\Domain, \initmem, (P_1,\dots,P_n,P_v))$ such that $c^0$ can reach an unsafe configuration of $\asms$ in a single stage if and only if $\varphi$ has a satisfying assignment.
Moreover, the size of the domain in the construction will be $\sizeM = 4$ and is thus constant.

We set $\Domain = \setcon{\initmem, \#, 0, 1}$.
For communication over this domain, we encode literals of $\varphi$ in binary.
Let $\bin_\#(i) \in (\setcon{0,1} . \#)^{\log(n)+1}$ be the binary encoding of $i$ into $n$ bits where each bit is separated by the symbol $\#$.
For instance, we get $\bin_\#(2) = 0 \# 0 \# 1 \# 0 \#$ in the case $n = 8$.
Given a literal $\ell$ of $\varphi$, we encode it by $\enc(\ell) = v \# \bin_\#(i)$, where $x_i$ is the variable in $\ell$ and $v$ its evaluation.

We have a separate thread $P_v$, called the verifier.
It iterates over the clauses and for each clause $C_i = \ell_{i_1} \vee \ell_{i_2} \vee \ell_{i_3}$, the thread picks a literal and writes $\enc(\ell_{i_1})$, $\enc(\ell_{i_2})$, or $\enc(\ell_{i_3})$ to the shared memory.
To this end, it has states $\setcon{q_1, \dots, q_m}$ and sequences of transitions 
\begin{align*}
	q_i \xrightarrow{!\enc(\ell_{i_j})} q_{i+1} ~\text{for}~ j =1,2,3 ~\text{and}~ i \in [1..m-1].
\end{align*}
The notation $!\enc(\ell_{i_j})$ indicates that the whole encoding of $\ell_{i_j}$ is written to the shared memory.
This can be easily achieved by adding $\log(n)+1$ many intermediary states.
Hence, $P_v$ has $\bigO(m \cdot \log(n))$ many states in total and writes the encoding of exactly $m$ literals to the shared memory.

For each variable $x_i$, we have a thread $P_i$.
Initially, on reading $\initmem$, the thread $P_i$ chooses the evaluation for variable $x_i$.
It stores the evaluation.
To this end, the thread has states $\Setcon{p^i_{(v,j)}}{v = 0,1 ~\text{and}~ j \in [1..m]}$.
The $m$ copies are needed to count the number of literal encodings that were written to the memory by the verifier.

For each literal $\ell \notin \setcon{x_i, \neg x_i}$, thread $P_i$ has sequences of transitions
\begin{align*}
	p^i_{(v,j)} \xrightarrow{?\enc(\ell)} p^i_{(v,j+1)}, \, j \in [1..m-1].
\end{align*}
Since $\ell$ contains a different variable than $x_i$, the thread $P_i$ does not need to check whether the evaluation in $\ell$ matches the stored one.
These transitions only ensure that $P_i$ can keep track of the number of encodings that was already written by the verifier.
Note that, as above, the sequences can be realized by adding intermediary states.

For literals containing $x_i$, the thread $P_i$ needs to check whether the evaluation of the literal matches the stored evaluation.
This can be realized as follows.
If $P_i$ decides to store evaluation $v$ for $x_i$, then only encodings of the form $v \# \bin_\#(i)$ can be read.
Hence, $P_i$ has the transition sequences
\begin{align*}
	p^i_{(v,j)} \xrightarrow{? v \# \bin_\#(i)} p^i_{(v,j+1)}, \, j \in [1..m-1].
\end{align*}
Note that, if the verifier $P_v$ writes a literal $\ell$ to the memory which contains $x_i$ but has the wrong evaluation, $P_i$ is not able to read the encoding of $\ell$ and deadlocks.
Moreover, note the importance of the symbol $\#$.
It avoids repeated reading of the same symbol which can cause false encodings.

By construction, we get that $\varphi$ has a satisfying assignment if and only if all threads reach their last state.
If $\varphi$ has a satisfying assignment, the threads $P_i$ choose exactly this assignment and store it.
Now the verifier $P_v$ chooses for each clause a literal that satisfies it and the $P_i$ can read the encodings of these literals and terminate.
For the other direction, assume that all $P_i$ and $P_v$ reach their last state.
Since there is no loss in the communication between the threads, the assignment chosen by the $P_i$ is satisfying for $\varphi$.
This is due to that all encodings of literals chosen by $P_v$ can be read without a $P_i$ getting stuck.

Since $P_v$ is the only thread that writes to the shared memory, the computation has only one stage.